\documentclass[
aps,%
pre,%
a4paper,%
final,%
reprint,%
showpacs,%
superscriptaddress,%
groupedaddress,%
nofootinbib
]{revtex4-1}

\usepackage{amsmath}
\usepackage{mathrsfs}  
\usepackage{amssymb}
\usepackage{amsthm}
\usepackage{bm}
\usepackage{layouts}

\usepackage[pdftex]{graphicx}
\graphicspath{{./figs/}}
\usepackage[pdftex,colorlinks=true,bookmarks=true,citecolor=blue,urlcolor=blue]{hyperref} 
\usepackage[caption=false]{subfig}

\newcommand{\field}[1]{\mathbb{#1}}
\newcommand{\fs}[1]{\mathsf{#1}}

\DeclareMathOperator*{\supp}{supp}
\DeclareMathOperator*{\esssup}{ess\,sup}
\DeclareMathOperator{\Res}{Res}
\DeclareMathOperator{\diag}{diag}
\DeclareMathOperator{\Wrons}{\mathscr{W}}
\newcommand{\norm}[1]{\left\lVert #1 \right\rVert}

\newcommand{\tp}{\intercal}

\newcommand{\ovl}[1]{\overline{#1}}

\newcommand{\bigO}[1]{\mathop{\mathscr{O}}\left(#1\right)}

\let\Re\relax
\DeclareMathOperator{\Re}{Re}
\let\Im\relax
\DeclareMathOperator{\Im}{Im}
\newcommand{\vv}[1]{\boldsymbol{#1}}
\newcommand{\vs}[1]{\boldsymbol{#1}}
\newcommand{\gf}[1]{\boldsymbol{\mathfrak{#1}}}

\newcommand{\OP}[1]{\mathscr{#1}}

\DeclareMathOperator{\sech}{sech}

\newcommand{\wtilde}[1]{\widetilde{#1}}

\newcommand{\et}{\textit{et~al.}}

\newtheorem{theorem}{Theorem}[section]
\newtheorem{prop}[theorem]{Proposition}
\newtheorem{corr}[theorem]{Corollary}
\newtheorem{lemma}[theorem]{Lemma}
\newtheorem{defn}{Definition}[section]

\newtheorem{rem}{Remark}[section]
\newtheorem{problem}{Problem}[section]
\usepackage{enumitem}

\begin{document}

\title{Fast Inverse Nonlinear Fourier Transformation using Exponential One-Step
Methods, Part I: Darboux Transformation}


\author{V.~Vaibhav}
\email{v.k.vaibhav@tudelft.nl}

\affiliation{Delft Center for Systems and Control, 
Delft University of Technology, Mekelweg 2. 2628 CD Delft, 
The Netherlands}
\date{\today}

\begin{abstract}
This paper considers the non-Hermitian Zakharov-Shabat (ZS) scattering problem
which forms the basis for defining the SU$(2)$-nonlinear Fourier transformation
(NFT). The theoretical underpinnings of this generalization of the conventional Fourier 
transformation is quite well established in the Ablowitz-Kaup-Newell-Segur (AKNS)
formalism; however, efficient numerical algorithms that could be employed in
practical applications are still unavailable.
    
In this paper, we present a unified 
framework for the forward and inverse NFT using exponential one-step methods
which are amenable to FFT-based fast polynomial arithmetic. Within this discrete
framework, we propose a fast Darboux transformation (FDT) algorithm having 
an operational complexity of $\bigO{KN+N\log^2N}$ such that the error in
the computed $N$-samples of the $K$-soliton vanishes as $\bigO{N^{-p}}$ where $p$ 
is the order of convergence of the underlying one-step method. For fixed $N$, 
this algorithm outperforms the the classical DT (CDT) algorithm which has a complexity of 
$\bigO{K^2N}$. We further present extension of these algorithms to the 
general version of DT which allows one to add solitons to arbitrary profiles 
that are admissible as scattering potentials in the ZS-problem. The general CDT/FDT 
algorithms have the same operational complexity
as that of the $K$-soliton case and the order of convergence matches that of 
the underlying one-step method. A comparative study of these algorithms is 
presented through exhaustive numerical tests.
\end{abstract}

\pacs{%
02.30.Zz,
02.30.Ik,
42.81.Dp,
03.65.Nk
}

\maketitle


\section*{Notations}
\label{sec:notations}
The set of real numbers (integers) is denoted by $\field{R}$ ($\field{Z}$) and 
the set of non-zero positive real numbers (integers) by $\field{R}_+$ 
($\field{Z}_+$). The set of complex numbers are denoted by $\field{C}$,
and, for $\zeta\in\field{C}$, $\Re(\zeta)$ and $\Im(\zeta)$ refer to the real
and the imaginary parts of $\zeta$, respectively. The complex conjugate of 
$\zeta\in\field{C}$ is denoted by $\zeta^*$ and $\sqrt{\zeta}$ denotes its
square root with a positive real part. The upper-half (lower-half) of $\field{C}$ 
is denoted by $\field{C}_+$ ($\field{C}_-$) and it closure by $\ovl{\field{C}}_+$
($\ovl{\field{C}}_-$). The set $\field{D}=\{z|\,z\in\field{C},\,|z|<1\}$
denotes an open unit disk in $\field{C}$ and $\ovl{\field{D}}$ denotes its 
closure. The set $\field{T}=\{z|\,z\in\field{C},\,|z|=1\}$ denotes the unit 
circle in $\field{C}$. The Pauli's spin matrices are denoted by, 
$\sigma_j,\,j=1,2,3$, which are defined as
\[
\sigma_1=\begin{pmatrix}
0 &  1\\
1 &  0
\end{pmatrix},\quad
\sigma_2=\begin{pmatrix}
0 &  -i\\
i &  0
\end{pmatrix},\quad
\sigma_3=\begin{pmatrix}
1 &  0\\
0 & -1
\end{pmatrix},
\]
where $i=\sqrt{-1}$. For uniformity of notations, we denote
$\sigma_0=\text{diag}(1,1)$. Matrix transposition is denoted by $(\cdot)^\tp$
and $I$ denotes the identity matrix. For any two vectors 
$\vv{u},\vv{v}\in\field{C}^2$, $\Wrons(\vv{u},\vv{v})\equiv (u_1v_2-u_2v_1)$ 
denotes the Wronskian of the two vectors and $[A,B]$ stands for the commutator 
of two matrices $A$ and $B$. Partial derivatives with respect to $x$ are denoted by
$\partial_x$ or $(\cdot)_x$ while repeated derivatives by $\partial^2_{x}$. 
The support of a function
$f:\Omega\rightarrow\field{R}$ in $\Omega$ is defined as $\supp
f=\ovl{\{x\in\Omega|\,f(x)\neq0\}}$. The Lebesgue spaces of complex-valued 
functions defined in $\field{R}$ are denoted by 
$\fs{L}^p$ for $1\leq p\leq\infty$ with their corresponding 
norm denoted by $\|\cdot\|_{\fs{L}^p}$ or $\|\cdot\|_p$. 

The inverse 
Fourier-Laplace transform of a function $F(\zeta)$ analytic in 
$\ovl{\field{C}}_+$ is defined as
\[
\vv{f}(\tau)=\frac{1}{2\pi}\int_{\Gamma}F(\zeta)e^{-i\zeta\tau}\,d\zeta,
\]
where $\Gamma$ is any contour parallel to the real line.

\section{Introduction}
This paper considers the two-component non-Hermitian scattering problem 
first studied by Zakharov and Shabat (ZS)~\cite{ZS1972}, which forms the basis
for defining the SU$(2)$-nonlinear Fourier transformation (NFT). 
For certain integrable nonlinear equations whose general description is provided by 
the AKNS-formalism~\cite{AKNS1974,AS1981}, the NFT offers a powerful means of solving 
the corresponding initial-value problem (IVP). One such example is the nonlinear
Schr\"odinger equation (NSE) that is commonly used to model channels for optical fiber
communication. The propagation of optical field in a loss-less single mode
fiber under Kerr-type focusing nonlinearity is governed by the 
NSE~\cite{HK1987,Agrawal2013} which can be cast into the following standard form
\begin{equation}\label{eq:NSE}
    i\partial_Zq=\partial^2_{T}q+2|q|^2q,\quad(T,Z)\in\field{R}\times\field{R}_+,
\end{equation}
where $q(T,Z)$ is a complex valued function associated with the slowly varying
envelope of the electric field, $Z\in\field{R}_+$ is the position along the 
fiber and $T$ is the retarded time. This equation also provides a satisfactory 
description of optical pulse propagation in the
guiding-center or path-averaged formulation~\cite{HK1990GC,HK1991GC,TBF2012} 
when more general scenarios such as
presence of fiber losses, lumped or distributed periodic amplification are
included in the mathematical model of the physical channel. The IVP
corresponding to~\eqref{eq:NSE} consists in finding the evolved field $q(T,Z)$ 
for a given initial condition $q(T,0)$ under vanishing boundary conditions. 
\begin{figure*}[t!]
\centering
\includegraphics[scale=1.0]{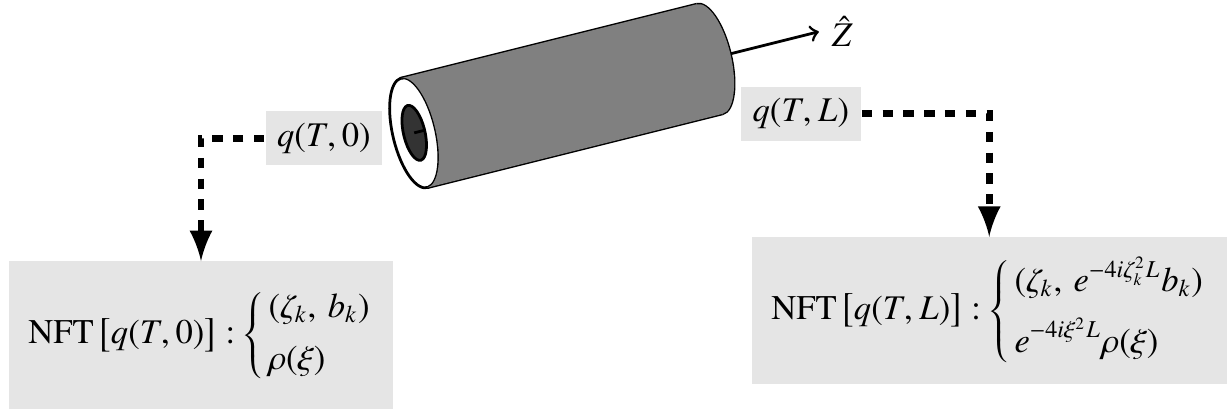}
\caption{\label{fig:evolution} The figure shows evolution of the nonlinear
Fourier spectrum along the length of the fiber. Here, the sequence
$(\zeta_k,b_k)$ denotes the
discrete spectrum and $\rho(\xi),\,\xi\in\field{R}$ is the continuous spectrum
also known as the reflection coefficient.}
\end{figure*}
For a given initial condition $q(T,0)$, the nonlinear Fourier spectrum consists 
of (i) a continuous part $\rho(\xi),\,\,\xi\in\field{R}$, and, (ii) a 
discrete part given by 
$\mathfrak{S}_K=\{(\zeta_k,b_k)\in\field{C}^2|\,\Im{\zeta_k}>0,\,k=1,2,\ldots,K\}$ 
which is an ordered pair of eigenvalues $\zeta_k$ and the respective norming 
constants $b_k$ (see~\cite{AKNS1974} or Sec.~\ref{sec:akns} for a complete introduction). 
The discrete spectrum is associated with the solitonic components of the 
potential which will be referred to as \emph{bound states}
in the rest of this paper. The energy in these states does not disperse away as in
the case of linear waves, a phenomenon which is adequately characterized by the term 
``bound states''. The evolution of the nonlinear Fourier (NF) spectrum depicted in
Fig.~\ref{fig:evolution} is reminiscent of evolution in a linear channel--a
property which is attributed to the integrablity of the nonlinear channel. 

In passing, we also note that the ZS-problem appears in various other 
physical systems, for instance, grating-assisted 
co-directional couplers (GACCs), a device used to couple 
light between two different guided modes of an optical fiber (see~\cite{FZ2000,BS2003} and
references therein), and, NMR spectroscopy where design of frequency selective pulses 
requires solution of a ZS-problem~\cite{RM1992-NMR,RM1992,RS1994}.

Amongst the key physical effects that affect the performance of an optical fiber
communication system, namely, chromatic dispersion, 
Kerr-type nonlinearity and optical noise, it is the latter two that have become the 
principle factors limiting the spectral efficiency of 
wavelength--division--multiplexed (WDM) networks at high signal powers. 
The reason behind this is largely the transmission methodologies that assume a
linear model of the channel. The NF spectrum, in contrast, 
offers a novel way of encoding information in optical pulses where the nonlinear
effects are adequately taken into account as opposed to being treated as a source
of signal distortion. The idea to use
discrete eigenvalues of the NF spectrum was first proposed by 
Hasegawa and Nyu~\cite{HN1993} which they termed as the \emph{eigenvalue
communication}. Recently, Yousefi and Kschischang~\cite{Yousefi2014compact}
have proposed nonlinear signal multiplexing in multi-user channels in order to
mitigate the problem of nonlinear cross-talk that occurs in WDM systems. 
We note that the most general modulation technique uses both the discrete as well as the
continuous part of the NF spectrum which was recently demonstrated
in~\cite{ALB2016}. We refer the reader to a comprehensive review 
article~\cite{TPLWFK2017} and the references therein for an overview of 
the progress in theoretical as well as experimental aspects of various 
NFT-based optical communication methodologies. It must be noted that practical
implementation of NFT-based transmission is still quite far from becoming a
reality~\cite{TPLWFK2017} and there are other potential ways to combat nonlinear signal distortions 
in optical fibers~\cite{CGKLY2017,DW2017}\footnote{At
this point, studies indicate that there is \emph{no clear winner} as far as
mitigation of impairments due to nonliearity in optical fibers is concerned; therefore, we continue our
efforts to improve NFT-based approach. Also noteworthy is the fact that
ZS-problem appears in various other systems of physical significance where this
research can find application.}.

In any NFT-based modulation technique, the importance of low-complexity 
NFT-algorithms cannot be over emphasized. In this paper, we focus on 
the development of fast algorithms for various modulation scenarios of a NFT-based
transmission system. As noted in~\cite{TPLWFK2017}, many of the existing
numerical approaches tend to become inaccurate as the signal power increases.
While this maybe attributed to lack of numerical precision, it could also be
numerical ill-conditioning or as a result of naive implementation. It is difficult to 
fully address these problems in this work, but let us remark that stability and 
convergence of the numerical algorithm plays a key role in determining its 
performance in realistic scenarios. We discuss these two aspects quite 
rigorously in this work.

Our primary goal here is to provide a theoretical foundation for the 
algorithms reported in~\cite{VW2017OFC} where we also showcased our preliminary 
results demonstrating the first fast inverse NFT. The specific problems 
for which we seek fast algorithms in this work are as follows: 
\begin{problem}[Generation of multi-solitons] 
\label{prob1}
Given an arbitrary discrete spectrum $\mathfrak{S}_K$ ($K$ being its
cardinality or, in other words, the number of bound states), compute the corresponding
multi-soliton potential.
\end{problem} 

\begin{problem}[Addition of bound states] 
\label{prob2}
Given an arbitrary potential $q_{\text{seed}}(x)$ referred to as the ``seed'' potential
(assumed to be admissible as a scattering potential in the ZS-problem) and
a given discrete spectrum $\mathfrak{S}_K$,  compute the 
``augmented'' potential such that its discrete spectrum is given by
$\mathfrak{S}_{\text{aug.}}=\mathfrak{S}_{\text{seed}}\cup\mathfrak{S}_K$ where
$\mathfrak{S}_K$ is known to be disjoint with $\mathfrak{S}_{\text{seed}}$, the
discrete spectrum of the seed potential.
\end{problem} 

\begin{problem}[Inversion of continuous spectrum] 
\label{prob3a}
Given an arbitrary continuous spectrum $\rho(\xi),\,\,\xi\in\field{R}$, such
that there exists a positive constant $C>0$ for which the estimate 
\[
|\rho(\xi)|\leq\frac{C}{1+|\xi|},
\]    
holds, compute the potential such that its continuous spectrum is 
$\rho(\xi)$ and the discrete spectrum is empty.
\end{problem} 

\begin{problem}[Inverse NFT] 
\label{prob3}
Given an arbitrary continuous spectrum $\rho(\xi),\,\,\xi\in\field{R}$,
satisfying the estimate in Prob.~\ref{prob3a} and a given discrete spectrum 
$\mathfrak{S}_K$, compute the potential such that its continuous spectrum is 
$\rho(\xi)$ and its discrete spectrum is $\mathfrak{S}_K$.
\end{problem} 

The first two of these problems can be solved, at least in principle, using the Darboux
transformations (DT)~\cite{Lin1990,GHZ2005}. The Prob.~\ref{prob1} can be solved 
with machine precision using DT with null-potential as the seed. The resulting 
complexity is $\bigO{K^2N}$ where $K$ is the number of eigenvalues and $N$ is 
the number of samples of the potential. The scenario in Prob.~\ref{prob1} corresponds to 
the modulation of the discrete NF spectrum which has been explored by a number of
groups~\cite{HKY2014,DHGZYLWKL2015,HK2016} and it has also been experimentally
demonstrated~\cite{TM2013,B2014,B2015,MTM2014,ABSI2015,BAI2016}.

The Prob.~\ref{prob2} cannot be solved without resorting to numerical methods for 
the ZS-problem because the so called Jost solutions (which are required in DT) are 
not known in a closed form for any arbitrary seed potential. 
Prob.~\ref{prob3a} and~\ref{prob3} will be treated in a sequel to this paper.

The numerical techniques for solving Prob.~\ref{prob1}--\ref{prob3} developed in
this work are based on exponential (linear) one-step methods~\cite{CM2002,Gautschi2012} 
for the discretization of the ZS-problem. The method yields a discrete framework
for solving the ZS-problem which resembles the transfer matrix approach for
solving wave-propagation problems in dielectric layered media~\cite[Chap.~1]{BW1999}.
These transfer matrices have polynomial entries--a form that is amenable to the FFT-based 
polynomial arithmetic~\cite{Henrici1993} and is also compatible with the
layer-peeling algorithm~\cite{BK1987}. All the methods considered in this
article exhibit either a first order or a second order of convergence, i.e., 
the numerical errors vanish as $\bigO{N^{-p}}$ where $p$ is the order of the
one-step method\footnote{The discrete system corresponding to 
Ablowitz-Ladik (AL) scattering problem~\cite[Chap.~3]{APT2004} is also 
amenable to FFT-based fast polynomial arithmetic and
satisfies the layer-peeling property~\cite{WP2013d}; however, it does not
illuminate on how to obtain a general recipe that could be applied to 
the ZS-problem in order to obtained a similar discrete system 
possessing a given order of convergence.}. 

Within this discrete framework, we develop two algorithms: 
(a) the \emph{classical Darboux transform} (CDT) which addresses Prob.~\ref{prob2}, 
and, (b) the \emph{fast Darboux transformation} (FDT) which addresses 
Prob.~\ref{prob1} and~\ref{prob2} both\footnote{It is worth noting that an 
alternative fast method of solving Prob.~\ref{prob1} is reported
in~\cite{WP2015} and it can be readily adapted to the discrete framework
considered in this work. However, this method offers no control over the 
norming constants, therefore, we do not address this algorithm here.}. The 
CDT algorithm is a direct numerical implementation of the DT in the continuum case
where the seed Jost solutions are computed by numerically solving the scattering
problem resulting in an overall complexity of $\bigO{K^2N}$. The FDT algorithm is 
entirely new and it is based on the pioneering work of Lubich on convolution
quadrature~\cite{Lubich1988I,Lubich1988II,Lubich1994}. In order to ensure compatibility
with Lubich's construction, we restricted ourselves to the implicit Euler method
and the trapezoidal rule. This algorithm has an operational complexity of 
$\bigO{N(K+\log^2N)}$ and an order of convergence that matches that of the
underlying one-step method, i.e., $\bigO{N^{-p}}$ where $p=1$ (implicit Euler), 
$p=2$ (trapezoidal rule). With increasing number of eigenvalues, FDT clearly
outperforms CDT. The numerical tests and error analysis of the numerical scheme
suggests that CDT is useful only for smaller number of eigenvalues. These tests 
further reveal that FDT is not only more accurate for the general case, it also has 
superior numerical conditioning with increasing number of eigenvalues as opposed 
to the CDT algorithm which becomes unstable.

\subsection{Outline of the paper}
This paper is organized as follows: In Sec.~\ref{sec:akns}, we summarize the
basic scattering theory and the Darboux transformation in the continuous regime. 

The discrete scattering framework for the ZS-problem is developed in 
Sec.~\ref{sec:differential-formulation} where 
the numerical discretization in the spectral domain is described in
Sec.~\ref{eq:discrete-spec-domain}, and, properties of the numerical Jost 
solutions are discussed in Sec.~\ref{sec:jost-sc}. We formulate the 
layer-peeling scheme in Sec.~\ref{sec:dicrete-inft} which is based on the
discrete framework developed in Sec.~\ref{eq:discrete-spec-domain}. Algorithmic
aspects are addressed in Sec.~\ref{sec:sequential} and~\ref{sec:FPA} where 
we describe the sequential algorithm and its fast version obtained 
using a divide-and-conquer strategy, respectively. The 
sections~\ref{sec:invscattering-Lubich} to~\ref{sec:general-DT} contain the main
contribution of this paper: The method of inversion of continuous scattering
coefficients using the Lubich's method is discussed in Sec.~\ref{sec:invscattering-Lubich}. 
In Sec.~\ref{sec:DT-pure-soliton}, we apply the Lubich's method to obtain the 
FDT algorithm for $K$-soliton potentials.
Finally, the general version of the CDT algorithm and the FDT algorithm is
discussed in Sec.~\ref{sec:general-DT}.

The benchmarking methods that used for
comparison are discussed in Sec.~\ref{sec:benchmark}. The necessary and 
sufficient condition for discrete inverse
scattering is discussed in Sec.~\ref{sec:finite-support-seq}. The stability and
convergence analysis of the numerical schemes developed in earlier section is
carried out in Sec.~\ref{sec:stability-convg-analysis}. The numerical experiments 
and results are discussed in the 
Sec.~\ref{sec:numerical-tests} which is followed by 
Sec.~\ref{sec:conclusion} which concludes the paper.

\section{The AKNS System}
\label{sec:akns}
In order to describe the fundamental basis of the nonlinear Fourier transform
(NFT), we briefly review the scattering theory for a $2\times2$ AKNS system 
corresponding to the NSE. Because the NSE shows up in various disciplines, we 
choose to present the theory in a form that is independent of the context and
conform to the way it appears in the classical texts on the scattering theory. For a 
complex valued field $q(x,t)$, we will work with the standard form of NSE 
which reads as
\begin{equation}\label{eq:nse-classic}
iq_t=q_{xx}+2|q|^2q,\quad(x,t)\in\field{R}\times\field{R}_+,
\end{equation}
where $t>0$ is the evolution parameter identified as a \emph{time-like}
variable (this turns out to be the propagation distance $Z$ for the fiber
model) and $x\in\field{R}$ is the domain over which the field is defined 
(which is the retarded time $T$ for the fiber model). Henceforth, we closely 
follow the formalism developed in~\cite{AKNS1974,AS1981} for the exposition 
in this article. The NFT of the complex-valued field $q(x,t)$ is introduced 
via the associated \emph{Zakharov-Shabat scattering problem}~\cite{ZS1972} 
which can be stated as follows:
Let $\zeta\in\field{R}$ and $\vv{v}=(v_1,v_2)^{\tp}\in\field{C}^2$, then 
\begin{align}
&\vv{v}_x = -i\zeta\sigma_3\vv{v}+U\vv{v},\label{eq:zs-prob}\\
&\vv{v}_t = 2i\zeta^2\sigma_3\vv{v} + 
    [-2\zeta U + i\sigma_3(U^2-U_x)]\vv{v}\label{eq:OpT},
\end{align}
where
\begin{equation}
U=\begin{pmatrix}
0& q(x,t)\\
r(x,t)&0
\end{pmatrix},\quad r(x,t)=-q^*(x,t),
\end{equation}
is identified as the \emph{scattering potential}. The second relation above 
corresponds to the focusing-type of nonlinearity for the NSE. The compatibility 
condition ($\vv{v}_{xt}=\vv{v}_{tx}$) between~\eqref{eq:zs-prob} and~\eqref{eq:OpT}, 
assuming $\zeta$ is independent of $t$, produces the NSE as stated 
in~\eqref{eq:nse-classic}.

The solution of the scattering problem~\eqref{eq:zs-prob}, henceforth referred
to as the ZS-problem, consists in finding the so called 
\emph{scattering coefficients} which are defined through 
special solutions of~\eqref{eq:zs-prob}
known as the \emph{Jost solutions} described in the next 
subsection. These Jost solutions also play an important role in defining the
\emph{Darboux transformation} (DT) which is a powerful technique 
for constructing more complex
potentials (as well as their Jost solutions) from simpler ones--this will be
discussed in the final part of this section. There, we will be primarily interested in
studying the form of DT which allows one to add bound states to a given
potential.
 
\subsection{Jost solutions}
\label{sec:Jost-solutions}
The \emph{Jost solutions} are linearly independent solutions of~\eqref{eq:zs-prob} 
such that they have a plane-wave like behavior at $+\infty$ or $-\infty$. In the following, 
we set $t=0$ and suppress the time-dependence of the solutions for the
sake of brevity.
\begin{itemize}
\item\emph{First kind}: The Jost solutions of the first kind, denoted
by $\vs{\psi}(x;\zeta)$ and $\overline{\vs{\psi}}(x;\zeta)$, are the linearly
independent solutions of~\eqref{eq:zs-prob} which have the following asymptotic 
behavior as $x\rightarrow\infty$:
$\vs{\psi}(x;\zeta)e^{-i\zeta x}\rightarrow(0,1)^{\tp}$ and 
$\overline{\vs{\psi}}(x;\zeta)e^{i\zeta x}\rightarrow (1,0)^{\tp}$.

\item\emph{Second kind}: The Jost solutions of the second kind, denoted by
$\vs{\phi}(x,\zeta)$ and $\overline{\vs{\phi}}(x,\zeta)$, are the linearly 
independent solutions of~\eqref{eq:zs-prob} which have the following asymptotic 
behavior as $x\rightarrow-\infty$: 
$\vs{\phi}(x;\zeta)e^{i\zeta x}\rightarrow(1,0)^{\tp}$ and 
$\overline{\vs{\phi}}(x;\zeta)e^{-i\zeta x}\rightarrow(0,-1)^{\tp}$.
\end{itemize}
The evolution of the Jost solutions in time is governed by the 
equation~\eqref{eq:OpT} for $t\in\field{R}_+$ under the asymptotic boundary 
conditions prescribed above. On account of the linear independence of 
$\vs{\psi}$ and $\overline{\vs{\psi}}$, we have
\begin{equation*}
\begin{split}
\vs{\phi}(x;\zeta)&=a(\zeta)\ovl{\vs{\psi}}(x;\zeta)
+b(\zeta)\vs{\psi}(x;\zeta),\\
\ovl{\vs{\phi}}(x;\zeta)&=-\ovl{a}(\zeta)\vs{\psi}(x;\zeta)
+\ovl{b}(\zeta)\ovl{\vs{\psi}}(x;\zeta).
\end{split}
\end{equation*}
Similarly, using the pair $\vs{\phi}$ and $\overline{\vs{\phi}}$, we have
\begin{equation*}
\begin{split}
\vs{\psi}(x;\zeta)&=-a(\zeta)\overline{\vs{\phi}}(x;\zeta)
+\ovl{b}(\zeta)\vs{\phi}(x;\zeta),\\
\overline{\vs{\psi}}(x;\zeta)&=\overline{a}(\zeta)\vs{\phi}(x;\zeta)
+{b}(\zeta)\overline{\vs{\phi}}(x;\zeta).
\end{split}
\end{equation*}
The coefficients appearing in the equations above can be written in terms of the
Jost solutions by using the Wronskian relations\footnote{For any pair of 
linearly independent vectors, $\vv{v},\,\vv{u}\in\field{C}^2$, their Wronskian 
which is defined as
\[
\Wrons\left(\vv{u},\vv{v}\right)=\left(\vv{u},\vv{v}\right)=u_1v_2-v_1u_2,
\]
is non-zero. If $\vv{u},\vv{v}$ also qualify as Jost solutions, then their
Wronskian is independent of $x$~\cite{AKNS1974}.}: 
\begin{equation}\label{eq:wrons-scoeff}
\begin{split}
&a(\zeta)= \Wrons\left(\vs{\phi},{\vs{\psi}}\right),\quad 
 b(\zeta)= \Wrons\left(\ovl{\vs{\psi}},\vs{\phi}\right),\\
&\ovl{a}(\zeta)=\Wrons\left(\ovl{\vs{\phi}},\overline{\vs{\psi}}\right),\quad
\ovl{b}(\zeta) =\Wrons\left(\ovl{\vs{\phi}},{\vs{\psi}}\right).
\end{split}
\end{equation}
These coefficients are known as the \emph{scattering coefficients} and the process
of computing them is referred to as \emph{forward scattering}. As it turns
out, we would also be interested in studying the analytic continuation of the Jost 
solutions with respect to $\zeta$, which in turn also determines the analytic continuation 
of the scattering coefficients. The motivation behind this is threefold: First, the 
inversion of the scattering coefficients cannot be done in general by knowing the value 
of the scattering coefficients over the real line (i.e. $\zeta\in\field{R}$).
Second, the knowledge of analyticity and decay properties of these functions in 
the complex plane allows us to establish certain theoretical estimates with 
greater ease. Lastly, in many cases, the knowledge of the analytic form
introduces a certain redundancy in the system that can be exploited by the
numerical algorithms to improve its numerical conditioning and stability.

In order to discuss the analytic continuation of the Jost solution with
respect to $\zeta$, let us specify the following two classes of functions for
the scattering potential (at $t=0$): 
Let $q(\cdot,0)\in\fs{L}^1$ such that $\supp q(\cdot,0)\subset\Omega=[L_1,L_2]$ or 
$|q(x,0)|\leq C\exp[-2d |x|]$ almost everywhere in
$\field{R}$ for some constants $C>0$ and $d>0$. In the former case, the Jost
solutions have analytic continuation in whole of the complex plane with respect to
$\zeta$. Consequently, the scattering coefficients  $a(\zeta)$, $b(\zeta)$,
$\ovl{a}(\zeta)$, $\ovl{b}(\zeta)$ are analytic functions of $\zeta\in\field{C}$.
In the latter case, the analyticity property can be summarized as follows~\cite[Sec.~IV.A]{AKNS1974}: 
The functions $e^{-i\zeta x}\vs{\psi}$ and $e^{i\zeta x}\vs{\phi}$ 
are analytic in the half-space $\{\zeta\in\field{C}|\,\Im{\zeta}>-d\}$. The functions 
$e^{i\zeta x}\overline{\vs{\psi}}$ and $e^{-i\zeta x}\overline{\vs{\phi}}$ 
are analytic in the half-space $\{\zeta\in\field{C}|\,\Im{\zeta}<d\}$. In this 
case, the coefficient $a(\zeta)$ is analytic for $\Im{\zeta}>-d$ while the coefficient 
$b(\zeta)$ is analytic in the strip defined by $-d<\Im\zeta<d$. More will be
said about the analyticity and decay properties of the scattering coefficients
in Sec.~\ref{sec:compact-one-sided}.

Furthermore, the symmetry properties
\begin{equation}
\begin{split}
\ovl{\vv{\psi}}(x;\zeta)
=i\sigma_2{\vv{\psi}}^*(x;\zeta^*)=\begin{pmatrix}
\psi_2^*(x;\zeta^*)\\
-\psi_1^*(x;\zeta^*)\\
\end{pmatrix},\\
\ovl{\vv{\phi}}(x;\zeta)
=i\sigma_2{\vv{\phi}}^*(x;\zeta^*)=\begin{pmatrix}
\phi_2^*(x;\zeta^*)\\
-\phi_1^*(x;\zeta^*)\\
\end{pmatrix},
\end{split}
\end{equation}
yield the relations $\ovl{a}(\zeta)=a^*(\zeta^*)$ and $\ovl{b}(\zeta)=b^*(\zeta^*)$.
\subsection{Scattering data and the nonlinear Fourier spectrum}
\label{sec:scattering-data}
The scattering coefficients introduced in the last section together with certain 
quantities defined below that facilitate the recovery of the scattering potential are 
collectively referred to as the \emph{scattering data}. The \emph{nonlinear Fourier
spectrum} can then be defined as any of the subsets which qualify as the ``primordial''
scattering data~\cite[App.~5]{AKNS1974}, i.e., the minimal set of quantities 
sufficient to determine the scattering potential, uniquely.

In general, the nonlinear Fourier spectrum for the potential $q(x,0)$ comprises 
a \emph{discrete} and a \emph{continuous spectrum}. The discrete spectrum consists 
of the so called \emph{eigenvalues} $\zeta_k\in\field{C}_+$, such that 
$a(\zeta_k)=0$, and, the \emph{norming constants} $b_k$ such that 
$\vs{\phi}(x;\zeta_k)=b_k\vs{\psi}(x;\zeta_k)$. For convenience, let the
discrete spectrum be denoted by the set
\begin{equation}\label{eq:set-discrete-spectrum}
\mathfrak{S}_K=\{(\zeta_k,b_k)\in\field{C}^2|\,\Im{\zeta_k}>0,\,k=1,2,\ldots,K\}.
\end{equation}
For compactly supported potentials, $b_k=b(\zeta_k)$. Note that some authors choose to
define the discrete spectrum using the pair $(\zeta_k,\rho_k)$ where
$\rho_k=b_k/\dot{a}(\zeta_k)$ is known as the \emph{spectral amplitude} corresponding
to $\zeta_k$ ($\dot{a}$ denotes the derivative of $a$). 

The continuous spectrum, also referred to as the \emph{reflection coefficient}, is 
defined by $\rho(\xi)={b(\xi)}/{a(\xi)}$ for $\xi\in\field{R}$. 
The coefficient $a(\zeta)$ and 
consequently the discrete eigenvalues do not evolve in time. The rest of the 
scattering data evolves according to the relations
$b_k(t)=b_ke^{-4i\zeta_k^2t}$ and $\rho(\xi,t)=\rho(\xi)e^{-4i\xi^2t}$.

\subsection{The Darboux transformation}\label{sec:DT}
The \emph{Darboux transformation} provides a purely algebraic means of adding bound states
to a seed solution~\cite{NM1984,Lin1990,GHZ2005}. In doing so
the $b$-coefficient of the potential remains invariant~\cite{Lin1990} while the
$a$-coefficient gets modified to reflect the addition of the bound states. In 
particular, starting from the ``vacuum'' solution (i.e. the solution for the 
null-potential), one can compute reflectionless potentials also referred to 
as the multi-soliton or, more precisely, the $K$-soliton potential with the desired 
discrete spectrum. The Darboux transformation is carried out by means of
\emph{Darboux matrices} which is described in the following paragraphs.

Let $\mathfrak{S}_K$ as defined by~\eqref{eq:set-discrete-spectrum} be the 
discrete spectrum to be added to the seed potential. 
Define the matrix form of the Jost solutions as 
\begin{equation}
{v}(x,t;\zeta) = (\vs{\phi},\vs{\psi})=
\begin{pmatrix}
\phi_1&\psi_1\\
\phi_2&\psi_2
\end{pmatrix}.
\end{equation}
The augmented matrix Jost solution ${v}_K(x,t;\zeta)$ can be obtained from the 
seed solution $v_0(x,t;\zeta)$ using the Darboux matrix as
\begin{equation*}
{v}_K(x,t;\zeta)=\mu_{K}(\zeta)D_{K}(x,t;\zeta,\mathfrak{S}_K)v_0(x,t;\zeta),
\quad\zeta\in\ovl{\field{C}}_+,
\end{equation*}
where $\mu_{K}(\zeta)$ is to be determined. In the following, we summarize the approach 
proposed by Neugebauer and 
Meinel~\cite{NM1984} which requires the Darboux matrix to be written as
\begin{equation*}
D_{K}(x,t;\zeta,\mathfrak{S}_K)=\sum_{k=0}^{K}D_k^{(K)}(x,t;\mathfrak{S}_K)\zeta^k,
\end{equation*}
where the coefficient matrices are such that (for the special case $r=-q^*$) 
$D^{(K)}_{K} = \sigma_0$ 
and 
\begin{equation*}
D_k^{(K)} = 
\begin{pmatrix}
 {d}^{(k, K)}_{0} & {d}^{(k, K)}_{1}\\
-{d}^{(k, K)*}_{1} & {d}^{(k, K)*}_{0}
\end{pmatrix},\quad k=0,1,\ldots,K-1.
\end{equation*}
From the Wronskian relation,
we know $a_0(\zeta)=\det[v_0]$; hence, it follows that
\begin{equation*}
\begin{split}
a_K(\zeta)&=\det\left[{v}_K(x,t;\zeta)\right]\\
&= [\mu_{K}(\zeta)]^2\det\left[D_{K}(x,t;\zeta,\mathfrak{S}_K)\right]a_0(\zeta).
\end{split}
\end{equation*}
It is shown in~\cite{NM1984} that $\det[D_{K}(x,t;\zeta,\mathfrak{S}_K)]$ is independent of $(x,t)$.
Further, the symmetry imposed by the condition $r=-q^*$, requires
\begin{equation*}
    \det\left[D_{K}(x,t;\zeta,\mathfrak{S}_K)\right] = \prod_{k=1}^{K}(\zeta-\zeta_k)(\zeta-{\zeta}^*_k),
\end{equation*}
which combined with the fact that~\cite{Lin1990}
\begin{equation*}
a_{K}(\zeta)=a_0(\zeta)\prod_{k=1}^{K}\left(\frac{\zeta-{\zeta}_k}{\zeta-{\zeta}^*_k}\right),
\end{equation*}
yields
\begin{equation*}
    \mu_{K}(\zeta)=\prod_{k=1}^{K}\frac{1}{(\zeta-{\zeta}^*_k)}.
\end{equation*}
From $\vs{\phi}_K(x,t;\zeta_k) =
b_{k}(t)\vs{\psi}_K(x,t;\zeta_k)$, we have 
\begin{equation}\label{eq:linear-eq-Darboux}
D_{K}(x,t;\zeta_k,\mathfrak{S}_K)
[\vs{\phi}_0(x,t;\zeta_k)-b_{k}(t)\vs{\psi}_0(x,t;\zeta_k)]=0.
\end{equation}
Note that $\vs{\phi}_0(x,t;\zeta_k) - b_{k}(t)\vs{\psi}_0(x,t;\zeta_k)\neq0$ on
account of $a_0(\zeta_k)\neq0$, i.e., $\zeta_k$ is not an
eigenvalue of the seed potential. The $2K$ system of 
equations in~\eqref{eq:linear-eq-Darboux}
can be used to compute the $2K$ unknown coefficients of 
the Darboux matrix. Let $U_{K}$ and $U_{0}$ correspond to the augmented potential 
$q_K$ and the seed potential $q_0$, respectively; then using the fact that
$v_K(x,t;\zeta)$ is a Jost solution, we have
\[
[D_{K}v_0]_x-(-i\zeta\sigma_3+U_K)D_{K}v_0=0,
\]
which expands to
\begin{equation*}
[\partial_xD_{K}-(-i\zeta\sigma_3+U_K)D_{K}+D_{K}(-i\zeta\sigma_3+U_0)]v_0=0.
\end{equation*}
Given that $v_0$ is invertible, we must have
\begin{equation*}
[\partial_xD_{K}-(-i\zeta\sigma_3+U_K)D_{K}+D_{K}(-i\zeta\sigma_3+U_0)]=0.
\end{equation*}
Equating the coefficient of $\zeta^K$ to zero, we have
\begin{equation}
\begin{split}
U_{K} 
&= U_0 + i[\sigma_3, D^{(K)}_{K-1}]\\
&=U_0 + 
\begin{pmatrix}
0 &  2id^{(K-1, K)}_{1}\\
2id^{(K-1, K)*}_{1} & 0
\end{pmatrix}.
\end{split}
\end{equation}
\subsubsection{Darboux matrix of degree one}
\label{sec:Darboux-degree-one}
For the sake of simplicity, let the us consider the seed solution
with empty discrete spectra. Let us define the successive discrete spectra
$\emptyset=\mathfrak{S}_0\subset\mathfrak{S}_1\subset\mathfrak{S}_2
\subset\ldots\subset\mathfrak{S}_K$ such that 
${\mathfrak{S}}_j=\{(\zeta_j,b_j)\}\cup{\mathfrak{S}}_{j-1}$ for
$j=1,2,\ldots,K$ where $(\zeta_j,b_j)$ are distinct elements of $\mathfrak{S}_K$.

For single bound state, described by $\mathfrak{S}_1$, putting
\begin{equation*}
\beta_0(x,t;\zeta_1,b_1) = \frac{\phi_1^{(0)}(x,t;\zeta_1)-b_{1}(t)\psi_1^{(0)}(x,t;\zeta_1)}
{\phi_2^{(0)}(x,t;\zeta_1) - b_{1}(t)\psi_2^{(0)}(x,t;\zeta_1)},
\end{equation*}
the solution of the corresponding linear system~\eqref{eq:linear-eq-Darboux} yields 
the Darboux matrix of degree one given by
\begin{equation}\label{eq:one-DT}
\begin{split}
&D_1(x,t; \zeta,\mathfrak{S}_1|\mathfrak{S}_0)\\
&= \zeta \sigma_0 - 
\begin{pmatrix}
\beta_0 &  1\\
    1  &  -\beta^*_0
\end{pmatrix}
\begin{pmatrix}
\zeta_1 &  0\\
    0 & {\zeta}^*_1
\end{pmatrix}\begin{pmatrix}
\beta_0 &  1\\
    1  &  -\beta^*_0
\end{pmatrix}^{-1}\\
&= \zeta \sigma_0 - 
\frac{1}{1+|\beta_0|^2}\begin{pmatrix}
    |\beta_0|^2\zeta_1+\zeta_1^* &  (\zeta_1-\zeta_1^*)\beta_0\\
    (\zeta_1-\zeta_1^*)\beta^*_0 &  \zeta_1+\zeta_1^*|\beta_0|^2
\end{pmatrix}.
\end{split}
\end{equation}
The augmented potential then works out as
\begin{equation}
q_1(x,t) = q_0(x,t) - 2i\frac{(\zeta_1-\zeta_1^*)\beta_0}{1+|\beta_0|^2}.
\end{equation}
The Jost solutions for this new potential can be obtained via the Darboux matrix
and the entire procedure can be repeated for adding another bound state to the
augmented potential. 
Suppressing the $x$ and $t$ dependence for the sake of brevity, it follows that the 
Darboux matrix of degree $K>1$ can be factorized into
Darboux matrices of degree one as
\begin{multline*}
D_K(\zeta,\mathfrak{S}_K|\mathfrak{S}_0)
=D_1(\zeta,\mathfrak{S}_{K}|\mathfrak{S}_{K-1})\\\times
D_1(\zeta,\mathfrak{S}_{K-1}|\mathfrak{S}_{K-2})\times
\ldots\times D_1(\zeta,\mathfrak{S}_1|\mathfrak{S}_0),
\end{multline*}
where $D_1(\zeta,\mathfrak{S}_{j}|\mathfrak{S}_{j-1}),\,j=1,\ldots,K$ are 
the successive Darboux matrices of degree 
one with the convention that 
$(\zeta_{j},b_{j})=\mathfrak{S}_{j}\cap\mathfrak{S}_{j-1}$ is the bound
state being added to the seed solution whose discrete spectra is
$\mathfrak{S}_{j-1}$. Using the expression in~\eqref{eq:one-DT}, we have
\begin{multline*}
D_1(\zeta,\mathfrak{S}_{j}|\mathfrak{S}_{j-1})= \zeta \sigma_0-\\
\frac{1}{1+|\beta_{j-1}|^2}\begin{pmatrix}
|\beta_{j-1}|^2\zeta_j+\zeta_j^* &  (\zeta_j-\zeta_j^*)\beta_{j-1}\\
(\zeta_j-\zeta_j^*)\beta^*_{j-1}&  \zeta_j+\zeta_j^*|\beta_{j-1}|^2
\end{pmatrix},
\end{multline*}
where
\begin{equation*}
    \beta_{j-1}(\zeta_j, b_j) =
    \frac{\phi_1^{(j-1)}(\zeta_j)-b_{j}\psi_1^{(j-1)}(\zeta_j)}
    {\phi_2^{(j-1)}(\zeta_j) - b_{j}\psi_2^{(j-1)}(\zeta_j)},
\end{equation*}
for $(\zeta_j,b_j)\in\mathfrak{S}_K$ and the successive Jost solutions, 
${v}_{j} = (\vs{\phi}_{j},\vs{\psi}_{j})$, needed in this ratio are computed as
\begin{equation*}
    {v}_j=\frac{1}{(\zeta-\zeta^*_j)}D_{1}(\zeta,\mathfrak{S}_j|\mathfrak{S}_{j-1})
    v_{j-1}.
\end{equation*}
The successive potentials are given by
\begin{equation*}
q_j = q_{j-1} -
2i\frac{(\zeta_j-\zeta_j^*)\beta_{j-1}}{1+|\beta_{j-1}|^2}.
\end{equation*}
See Fig.~\ref{fig:DT-block} for a schematic representation of the DT.
\begin{figure*}[!t]
\centering
\includegraphics[scale=1.0]{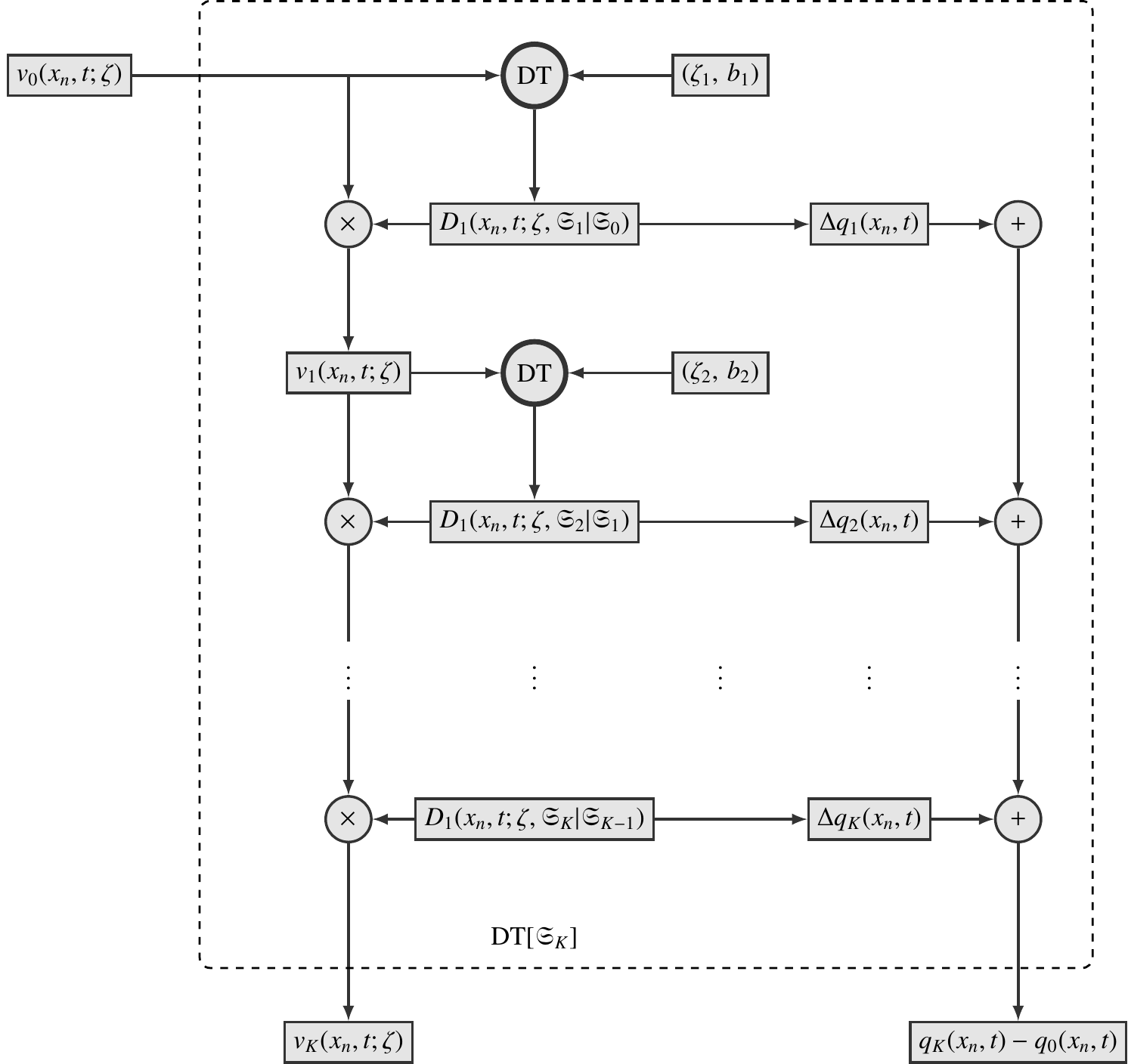}
\caption{\label{fig:DT-block} The figure shows the schematic of the classical 
Darboux transformation for a given discrete spectrum, $\mathfrak{S}_K$, where
the discrete spectrum of the seed solution is empty, i.e. 
$\mathfrak{S}_0=\emptyset$, and the given grid point, $x_n$. The part enclosed 
within broken lines is referred to as the complete DT-block (labeled as
DT$[\mathfrak{S}_K]$). The sole input to this block is the seed Jost
solution, $v_0(x_n,t;\zeta)$. The output of the DT-block consists of 
the augmented Jost solution, $v_K(x_n,t;\zeta)$, and the
difference between the augmented and the seed potential,
$q_K(x_n,t)-q_0(x_n,t)$. Here, $\Delta q_j(x_n,t) = q_j(x_n,t)
- q_{j-1}(x_n,t)$ and
${\mathfrak{S}}_j={\mathfrak{S}}_{j-1}\cup\{(\zeta_j,b_j)\}$ with
${\mathfrak{S}}_1=\{(\zeta_1,b_1)\}\cup\mathfrak{S}_0$ where
$(\zeta_j,b_j),\,j=1,2,\ldots,K$ are the distinct elements of $\mathfrak{S}_K$
(see Sec.~\ref{sec:scattering-data}).}
\end{figure*}

If the seed Jost solution $v_0(x,t;\zeta)$ corresponding to the seed 
potential $q_0(x,t)$ is known, then the Darboux transformations can be 
readily carried out over any set of grid points 
$\{x_n\}\subset\field{R}$ in order to compute the augmented potential 
at these grid points. The resulting order 
of operational complexity, excluding the cost of evaluating the seed potential 
and the seed Jost solution, works out to be $\bigO{K^2N}$ where $N$ is the
number of samples of the augmented potential. For the special case of $K$-solitons, 
the seed potential as well as the seed Jost solutions are trivially known; 
therefore, this method provides us with an algorithm for computing the 
$K$-soliton potentials with machine precision. In general, closed form solutions 
are rarely known for arbitrary potentials; nevertheless, this procedure can be 
carried out with numerically computed Jost solutions in any discrete framework. This 
scheme will be referred to as the \emph{classical Darboux 
transformation} (CDT) in the rest of the article. The error analysis of this method is carried out in 
Sec.~\ref{sec:accuracy-DT}.

For multi-solitons, the asymptotic
form of the potential as $x\rightarrow\infty$ works out to be
\[
q_K(x,t)\sim2i\sum_{j=1}^K\frac{(\zeta_j-\zeta_j^*)}{a^*_{j-1}(\zeta_j)}
b^*_j(t)e^{-2i\zeta^*_jx},
\]
and as $x\rightarrow-\infty$
\[
q_K(x,t)\sim2i\sum_{j=1}^K\frac{(\zeta_j-\zeta_j^*)}{a_{j-1}(\zeta_j)}
\frac{1}{b_j(t)}e^{-2i\zeta_jx},
\]
where $a_j(\zeta)=a(\zeta;\mathfrak{S}_j)$ are the successive $a$-coefficients. 
Therefore, $q_K(x,t)$ exhibits exponential decay with a decay constant that is 
given by $d_K=\min_{1\leq j\leq K}\Im{\zeta}_j$. \emph{This observation 
allows us to conclude that round off errors in the CDT scheme can be minimized 
if the eigenvalues are ``added'' in the decreasing order of the magnitude of their
imaginary parts~\cite{VW2016OFC}}. Further, the knowledge of the decay constant
can be used to choose an optimal computational domain so that the numerical errors 
due to domain truncation is minimized (see Sec.~\ref{ex:multi-soliton}).

\subsubsection{Effective support of multi-soliton potentials}
A multi-soliton potential has an unbounded support, therefore, in any practical
application it is mandatory to introduce an effective support with desired
energy content. Posed conversely, one may also be interested in choosing the
discrete spectrum which leads to a prescribed effective support with desired energy
content initially or over a finite duration of evolution. 

In case of multi-solitons, the energy content of the side lobe 
which we wish to truncate is trivially
available in the CDT scheme and it can be used as a truncation criteria. Let
$\chi_{\Omega}$ denote the characteristic function of $\Omega$ and let
$[-L,L]$ ($L>0$) be the domain that needs to be determined so that
\begin{equation}\label{eq:side-lobe}
\|q_K\chi_{(-\infty,-L]}\|^2_{\fs{L}^2}+\|q_K\chi_{[L,\infty)}\|^2_{\fs{L}^2}
=\epsilon\|q_K\|^2_{\fs{L}^2}.
\end{equation}
Suppressing the dependence on $t$ for the sake of brevity, 
the asymptotic expansion of $\phi_K(-L;\zeta)e^{-i\zeta L}$ with respect to
$\zeta$ yields~\cite[Sec.~IV.A]{AKNS1974}
\begin{equation}
\|q_K\chi_{(-\infty,-L]}\|^2_{\fs{L}^2}=\sum_{j=1}^K\frac{4\Im(\zeta_j)}
{\left[1+|\beta_{j-1}(-L;\zeta_j,b_j)|^{-2}\right]},
\end{equation}
and that corresponding to $\psi_K(L;\zeta)e^{-i\zeta L}$ yields
\begin{equation}
\|q_K\chi_{[L,\infty)}\|^2_{\fs{L}^2}=\sum_{j=1}^K\frac{4\Im(\zeta_j)}
{\left[1+|\beta_{j-1}(L;\zeta_j,b_j)|^2\right]}.
\end{equation}
These relationships are also known as the \emph{nonlinear Parseval's
relationships}. 
Asymptotic estimates when $L\gg1$ can be easily obtained from the above relations:
\begin{align*}
&\|q_K\chi_{(-\infty,-L]}\|^2_{\fs{L}^2}\sim
\sum_{j=1}^K\frac{4\Im(\zeta_j)}{|a_{j-1}(\zeta_j)|^2}\frac{1}{|b_j|^2}e^{-4\Im(\zeta_j)L},\\
&\|q_K\chi_{[L,\infty)}\|^2_{\fs{L}^2}\sim
\sum_{j=1}^K\frac{4\Im(\zeta_j)}{|a_{j-1}(\zeta_j)|^2}{|b_j|^2}e^{-4\Im(\zeta_j)L}.
\end{align*}
This allows us to obtain an asymptotic formula for the effective support of a
$K$-soliton potential. Define $L=L(\epsilon;\mathfrak{S}_K)>0$ such that 
\[
|q_K\chi_{[-L,L]}\|^2_{\fs{L}^2}=(1-\epsilon)\|q_K\|^2_{\fs{L}^2}=4(1-\epsilon)\sum_{j=1}^K\eta_j,
\] 
then
\begin{equation}
L\approx W=\frac{1}{2\eta_{\text{min}}}
\log\left[\frac{\sum_{j=1}^K\omega_j\eta_j}{\epsilon\sum_{j=1}^K\eta_j}\right],
\end{equation}
under the assumption $\epsilon\sum_{j=1}^K\eta_j\ll\sum_{j=1}^K\omega_j\eta_j$
where
\[
\omega_j=\frac{1}{|a_{j-1}(\zeta_j)|^2}\left(|b_j|^2+\frac{1}{|b_j|^2}\right).
\]
Finally, let us note that a binary search algorithm (bisection method) can be devised to solve the
nonlinear equation~\eqref{eq:side-lobe} for
$L=L(\epsilon,\mathfrak{S}_K)$ where $[0,W]$ can be taken as the
bracketing interval for the root\footnote{Numerical tests indicates that $[-W,W]$
tends to over estimate the effective support.}. The complexity of such an
algorithm (for fixed $t$) works out to be $\bigO{mK^2}$ where $m$ is the number of
bisection steps needed.

\subsubsection{Scattering coefficients of a truncated multi-soliton}
\label{sec:trancated-soliton}
Let $x=0$ be taken as the point of truncation. Then a multi-soliton potential 
can be seen as comprising a left-sided profile (supported in
${\field{R}}_-\cup\{0\}$) and a right-sided profile (supported in
$\{0\}\cup{\field{R}}_+$). The respective scattering coefficients of each of the
truncated potentials turn out to be a rational function of $\zeta$. These 
observations were already made by several 
authors~\cite{Lamb1980,RM1992,RS1994,S2002,SK2008} and a number of different
methods do exist for inversion of the scattering data which exploit the rational 
character of the truncated scattering coefficients. Our numerical scheme also 
exploits this property; therefore, we discuss this case in some detail below.

Let us consider the left-sided profile, denoted by $q^{(-)}(x,t)$. The Jost
solution $\vv{\phi}^{(-)}(x,t;\zeta)$ at $x=0$ can be computed using the 
Darboux transformation as described above. The Jost solution
$\vv{\psi}^{(-)}(x,t;\zeta)$ at $x=0$ corresponds to that of a null-potential, i.e.,
$\vs{\psi}^{(-)}(0,t;\zeta)=(0,1)^{\tp}$. The scattering coefficients for the 
left-sided profile, therefore, works out to be
\begin{equation*}
a^{(-)}(\zeta,t) = \phi^{(-)}_1(0,t;\zeta),\quad 
b^{(-)}(\zeta,t) = \phi^{(-)}_2(0,t;\zeta).
\end{equation*}
This corresponds to the first column of the Darboux matrix
$D_{K}(0,t;\zeta,\mathfrak{S}_K)$,
therefore, a purely rational function of $\zeta$ analytic in $\ovl{\field{C}}_+$. 
Now, let us consider the right-sided profile, denoted by $q^{(+)}(x,t)$. The 
Jost solution $\vv{\psi}^{(+)}(x,t;\zeta)$ at $x=0$ can be computed using the 
Darboux transformation as before while the Jost solution
$\vv{\phi}^{(+)}(x,t;\zeta)$ at $x=0$ is given by $\vs{\phi}^{(+)}(0,t;\zeta)=(1,0)^{\tp}$.
Therefore, the relevant scattering coefficients for the right-sided profile works
out to be
\begin{equation*}
a^{(+)}(\zeta,t) = \psi^{(+)}_2(0,t;\zeta),\quad
\ovl{b}^{(+)}(\zeta,t) = \psi^{(+)}_1(0,t;\zeta).
\end{equation*}
This corresponds to the second column of the Darboux matrix 
$D_{K}(0,t;\zeta,\mathfrak{S}_K)$ and, therefore, a purely rational function of 
$\zeta$ analytic in $\ovl{\field{C}}_+$.
\begin{rem}[Conjugation and reflection]\label{rem:reflection}
The inverse scattering problem for the right-sided profile can be transformed to
that of a left-sided profile in the following way: putting $y = -x$, we have
\begin{align*}
&\vv{v}_y(-y;\zeta) = i\zeta\sigma_3\vv{v}(-y;\zeta)-U(-y)\vv{v}(-y;\zeta),\\
&\vv{w}_y = -i\zeta\sigma_3\vv{w}+{U}^*(-y)\vv{w},
\end{align*}
where  $\vv{w}(y) = \sigma_1\vv{v}(-y;\zeta)$. Denote the Jost solutions of the
new system (i.e. with potential $U^*(-y)$) by $\vs{\Psi}(y;\zeta)$, 
$\ovl{\vs{\Psi}}(y;\zeta)$ (first kind) and $\vs{\Phi}(y;\zeta)$, 
$\ovl{\vs{\Phi}}(y;\zeta)$ (second kind), then
\begin{align*}
&\vs{\Psi}(y;\zeta)=\sigma_1\vs{\phi}(-y;\zeta),\quad
\ovl{\vs{\Psi}}(y;\zeta)=-\sigma_1\ovl{\vs{\phi}}(-y;\zeta),\\
&\vs{\Phi}(y;\zeta)=\sigma_1\vs{\psi}(-y;\zeta),\quad
\ovl{\vs{\Phi}}(y;\zeta)=-\sigma_1\ovl{\vs{\psi}}(-y;\zeta).
\end{align*}
Let $A(\zeta)$, $B(\zeta)$, $\ovl{A}(\zeta)$ and $\ovl{B}(\zeta)$ be the
scattering coefficients for the new system, then
\begin{align*}
&A(\zeta)       =\OP{W}(\vs{\Phi},{\vs{\Psi}})=a(\zeta),
    \quad\ovl{A}(\zeta) =\OP{W}(\ovl{\vs{\Phi}},\ovl{\vs{\Psi}})=\ovl{a}(\zeta),\\
&B(\zeta)       =\OP{W}(\ovl{\vs{\Psi}},\vs{\Phi}) = \ovl{b}(\zeta),
    \quad\ovl{B}(\zeta) =\OP{W}(\ovl{\vs{\Phi}},{\vs{\Psi}})=b(\zeta).
\end{align*}
The discrete eigenvalues do not change, however, the norming constants change as 
$B_k=1/b_k$.
Now, the scattering coefficients for the left-sided profile obtained as result
of truncating the new potential from the right at $x=0$ work out to be
\begin{equation*}
\begin{split}
A^{(-)}(\zeta,t) =\Phi_1(0,t;\zeta) = \psi_2(0,t;\zeta),\\
B^{(-)}(\zeta,t) =\Phi_2(0,t;\zeta) = \psi_1(0,t;\zeta).
\end{split}
\end{equation*}
Therefore, an implementation for the case of left-sided profile is sufficient
to solve problems of general nature encountered in forward/inverse NFT.
\end{rem}

\begin{rem}[Translation]\label{rem:truncation}
Let us note that there is no loss of generality in choosing the point of
truncation to be $x=0$ on account of the translational properties of 
the discrete spectrum. If we wish to choose 
the point of truncation to be $x=x_0$, we can consider the transformation 
$x = y + x_0$. Define the new potential to be $\wtilde{U}(y) = U(y+x_0)$ so
that  
\begin{align*}
&\vv{v}_y(y+x_0;\zeta) = [-i\zeta\sigma_3+{U}(y+x_0)]\vv{v}(y+x_0;\zeta),\\
&\vv{w}_y = -i\zeta\sigma_3\vv{w}+\wtilde{U}(y)\vv{w},
\end{align*}
where $\vv{w}(y;\zeta)=\vv{v}(y+x_0;\zeta)$.
Denote the Jost solutions of the
new system by $\vs{\Psi}(y;\zeta)$, $\ovl{\vs{\Psi}}(y;\zeta)$ 
(first kind) and $\vs{\Phi}(y;\zeta)$, $\ovl{\vs{\Phi}}(y;\zeta)$ (second kind), then
\begin{align*}
&\vs{\Psi}(y;\zeta)=\vs{\psi}(y+x_0;\zeta)e^{-i\zeta x_0},\\
&\ovl{\vs{\Psi}}(y;\zeta)=\ovl{\vs{\psi}}(y+x_0;\zeta)e^{+i\zeta x_0},\\
&\vs{\Phi}(y;\zeta)=\vs{\phi}(y+x_0;\zeta)e^{+i\zeta x_0},\\
&\ovl{\vs{\Phi}}(y;\zeta)=\ovl{\vs{\phi}}(y+x_0;\zeta)e^{-i\zeta x_0}.
\end{align*}
Let $A(\zeta)$, $B(\zeta)$, $\ovl{A}(\zeta)$ and $\ovl{B}(\zeta)$ be the
scattering coefficients for the new system, then
\begin{align*}
&A(\zeta)=\OP{W}(\vs{\Phi},{\vs{\Psi}})=a(\zeta),\,\, 
B(\zeta) =\OP{W}(\ovl{\vs{\Psi}},\vs{\Phi}) = {b}(\zeta)e^{2i\zeta x_0},\\
&\ovl{A}(\zeta)=\OP{W}(\ovl{\vs{\Phi}},\ovl{\vs{\Psi}})=\ovl{a}(\zeta),\,\,
\ovl{B}(\zeta)
=\OP{W}(\ovl{\vs{\Phi}},{\vs{\Psi}})=\frac{\ovl{b}(\zeta)}{e^{2i\zeta x_0}}.
\end{align*}
The discrete eigenvalues do not change, however, the norming constants change as 
$B_k=b_ke^{-2i\zeta_k x_0}$.
\end{rem}

\section{Discrete Forward and Inverse Scattering}
\label{sec:differential-formulation}
In this section, we discuss certain discretization schemes for the 
scattering problem in~\eqref{eq:zs-prob} such that they are amenable 
to FFT-based fast polynomial arithmetic~\cite{Henrici1993}. 
This method of obtaining a discrete scattering problem is referred to as the 
\emph{spectral-domain} approach\footnote{See~\cite{BK1987} for alternative approaches.}.
We begin with the transformation $\tilde{\vv{v}}=e^{i\sigma_3\zeta x}\vv{v}$
so that~\eqref{eq:zs-prob} becomes
\begin{equation*}
\partial_x[e^{i\sigma_3\zeta x}\vv{v}]
=e^{i\sigma_3\zeta x}Ue^{-i\sigma_3\zeta x}[e^{i\sigma_3\zeta x}\vv{v}],
\end{equation*}
or,
\begin{equation}\label{eq:exp-int}
\begin{split}
\tilde{\vv{v}}_x&=\wtilde{U}\tilde{\vv{v}},\\
\wtilde{U}&=e^{i\sigma_3\zeta x}Ue^{-i\sigma_3\zeta x}
=\begin{pmatrix}
0 & qe^{2i\zeta x}\\
re^{-2i\zeta x} & 0
\end{pmatrix}.
\end{split}
\end{equation}
The next step is to apply linear one-step method~\cite{Gautschi2012}
to~\eqref{eq:exp-int} in order to setup a recurrence relation initialized by the
given initial condition. Let us note that the method of numerical integration 
just described above is identified as 
the \emph{exponential integrator} based on linear one-step methods, in
particular, the integrating factor (IF) method~\cite{CM2002}. One of the 
advantages of the transformation carried out above in arriving 
at~\eqref{eq:exp-int} is that the ``vacuum'' solution obtained from the 
discrete problem is exact. 

\begin{rem} In the literature, the usage of the terms 
``forward scattering'' and ``inverse scattering'' is not made precise;
for instance, ``forward scattering'' could refer to computation of
the scattering coefficients $a$ and $b$ or the nonlinear Fourier spectrum.
In order to avoid any confusion arising in the usage of these terms, we follow
the convention that the term ``forward scattering'' refers to the computation of 
the Jost solutions while the term ``inverse scattering'' refers to the process 
of recovering the samples of the scattering potential from (the polynomial form of) 
the Jost solutions. Note that in almost all cases, knowledge of the Jost solutions 
trivially allows one to compute the truncated discrete scattering coefficients 
and vice versa, therefore, no confusion should arise in what constitutes as 
input to the inverse scattering process.
\end{rem}

\subsection{Discretization in the spectral-domain}
\label{eq:discrete-spec-domain}
In order to discuss various discretization schemes, we take an equispaced grid defined 
by $x_n= L_1 + nh,\,\,n=0,1,\ldots,N,$ with $x_{N}=L_2$ where $h$ is the grid spacing.
Define $\ell_-,\ell_+\in\field{R}$ such that $h\ell_-= -L_1$, $h\ell_+= L_2$.
Further, let us define $z=e^{i\zeta h}$
and treat $\zeta$ as a fixed parameter. For the 
potential functions sampled on the grid, we set $q_n=q(x_n,t)$, $r_{n}=r(x_n,t)$ where 
the time-dependence is suppressed. Using the same convention, $U_{n}=U(x_n,t)$ and 
$\wtilde{U}_{n}=\wtilde{U}(x_n,t)$.
\subsubsection{Forward Euler method}
\label{sec:discrete-FE}
The forward Euler (FE) method is the simplest of the finite-difference schemes. It
can be stated as
\begin{equation*}
\begin{split}
&\tilde{\vv{v}}_{n+1} = \left(\sigma_0+\wtilde{U}_{n}\right)\tilde{\vv{v}}_{n},\\
&\vv{v}_{n+1} = e^{-i\sigma_3\zeta h}\left(\sigma_0+{U}_{n}\right){\vv{v}}_{n}.
\end{split}
\end{equation*}
Setting $Q_{n}=hq_n$, $R_{n}=hr_n$ and $\Theta_n=(1-Q_{n}R_{n})$, we have
\begin{equation}\label{eq:scatter-FE}
\vv{v}_{n+1}
=z^{-1}
\begin{pmatrix}
1& Q_{n}\\
z^2R_{n}&z^2
\end{pmatrix}\vv{v}_n=z^{-1}M_{n+1}(z^2)\vv{v}_n,
\end{equation}
or, equivalently,
\begin{equation}
\frac{z^{-1}}{\Theta_n}
\begin{pmatrix}
z^2 & -Q_{n}\\
-z^2R_{n} & 1
\end{pmatrix}\vv{v}_{n+1} = \vv{v}_n.
\end{equation}

Let us note that the transfer matrix can be transformed to a form that resembles 
that of the implicit Euler method described in the next section: Putting 
$\vv{w}_n=e^{i\sigma_3\zeta h}{\vv{v}}_{n}$, we have
\begin{equation}
    \vv{w}_{n+1}={z^{-1}}
\begin{pmatrix}
1& z^2Q_{n}\\
R_{n}&z^2
\end{pmatrix}\vv{w}_n.
\end{equation}

\subsubsection{Implicit Euler method}
\label{sec:discrete-BDF1}
The backward differentiation formula of order one (BDF1) is also known as the implicit
Euler method. The discretization of~\eqref{eq:exp-int} using this method reads as 
\begin{equation*}
\begin{split}
&\tilde{\vv{v}}_{n+1} =
\left(\sigma_0-h\wtilde{U}_{n+1}\right)^{-1}\tilde{\vv{v}}_{n},\\
&{\vv{v}}_{n+1} = \left(\sigma_0-h{U}_{n+1}\right)^{-1}
e^{-i\sigma_3\zeta h}{\vv{v}}_{n}.
\end{split}
\end{equation*}
Setting $Q_{n}=hq_n$, $R_{n}=hr_n$ and
$\Theta_n=(1-Q_{n}R_{n})$, this scheme can be stated as follows:
\begin{equation}\label{eq:scatter-BDF1}
    \vv{v}_{n+1}=\frac{z^{-1}}{\Theta_{n+1}}
\begin{pmatrix}
1& z^2Q_{n+1}\\
R_{n+1}&z^2
\end{pmatrix}\vv{v}_n
=z^{-1}M_{n+1}(z^2)\vv{v}_n,
\end{equation}
or, equivalently,
\begin{equation*}
{z^{-1}}
\begin{pmatrix}
z^2& -z^2Q_{n+1}\\
-R_{n+1}&1
\end{pmatrix}\vv{v}_{n+1} = \vv{v}_n.
\end{equation*}
\subsubsection{Trapezoidal rule}
\label{sec:discrete-TR}
The trapezoidal rule (TR) happens to be one of the most popular methods of
integrating ODEs numerically. 
The discretization of~\eqref{eq:exp-int} using this method reads as 
\begin{equation*}
\begin{split}
&\tilde{\vv{v}}_{n+1} 
= \left(\sigma_0-\frac{h}{2}\wtilde{U}_{n+1}\right)^{-1}
\left(\sigma_0+\frac{h}{2}\wtilde{U}_{n}\right)\tilde{\vv{v}}_{n},\\
&{\vv{v}}_{n+1} 
= \left(\sigma_0-\frac{h}{2}{U}_{n+1}\right)^{-1}
e^{-i\sigma_3\zeta h}\left(\sigma_0+\frac{h}{2}{U}_{n}\right){\vv{v}}_{n}.
\end{split}
\end{equation*}
Setting $2Q_{n}=hq_n$, $2R_{n}=hr_n$ and $\Theta_n=1-Q_nR_n$, this scheme can
be stated as follows:
\begin{equation}\label{eq:scatter-TR}
\begin{split}
\vv{v}_{n+1}&=\frac{z^{-1}}{\Theta_{n+1}}
\begin{pmatrix}
1+z^2Q_{n+1}R_n& z^2Q_{n+1}+Q_n\\
R_{n+1}+z^2R_n & R_{n+1}Q_n + z^2
\end{pmatrix}\vv{v}_n\\
&=z^{-1}M_{n+1}(z^2)\vv{v}_n,
\end{split}
\end{equation}
or, equivalently,
\begin{equation*}
\frac{z^{-1}}{\Theta_{n}}
\begin{pmatrix}
R_{n+1}Q_n + z^2& -z^2Q_{n+1}-Q_n\\
-R_{n+1}-z^2R_n & 1+z^2Q_{n+1}R_n
\end{pmatrix}\vv{v}_{n+1} = \vv{v}_n.
\end{equation*}

\subsection{Jost solutions and scattering coefficients}
\label{sec:jost-sc}
In order to express the discrete approximation to the Jost solutions, let us
define the vector-valued polynomial
\begin{equation}\label{eq:poly-vec}
\vv{P}_n(z)=\begin{pmatrix}
            P^{(n)}_{1}(z)\\
            P^{(n)}_{2}(z)
        \end{pmatrix}
         =\sum_{k=0}^n
            \vv{P}^{(n)}_{k}z^k
         =\sum_{k=0}^n
        \begin{pmatrix}
            P^{(n)}_{1,k}\\
            P^{(n)}_{2,k}
        \end{pmatrix}z^k.
\end{equation}
The Jost solutions $\vs{\psi}$ and $\vs{\phi}$, for the forward/implicit Euler method 
and the trapezoidal rule, can be written in the form
\begin{equation}\label{eq:two-jost}
\vs{\psi}_n = z^{\ell_+}z^{-m}\vv{S}_m(z^2),\quad
\vs{\phi}_n = z^{\ell_-}z^{-n}\vv{P}_n(z^2),
\end{equation}
where $m+n=N$. Note that the expressions above correspond to the boundary
conditions $\vs{\psi}_N=z^{\ell_+}(0,1)^{\tp}$ and
$\vs{\phi}_0=z^{\ell_-}(1,0)^{\tp}$ which translate to 
$\vv{S}_0=(0,1)^{\tp}$ and $\vv{P}_0=(1,0)^{\tp}$, respectively. The other Jost
solutions, 
$\ovl{\vs{\psi}}_n$ and $\ovl{\vs{\phi}}_n$, can be written as
\begin{align*}
&\ovl{\vs{\psi}}_n = z^{-\ell_+}z^{m}(i\sigma_2)\vv{S}^*_m(1/z^{*2}),\\
&\ovl{\vs{\phi}}_n = z^{-\ell_-}z^{n}(i\sigma_2)\vv{P}^*_n(1/z^{*2}).
\end{align*}
The recurrence relation for the polynomial functions defined in~\eqref{eq:two-jost} 
take the form
\begin{equation}\label{eq:poly-scatter}
\begin{split}
&\vv{S}_{m+1}(z^2) = \wtilde{M}_{n}(z^2)\vv{S}_m(z^2),\\
&\vv{P}_{n+1}(z^2) = M_{n+1}(z^2)\vv{P}_n(z^2),
\end{split}
\end{equation}
where $M_{n+1}(z^2)$ with its inverse $z^{-2}\wtilde{M}_{n+1}(z^2)$ 
is determined by the respective discretization scheme. The 
discrete approximation to the scattering coefficients is obtained from the scattered
field: $\vs{\phi}_{N}=(a_{N} z^{-\ell_+},b_{N} z^{\ell_+})^{\tp}$ yields
\begin{equation}
a_{N}(z^2)={P}^{(N)}_1(z^2),\quad
b_{N}(z^2)=(z^2)^{-\ell_{+}}{P}^{(N)}_2(z^2).
\end{equation}
and $\vs{\psi}_{0}=(\ovl{b}_{N} z^{\ell_-},a_{N} z^{-\ell_-})^{\tp}$ yields
\begin{equation}
a_{N}(z^2)={S}^{(N)}_2(z^2),\quad
\ovl{b}_{N}(z^2)=(z^2)^{-\ell_{-}}{S}^{(N)}_1(z^2).
\end{equation}
The quantities $a_{N}$, $b_{N}$ and $\ovl{b}_N$ above are referred to as the 
\emph{discrete scattering coefficients}. Note that these coefficients can 
only be defined for $\Re\zeta\in [-{\pi}/{2h},\,{\pi}/{2h}]$.
\begin{rem}
For the sake of brevity, we may occasionally refer to the polynomials
$\vv{S}_m(z^2)$ and $\vv{P}_n(z^2)$ (as opposed to $\vs{\psi}_n$ and 
$\vs{\phi}_n$) as the (discrete) Jost solutions.
\end{rem}

\subsubsection{Discrete spectrum} 
\label{sec:eig-and-nconst}
The eigenvalues are computed by forming $a_{N}(z^2)$ and employing a suitable
root-finding algorithm (see~\cite{WP2013b} and the references therein for more details).
It turns out that the computation of the norming constants by evaluating $b_{N}$ is
ill-conditioned on account of the vanishingly small contribution from the
solitonic components of the potential. Note that addition of bound states leaves
$b$-coefficients invariant; therefore, recovery of the norming constant from
$b(\zeta)$ cannot be expected to succeed in all cases. In order to remedy this
problem, we use the general definition of the norming constants\footnote{Similar
approach is reported in~\cite{HK2016} and~\cite{V2016}, however, it is not
emphasized in these papers that the norming constants are never defined to be a
value of $b(\zeta)$ unless it is guaranteed to be analytic in
$\field{C}_+$. Note that the study of the errors introduced by the numerical
discretization also provides significant insight into why the evaluation of
$b_N(z^2)$ at complex values of $\zeta$ is ill-conditioned (see
Sec.~\ref{sec:nconst-ill-cond}).}: To this end, we proceed by
computing the truncated scattering coefficients.
Consider the case of potentials truncated from the right, i.e., 
$q^{(-)}(x)=\theta(x_1-x)q(x)$ where $x_1$ is
the point of truncation and $\theta(x)$ is the Heaviside step function. The new 
potential now supported in $(-\infty,x_1]$ is interpreted as left-sided with 
respect to $x_1$. The scattering coefficient can be stated in terms of the Jost
solutions of the original potential as~\cite{Lamb1980}
\begin{equation}\label{eq:left-sided-sc}
\begin{split}
&a^{(-)}(\zeta) = \phi_1(x_1;\zeta)e^{i\zeta x_1},\\
&b^{(-)}(\zeta) = \phi_2(x_1;\zeta)e^{-i\zeta x_1}.
\end{split}
\end{equation}
Similarly, for potentials truncated from the left, we have
\begin{equation}\label{eq:right-sided-sc}
\begin{split}
&a^{(+)}(\zeta)      = \psi_2(x_1;\zeta)e^{-i\zeta x_1},\\
&\ovl{b}^{(+)}(\zeta) = \psi_1(x_1;\zeta)e^{i\zeta x_1}.
\end{split}
\end{equation}
Denoting the corresponding discrete scattering coefficients by 
$a^{(-)}_n$, $b^{(-)}_n$, $a^{(+)}_m$ and $\overline{b}^{(+)}_m$, 
where $m+n=\ell_-+\ell_+$, we have
\begin{align*}
&\vs{\phi}_n=
\begin{pmatrix}
z^{ \ell_- - n}a^{(-)}_{n}\\    
z^{-\ell_- + n}b^{(-)}_{n}
\end{pmatrix}=z^{\ell_-}z^{-n}\vv{P}_n(z^2),\\
&\vs{\psi}_n=
\begin{pmatrix}
z^{-\ell_+ + m}\ovl{b}^{(+)}_{m}\\    
z^{ \ell_+ - m}a^{(+)}_{m}\\    
\end{pmatrix}=z^{\ell_+}z^{-m}\vv{S}_m(z^2),
\end{align*}
where $m=N-n$. Here $n$ can be chosen to be $N/2$. Once an admissible root, 
$z_k$, of $a_N(z^2)$ that corresponds to a soliton is
determined\footnote{Given that $z_k=\exp(i\zeta_k h)$ and $\Im\zeta_k>0$, we
must have $|z_k|<1$.}, the corresponding norming constant is obtained via the
proportionality of $\vs{\phi}_n$ and $\vs{\psi}_n$ which translates to 
\begin{equation}\label{eq:norming-constant}
\begin{split}
&b_k = \frac{b^{(-)}_{n}(z^2_k)}{a^{(+)}_{m}(z^2_k)} 
= (z^2)^{\ell_- - n}\frac{P^{(n)}_2(z^2_k)}{S^{(m)}_2(z^2_k)},\\
&\frac{1}{b_k} = \frac{\ovl{b}^{(+)}_{m}(z^2_k)}{a^{(-)}_{n}(z^2_k)} 
= (z^2)^{\ell_+ - m}\frac{S^{(m)}_1(z^2_k)}{P^{(n)}_1(z^2_k)}.
\end{split}
\end{equation}
The truncated potential does not share discrete eigenvalues with the original
potential; therefore, $a^{(+)}_{m}(z_k^2)\neq0$ and $a^{(-)}_{n}(z_k^2)\neq0$. The
computation of the truncated scattering coefficients can be accomplished by
direct evaluation of transfer matrices and subsequently forming the cumulative 
product leading to an operational complexity of $\bigO{N}$ for each eigenvalue 
(see Sec.~\ref{sec:seq-scatter}).

It must be noted that our fast algorithm for forward scattering as discussed
Sec.~\ref{sec:FNFT} is entirely compatible with the approach suggested
here. The scattering coefficients are easily obtainable from the truncated
scattering coefficients using the Wronskian relations given in 
Sec.~\ref{sec:Jost-solutions} as
\begin{equation}\label{eq:full-sc}
\begin{split}
a_{N}(z^2)&=\OP{W}\left(\vv{P}_n(z^2),\,\vv{S}_m(z^2)\right),\\
b_{N}(z^2)&=(z^4)^{(\ell_--n)}\vv{P}_n(z^2)\cdot \vv{S}^{*}_m(1/z^{*2}),\\
\ovl{b}_{N}(z^2)&=(z^4)^{(\ell_+-m)}\vv{P}^*_n(1/z^{*2})\cdot \vv{S}_m(z^2).
\end{split}
\end{equation}
Every polynomial multiplication involved above can be carried out efficiently 
using the FFT algorithm (see Sec.~\ref{sec:FNFT}).
\begin{figure*}[!t]
\centering
\subfloat[]{\label{fig:seq-algo-f}\includegraphics[scale=1]{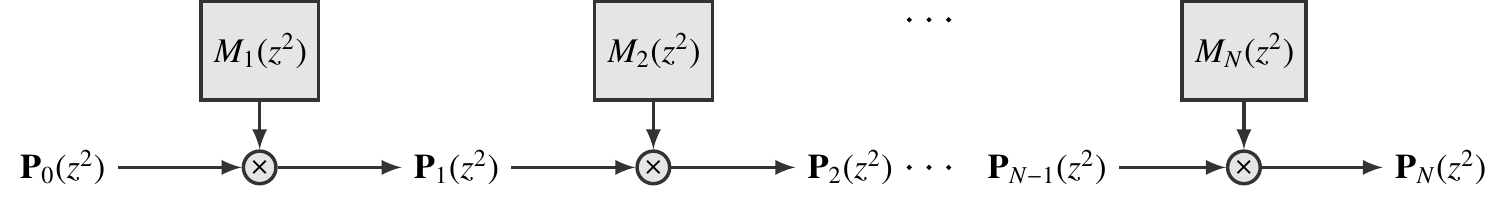}}\\
\subfloat[]{\label{fig:seq-algo-b}\includegraphics[scale=1]{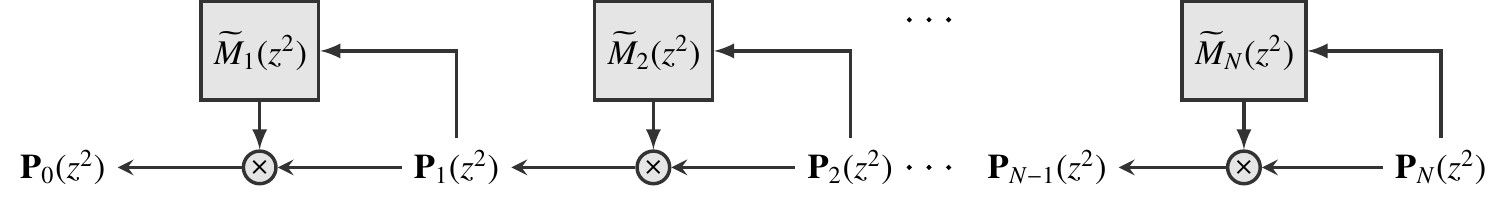}}
\caption{\label{fig:seq-algo} The figure depicts the sequential discrete forward
and inverse scattering algorithm in (a) and (b), respectively. The forward
scattering algorithm is identical to the transfer matrix approach used to solve
wave-propagation problems in dielectric layered media. The inverse scattering 
algorithm shown here is also 
known as the layer-peeling algorithm. It consists in using
$\vv{P}_{n+1}(z^2)$ to determine the transfer matrix $\wtilde{M}_{n+1}(z^2)$ 
so that the entire step can be repeated with $\vv{P}_{n}(z^2)$ as depicted in (b).}
\end{figure*}
\subsection{Inversion of discrete scattering coefficients}
\label{sec:dicrete-inft}
In this section, we consider the problem of recovering the discrete samples of the
scattering potential from the discrete scattering coefficients known in the
polynomial form. This step is referred to as the \emph{discrete inverse
scattering} step. Starting from the recurrence relation~\eqref{eq:poly-scatter},
we develop a layer-peeling algorithm similar to that reported by Brenne and
Skaar~\cite{BS2003}. The common aspect of the layer-peeling step for all kinds
of discretization schemes is that using nothing but the knowledge of
$\vv{P}_{n+1}(z^2)$, one should be able to retrieve the samples of the potential
needed to compute the transfer matrix $\wtilde{M}_{n+1}(z^2)$ so that the entire step can be
repeated with $\vv{P}_{n}(z^2)$ until all the samples of the potential are 
recovered (as illustrated in Fig.~\ref{fig:seq-algo-b}). In the following, we summarize the main 
results which facilitate the layer-peeling step corresponding to the each of the 
discretization schemes introduced so far. A detailed study of the recurrence 
relation and the proof of the necessary and sufficient conditions for 
discrete inverse scattering is provided in Sec.~\ref{sec:finite-support-seq}.

\subsubsection{Forward Euler method}
\label{sec:forward-Euler-summary}
The recurrence relation for the forward Euler method yields
\begin{equation}\label{eq:FE-cond}
P^{(n+1)}_{1,0} = 1.
\end{equation}
The layer-peeling algorithm based on the forward Euler method uses the relation
\begin{equation}
R_{n} = \frac{P^{(n+1)}_{2,1}}{P^{(n+1)}_{1,0}}, 
\end{equation}    
where $P^{(n+1)}_{1,0}\neq0$ on account of~\eqref{eq:FE-cond}. As 
evident from~\eqref{eq:scatter-FE}, the transfer matrix, $M_{n+1}(z^2)$, connecting 
$\vv{P}_n(z^2)$ and $\vv{P}_{n+1}(z^2)$ is therefore completely determined by
$R_{n}$ (with $Q_{n}=-R^*_{n}$).


\subsubsection{Implicit Euler method}
\label{sec:implicit-Euler-summary}
The recurrence relation for the implicit Euler method yields
\begin{equation}\label{eq:BDF1-cond}
P^{(n+1)}_{1,0} = \prod_{k=1}^{n+1}\Theta^{-1}_{k}>0,\quad\vv{P}^{(n+1)}_{n+1}=0.
\end{equation}
The layer-peeling algorithm based on the implicit Euler method uses the relation
\begin{equation}
R_{n+1} = \frac{P^{(n+1)}_{2,0}}{P^{(n+1)}_{1,0}}, 
\end{equation}    
where $P^{(n+1)}_{1,0}\neq0$ on account of~\eqref{eq:BDF1-cond}. As 
evident from~\eqref{eq:scatter-BDF1}, the transfer matrix, 
$\wtilde{M}_{n+1}(z^2)$, connecting $\vv{P}_n(z^2)$ and $\vv{P}_{n+1}(z^2)$ is 
therefore completely determined by $R_{n+1}$ (with $Q_{n+1}=-R^*_{n+1}$).

\subsubsection{Trapezoidal rule}
\label{sec:discrete-TR-summary}
Let us assume $Q_0=0$. The recurrence relation for the trapezoidal rule yields
\begin{equation}
\label{eq:TR-cond}
\begin{split}
P^{(n+1)}_{1,0}
&=\Theta^{-1}_{n+1}\prod_{k=1}^{n}\left(\frac{1+Q_kR_k}{1-Q_kR_k}\right)\\
&=\Theta^{-1}_{n+1}\prod_{k=1}^{n}\left(\frac{2-\Theta_k}{\Theta_k}\right),\\
\vv{P}^{(n+1)}_{n+1}
&= 0,
\end{split}
\end{equation}
where the last relationship follows from the assumption $Q_0=0$. For sufficiently 
small $h$, it is reasonable to assume that 
$1+Q_nR_n=2-\Theta_n>0$ so that $P^{(n)}_{1,0}>0$ (it also implies that 
$|Q_n|=|R_n|<1$). The layer-peeling algorithm based on the trapezoidal scheme
uses the relations
\begin{equation}
R_{n+1} = \frac{P^{(n+1)}_{2,0}}{P^{(n+1)}_{1,0}},\quad
R_n = \frac{\chi}{1 + \sqrt{1+|\chi|^2}},
\end{equation}
where
\begin{equation*}
\chi=\frac{P^{(n+1)}_{2,1} - R_{n+1}P^{(n+1)}_{1,1}}
{P^{(n+1)}_{1,0} - Q_{n+1}P^{(n+1)}_{2,0}}.
\end{equation*}
Note that $P^{(n+1)}_{1,0}\neq0$ and ${P^{(n+1)}_{1,0} -
Q_{n+1}P^{(n+1)}_{2,0}}\neq0$.
As evident from~\eqref{eq:scatter-TR}, the transfer matrix, $\wtilde{M}_{n+1}(z^2)$, 
connecting $\vv{P}_n(z^2)$ and $\vv{P}_{n+1}(z^2)$ is completely determined by
the samples $R_{n+1}$ and $R_n$ (with $Q_{n+1}=-R^*_{n+1}$ and $Q_{n}=-R^*_{n}$).

\subsection{Sequential algorithm}\label{sec:sequential}
\subsubsection{Forward scattering}
\label{sec:seq-scatter}
The computation of the Jost solution
for a given value of the spectral parameter, $\zeta\in\field{C}$ is considered
here as the forward scattering step. The direct use of the 
recurrence relations obtained in Sec.~\ref{sec:one-step-method} gives us a
sequential algorithm (see the illustration in Fig.~\ref{fig:seq-algo-f}). If 
$\varpi(n),\,n\in\field{Z}_+$, denotes the complexity of 
computing the Jost solution $\vv{P}_n(z^2)$ for a given $\zeta$, then $\varpi(n+1) =
4+\varpi(n)$, counting only the multiplications involved. This recurrence
relation yields $\varpi(N)=4N$. It must be noted that the sequential algorithms 
can be useful for computing norming constants as discussed
in Sec.~\ref{sec:eig-and-nconst} if the eigenvalues are known beforehand. If good 
initial guesses are known for the eigenvalues, search based methods such as
Newton's method of finding the eigenvalues can also benefit from sequential 
algorithms~\cite{WP2013b}.

The sequential algorithm for computing the polynomial coefficients of
$\vv{P}_N(z^2)$ can also be obtained in the same manner where transfer matrices
are now treated as polynomial matrices. If $\varpi(n)$ denotes the complexity of 
computing the polynomial coefficients for the Jost solution $\vv{P}_n(z^2)$, then
$\varpi(n+1) = 4(n+1)+\varpi(n)$, counting only the multiplications involved. This
yields $\varpi(N)=2(N+1)(N+2)=\bigO{N^2}$ which is extremely prohibitive for
large number of samples. This task can be accomplished much more efficiently 
using a divide-and-conquer strategy together with FFT-based fast polynomial 
arithmetic as described in Sec.~\ref{sec:FNFT}.

\subsubsection{Inverse scattering}
\label{sec:seq-invscatter}
The inverse scattering step here refers to the retrieval of the samples of the
scattering potential from the known polynomial form of the discrete scattering
coefficients. This can be accomplished by a sequential layer-peeling algorithm as
described in Sec.~\ref{sec:dicrete-inft} (see the illustration in
Fig.~\ref{fig:seq-algo-b}). If $\varpi(n),\,n\in\field{Z}_+$, 
denotes the complexity of inversion of $\vv{P}_n(z^2)$, then 
$\varpi(n) = 4(n+1)+\varpi(n-1)$ counting 
only the multiplications. This again yields a complexity of $\bigO{N^2}$ for
inverting $\vv{P}_N(z^2)$. This
task can also be accomplished much more efficiently using a divide-and-conquer strategy
together with FFT-based fast polynomial arithmetic as described in
Sec.~\ref{sec:fast-layer-peeling}.

\begin{figure*}[!t]
\centering
\includegraphics[scale=1]{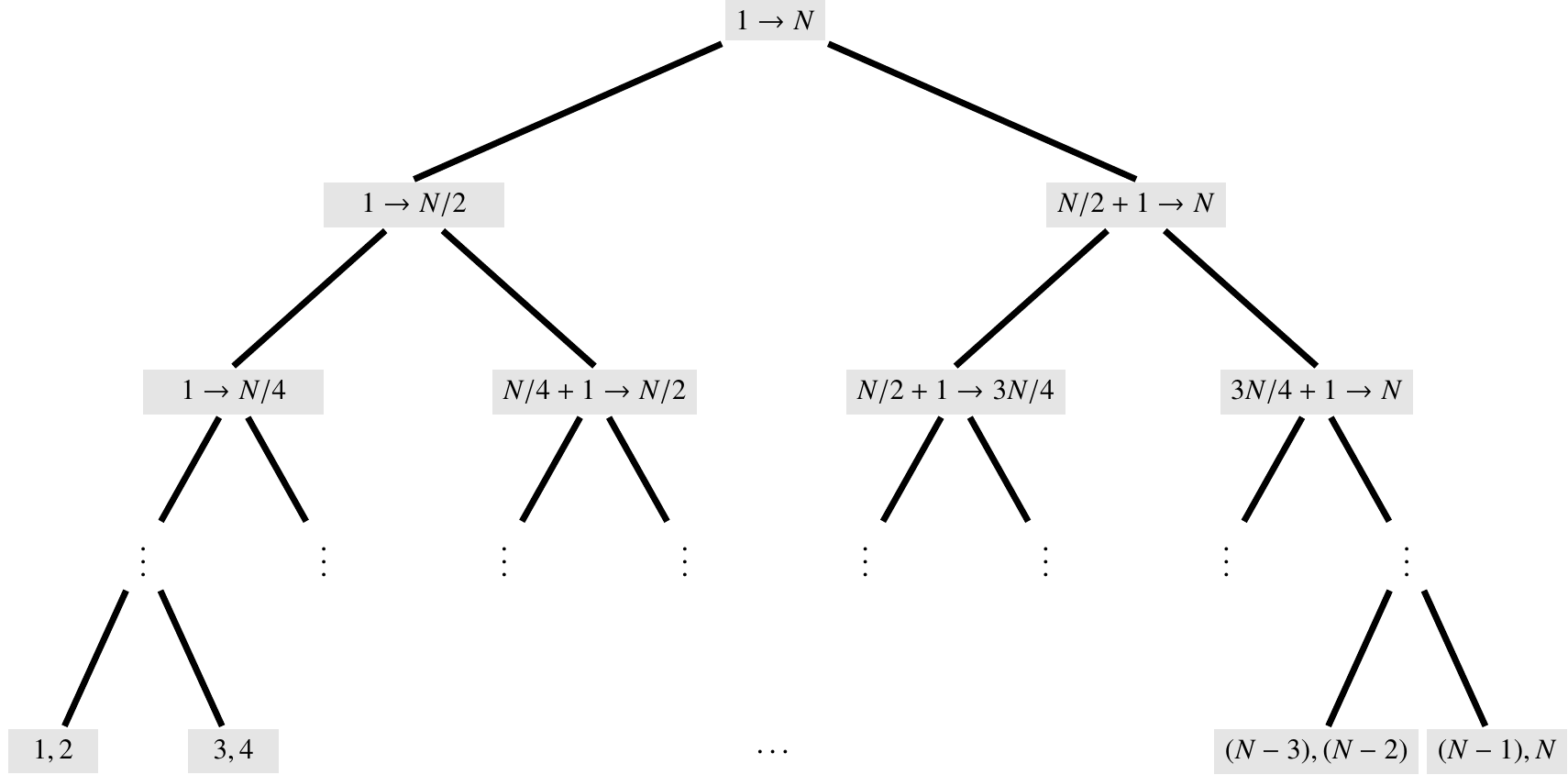}
    \caption{\label{fig:binary-tree}The figure shows the binary-tree structure
    obtained as a result of applying a
    divide-and-conquer strategy to the conventional layer-peeling method. The
    node label depicts the range of indices of the segments/layers ordered from
    left to right in the computational domain.}
\end{figure*}

\subsection{Fast algorithm: A divide-and-conquer strategy}\label{sec:FPA}
\subsubsection{Forward scattering}
\label{sec:FNFT}
The scattering algorithm consists in forming cumulative product of, say $N$,
transfer matrices. Given that the transfer matrices have polynomial entries 
(of maximum degree one), one can use FFT-based polynomial 
multiplication~\cite{Henrici1993} to obtain a fast forward scattering 
algorithm. In this article we restrict 
ourselves to the case where $N$ is a power of $2$. Most efficient use of the
FFT-based multiplication can be made if we use a divide-and-conquer strategy as
in~\cite{WP2013d} where products are formed pair-wise culminating in the full
transfer matrix. The complexity of obtaining the cumulative transfer matrix from $n$
transfer matrices, denoted by $\varpi(n)$, then satisfies the recurrence relation
\begin{equation}
    \varpi(n) = 8\nu(n)+2\varpi(n/2),
\end{equation}
where $\nu(n)=n(3\log_22n+2)$ is the complexity of multiplying two polynomials
of degree $n-1$ using the FFT algorithm. The number of pairs is given by
$l=\log_2N$ so that the recurrence relation yields
\[
\varpi(N) = 2^{l}\varpi(1) +
8\sum_{k=0}^{l-1}2^{k}\nu(2^{l-k}),
\]
which simplifies to 
\begin{equation}
    \varpi(N) = N\varpi(1) + 4N[3(\log_2N)^2+13\log_2N].
\end{equation}
Therefore, the complexity of the forward scattering algorithm is
$\bigO{N\log^2N}$. Note that $\varpi(1)$ denotes the cost of obtaining each of
the transfer matrices.

Evaluation of $\vv{P}_{N}(z^2)$ at an arbitrary complex point can be 
done using Horner's method~\cite[Chap.~3]{Henrici1964} which has the complexity of $\bigO{N}$. However, 
multipoint evaluation at $M\,\,(\geq N)$ Fourier
nodes can be carried out with complexity $\bigO{M\log M}$ where $M$ is a power
of $2$.

\subsubsection{Inverse scattering}
\label{sec:fast-layer-peeling}
In this section, we describe how to obtain a fast layer-peeling algorithm by
adapting McClary's approach~\cite{McClary1983} for our discrete inverse
scattering problem. Consider the grid $(x_n)_{0\leq n\leq N}$ and let us label 
the segment $[x_n,x_{n+1}]$ by $n+1$ for $n<N$. Recall that the inverse of the
transfer matrix $M_n(z^2)$ is $z^{-2}\wtilde{M}_n(z^2)$. The cumulative transfer
matrix from the $n$-th segment to the $(n-m+1)$-th segment is given by
\begin{multline}
z^{-2}\wtilde{M}_{n-m+1}(z^2)\times\ldots\times z^{-2}\wtilde{M}_{n-1}(z^2)\\
\times 
z^{-2}\wtilde{M}_{n}(z^2)=z^{-2m}\wtilde{M}_{n-m+1,\ldots,n-1,n}(z^2).
\end{multline}
Note that in order to determine the transfer matrices for last $l$ segments
starting from the $n$-th segment, it is sufficient to have a partial
knowledge of the Jost solution, more specifically\footnote{We discuss the case
where the underlying one-step method is the trapezoidal rule on account of the
fact that the corresponding transfer matrix is the most general among the 
methods considered in this article.}, 
$\{\vv{P}_n\}_{l+1}$, where $\{\cdot\}_l$ denotes truncation after 
first $l$ coefficients. Let the complimentary polynomial vector be defined as 
\begin{equation}
\{\vv{P}_n(z^2)\}^{c}_{l}= z^{-2l}\left(\vv{P}_n(z^2)-\{\vv{P}_n(z^2)\}_{l}\right),
\end{equation}
and consider the inverse propagation relation in terms of the inverse of the transfer
matrices:
\begin{multline}
\vv{P}_{n-m}(z^2)=z^{-2m}\wtilde{M}_{n-m+1,\ldots,n-1,n}(z^2)\times\\
\left(\{\vv{P}_n(z^2)\}_{l+1} + z^{2(l+1)}\{\vv{P}_n(z^2)\}^c_{l+1}\right).
\end{multline}
For every $m>0$, the first two
coefficients of the polynomial $\vv{P}_{n-m}(z^2)$ are required 
in order to determine the transfer matrix for the segment $n-m$;
therefore, $2(l+1-m)>0$ ensures that no contribution comes from the
complimentary polynomial in computing these first two coefficients. 
It then follows that the transfer matrices 
\[
\wtilde{M}_{n}(z^2),\,\wtilde{M}_{n-1}(z^2),\ldots,\,\wtilde{M}_{n-l+1}(z^2)
\]
can be determined without needing the complimentary polynomial
$\{\vv{P}_n(z^2)\}^c_{l+1}$. Once the matrices are determined, the Jost solution
needed to determine the transfer matrices for $n-l$ segments works out to be
\begin{equation}
\vv{P}_{n-l}(z^2)=z^{-2l}\wtilde{M}_{n-l+1,\ldots,n-1,n}(z^2)\vv{P}_n(z^2).
\end{equation}
All polynomial multiplications can be carried out using the FFT-algorithm.
The observations made above makes it clear that a divide-and-conquer strategy
can be easily devised in order to speed up the layer-peeling algorithm. For the 
inversion of the discrete scattering coefficients, we start with the associated
Jost solution $\vv{P}_{N}(z^2)$ 
where $N$ is a power of $2$, we devise a divide-and-conquer strategy that reduces the
original problem into two equal size (in terms of number of segments) 
subproblems\footnote{Note that the analysis in Sec.~\ref{sec:dicrete-inft} reveals 
that the number of coefficients associated with $\vv{P}_{N}(z^2)$ is 
exactly $N$.}. The algorithm can be described as follows:
\begin{enumerate}[label={\roman*.}]
\item Define a binary tree with the number of levels given by $l=\log_2N$ 
(see Fig.~\ref{fig:binary-tree}). Every parent node forks into two child nodes 
eventually terminating the tree at the leaf nodes.
\item Associate $N$ segments with the root node which is assumed to be at the level
zero. Number of segments associated with every child node is half of that of
the parent node. If $\mathcal{S}(k)$ denotes the number of segments
associated with nodes at the $k$-th level, then $\mathcal{S}(k) = N2^{-k}$ for
$k=0,1,\ldots,l-1$.
\item Every node in the binary tree is labeled by the index-coordinates $(j,k)$
where $k$ is the level and $j$ being the
horizontal position of the node from the left in any particular level,
say $k$, so that $0\leq j\leq k$. If the index of the last segment
associated with a given node $(j,k)$ is denoted by $N_{jk}$, then
$N_{jk}=2^j\mathcal{S}(k)$.
\item All polynomial products to be formed at any node at the $k$-th level requires
executing an FFT-algorithm for vectors of length no more than 
$2\mathcal{S}(k)$.
\item The segments associated with a node dictate the associated cumulative transfer 
matrix and the Jost solution (with the required number of coefficients) 
needed in order to determine the entries of constituting transfer matrices.
For the node $(j,k)$, the associated cumulative transfer matrix is
\[
z^{-2n}\wtilde{M}_{N_{jk}-n+1,\ldots,N_{jk}-1,N_{jk}}(z^2),\quad n=\mathcal{S}(k),
\]
and the associated Jost solution is $\{\vv{P}_{N_{jk}}(z^2)\}_{n+1}$.
\item Our algorithm requires exactly two types of operations to be carried out at
every node except for the leaf nodes. The first is the computation of the
cumulative transfer matrix once the constituting matrices are known at
the child nodes. The second is computing the Jost solution needed by 
any of the child nodes. Both of these operations boil down to polynomial
multiplications, therefore, it can be carried out efficiently using the
FFT-algorithm. The samples of the potential are determined at the
leaf nodes.
\end{enumerate}

Denoting the complexity of multiplying two polynomials of degree $n-1$ 
(via the FFT-algorithm) by $\nu(n)$, the recurrence relation for the 
complexity of the fast layer-peeling procedure, denoted by $\varpi(n)$ (where
$n=\mathcal{S}(k)$, the number of segments at level $k$), can be stated as
\begin{equation}
    \varpi(n) = 4\nu(n)+8\nu(n)+2\varpi(n/2).
\end{equation}
The first term on the RHS corresponds to the determination of Jost solution
for the second child node assuming that the Jost solution is known at the parent
node and the cumulative transfer matrix is known at the first child node. The 
second term corresponds to the determination of the cumulative transfer matrix 
at the corresponding parent node
using the transfer matrices of the child nodes. Observing
\[
\varpi(N) = 2^{l-1}\varpi(2) 
+12\sum_{k=0}^{l-2}2^{k}\nu(2^{l-k})
-8\nu(N),
\]
where the last term on RHS is a correction for the root node since the determination of
the cumulative transfer matrix at the root level is unnecessary. Using
$\nu(n)=n(3\log_22n+2)$, we have
\begin{multline}
    \varpi(N) = (N/2)\varpi(2)\\+ 6N[3(\log_2N)^2+13\log_2N-68/3],
\end{multline}
valid for $N\geq4$ where $\varpi(2)$ refers to the cost of executing the leaf
node. Therefore, the fast layer-peeling algorithm has the complexity of
$\bigO{N\log^2N}$. 
\subsection{Inversion of scattering coefficients}
\label{sec:invscattering-Lubich}
Let us assume that the scattering coefficients $a(\zeta)$ and $b(\zeta)$ are analytic 
in $\ovl{\field{C}}_+$ such that for $\zeta\in\ovl{\field{C}}_+$ and some $C>0$, we have
\begin{equation}\label{eq:growth-a-b}
|a(\zeta)-1|\leq \frac{C}{1+|\zeta|},\quad 
|\breve{b}(\zeta)|\leq \frac{C}{1+|\zeta|},
\end{equation}
where $\breve{b}(\zeta)=b(\zeta)e^{2i\zeta L_2}$. The precise conditions under
which such a situation may arise is discussed in theorems~\ref{thm:jost-estimate2}
and~\ref{thm:jost-estimate-exp}. We further assume that the potential
is supported in a domain of the form $(-\infty,L_2]$ or $[L_1,L_2]$. In this 
section, we would like to develop a method to compute the discrete scattering 
coefficients from the analytic form of the scattering coefficients so that the 
corresponding inverse problem can be solved numerically using the layer-peeling algorithm
discussed in Sec.~\ref{sec:dicrete-inft}. It turns out that this task can
be efficiently accomplished using the method developed by Lubich~\cite{Lubich1988I} 
which is used in computing the quadrature weights for convolution-type
integrals~\footnote{The method based on the trapezoidal rule also appears in control
literature where it is known as the Tustin's method~\cite{Tustin1947}.}.

Introduce the function $\delta(z)$ as in~\cite{Lubich1988I} which corresponds to
the A-stable one-step methods, namely, BDF1 and TR:
\begin{equation}
    \delta(z) = \begin{cases}
        (1-z)& (\text{BDF1}),\\
        2\frac{(1-z)}{1+z} &(\text{TR}).
    \end{cases}
\end{equation}
Putting $z=e^{i\zeta h}$, let us define the coefficients $a_k$ and $\breve{b}_k$ as
\begin{equation}
\begin{split}
&a\left(\frac{i\delta(z^2)}{2h}\right)= 1 +\sum_{k=0}^{\infty}a_kz^{2k},\\
&\breve{b}\left(\frac{i\delta(z^2)}{2h}\right)=\sum_{k=0}^{\infty}\breve{b}_kz^{2k}.
\end{split}
\end{equation}
The coefficients can be obtained using the Cauchy integrals
\begin{equation}\label{eq:discrete-coeffs}
\begin{split}
&a_{k} = \frac{1}{2\pi i}
\oint_{|z|=\varrho}\left[a\left(\frac{i\delta(z)}{2h}\right)-1\right]z^{-k-1}dz,\\
&\breve{b}_{k} = \frac{1}{2\pi i}
\oint_{|z|=\varrho}\left[\breve{b}\left(\frac{i\delta(z)}{2h}\right)\right]
z^{-k-1}dz,
\end{split}
\end{equation}
which can be easily computed using FFT. Note that the zeroth coefficient can be
computed exactly as
\begin{equation}\label{eq:discrete-coeffs-zero}
a_{0} = \left[a\left(\frac{i\delta(0)}{2h}\right)-1\right],\quad
\breve{b}_{0} = \left[\breve{b}\left(\frac{i\delta(0)}{2h}\right)\right].
\end{equation}
On account of the decay property of the scattering coefficients with respect to 
$\zeta$, $a_0=\bigO{h}$ and $\breve{b}_{0}=\bigO{h}$. 

Let $f_k(h)$ denote either $a_k$ or $\breve{b}_k$
and let $F(z^2)$ represent the corresponding integrand in~\eqref{eq:discrete-coeffs}. 
Following~\cite{Lubich1988II}, we obtain the approximation $f_k(h;M)$ for $f_k(h)$ as
\begin{equation*}
    f_k(h;M) = \frac{1}{M\varrho^k}\sum_{j=0}^{M-1}F_j e^{-i\frac{2\pi jk}{M}},  
\end{equation*}
where $F_j = F(\varrho e^{-i\frac{2\pi jk}{M}})$. Choosing $\varrho\leq1$
ensures that $\Im\zeta\geq0$. In order to
achieve an accuracy of $\bigO{\epsilon}$ for computing $f_k(h;M)$ for
$k=0,1,\ldots,N$ choose $\log\varrho=(1/N)\log\epsilon$ and
$M=N\log(1/\epsilon)$. The Lubich's method, therefore, delivers discrete 
scattering coefficients with $\bigO{M\log M}$
complexity excluding the cost of function evaluations.
\begin{rem}
If it is known that the scattering coefficients are also analytic in
$\field{C}_-$, say, in the strip 
$\field{S}_-(\mu)=\{\zeta\in\field{C}_-|\Im{\zeta}\geq-\mu\}$, then
Cauchy's estimate can be used to show that the Lubich coefficients decay
exponentially with $k$. Let
$\Gamma=\{z\in\field{C}|\,|z|=\varrho,\,\varrho>1\}$ 
be such that $[i\delta(z)/2h]\in\ovl{\field{C}}_+\cup\field{S}(\mu)$ for all
$z\in\Gamma$. Then, Cauchy's estimate gives
\[
|f_k(h)|\leq \varrho^{-k}\max_{z\in\Gamma}\left|f\left(\frac{i\delta(z)}{2h}\right)\right|,
\]
where $f(\zeta)$ stands for $a(\zeta)$ or $\breve{b}(\zeta)$ and $f_k(h)$
denotes the $k$-th Lubich coefficients.
\end{rem}

\subsubsection{Relationship with inverse Fourier-Laplace transform}
\label{sec:FL-transform-lubich}
In case of rational scattering coefficients, the Lubich coefficients $a_k$ and
$\breve{b}_k$ can be computed using the inverse Fourier-Laplace transform of the
scattering coefficients. For rational functions\footnote{It suffices for our
purpose to consider rational functions with simple poles 
(See~\ref{sec:DT-pure-soliton}).}, resolution into
partial fractions offers a straightforward means of computing inverse 
Fourier-Laplace transform. This property can be exploited to lower the cost of computing 
the discrete scattering coefficients as follows: Define the functions $\alpha(\tau)$ and 
$\breve{\beta}(\tau)$ as 
\begin{equation}\label{eq:integ-alpha-beta}
\begin{split}
&\alpha(\tau) = \frac{1}{2\pi}\int_{-\infty}^{\infty}[a(\zeta)-1]
e^{-i\zeta\tau}d\zeta,\\
&\breve{\beta}(\tau) = \frac{1}{2\pi}\int_{-\infty}^{\infty}\breve{b}(\zeta) e^{-i\zeta\tau}d\zeta.
\end{split}
\end{equation}
Note that for $\tau < 0$, the contour can be closed in $\field{C}_+$ and
the integrals would evaluate to zero, therefore $\alpha(\tau)$ and
$\breve{\beta}(\tau)$ are
causal. According to~\cite[Theorem 4.1]{Lubich1988I}, the coefficients $a_k$ and $\breve{b}_k$ 
approximate the quantities $(2h)\alpha(2hk)$ and $(2h)\breve{\beta}(2hk)$ up 
to $\bigO{h^{p+1}}$, respectively, for $k>0$ (note that the zeroth coefficient is given
by~\eqref{eq:discrete-coeffs-zero} which merely requires function evaluation).
For the trapezoidal rule, this property is proven in
Appendix~\ref{app:lubich-rational}. It is observed that agreement between true Lubich 
coefficients and those computed as stated above improves with increasing $k$.
Therefore, one should choose $k>N_{\text{th}}$ where $N_{\text{th}}>0$ is a
suitably chosen threshold in order to switch to the partial-fraction variant of
computing Lubich coefficients.

\subsection{Inversion of rational scattering coefficients: Truncated multi-solitons}
\label{sec:DT-pure-soliton}
In order to obtain a fast version of the Darboux transformations (DT) for 
generating multi-solitons (Problem~\ref{prob1}), we would like to employ the scattering coefficients 
obtained as a result of truncation of a $K$-soliton potential at $x=0$. As shown in 
Sec.~\ref{sec:trancated-soliton}, the scattering 
coefficients are rational functions of $\zeta$ with no poles in $\ovl{\field{C}}_+$. Therefore, 
the Lubich's method of obtaining discrete scattering coefficients as 
described in Sec.~\ref{sec:invscattering-Lubich} is also applicable here. It must be noted that in 
order to obtain the complete $K$-soliton potential at a given time $t$, the 
truncation must be done after computing the time-evolved Darboux matrix.

Discrete inverse scattering proceeds by computing the polynomial vector
$\vv{P}_N(z^2)$ associated with the discrete scattering coefficients. Without
the loss of generality, we assume that the truncation is done at $x=0$ (see
Remark~\ref{rem:truncation}). Let the discrete spectrum of the $K$-soliton be
$\mathfrak{S}_K$ as defined in Sec.~\ref{sec:scattering-data}. Using the notations introduced in
Sec.~\ref{sec:Darboux-degree-one} (we drop the dependence of the Darboux matrices on
$\mathfrak{S}_K$ for the sake of brevity) and setting $N_1=N/2\in\field{Z}_+$, for the 
left-sided profile, we have
\begin{align*}
P_1^{(N_1)}(z^2) &=
\left\{\mu_{K}\left(\frac{i\delta(z^2)}{2h}\right)
\left[D_K\left(0,t;\frac{i\delta(z^2)}{2h}\right)\right]_{11}\right\}_{N_1},\\
P_2^{(N_1)}(z^2) &= \left\{\mu_{K}\left(\frac{i\delta(z^2)}{2h}\right)
\left[D_K\left(0,t;\frac{i\delta(z^2)}{2h}\right)\right]_{21}\right\}_{N_1},
\end{align*}
where truncation after $N_1$ terms is implied by the notation $\{\cdot\}_{N_1}$. This determines 
$U(x)$ for $x<0$. The right-sided profile can be generated using the 
transformation described in Remark~\ref{rem:reflection} so that
\begin{align*}
P_1^{(N_1)}(z^2) &=
\left\{\mu_{K}\left(\frac{i\delta(z^2)}{2h}\right)
\left[D_K\left(0,t;\frac{i\delta(z^2)}{2h}\right)\right]_{22}\right\}_{N_1},\\
P_2^{(N_1)}(z^2) &= \left\{\mu_{K}\left(\frac{i\delta(z^2)}{2h}\right)
\left[D_K\left(0,t;\frac{i\delta(z^2)}{2h}\right)\right]_{12}\right\}_{N_1},
\end{align*}
This would determine $U^*(-x)$ for $x<0$. Combining the two parts determines the 
complete multi-soliton potential. Note that the foregoing description also applies to
any set of rational functions which qualify as scattering coefficients of a
left-sided or a right-sided profile, respectively. 

The operational complexity of
this algorithm can be computed by taking into account the complexity of DT at
$x=0$, which is $\bigO{K^2}$, and the complexity of computation of Lubich
coefficients which is $\bigO{KM}+\bigO{M\log M}$ where $M$ is the number of
nodes used in evaluating the Cauchy integral. Given that $K\ll M$ and
$M=\bigO{N}$, the overall complexity of generating the multi-soliton including
the layer-peeling step works out to be $\bigO{N(K+\log^2N)}$. The algorithm
presented in this section is referred to as the 
\emph{fast Darboux transformation}~(FDT) algorithm. As pointed out
in Sec.~\ref{sec:Darboux-degree-one}, the CDT
algorithm offers machine precision for computing $K$-soliton potentials with an
operational complexity of $\bigO{K^2N}$. The fundamental difference between the
CDT and the FDT algorithm is depicted in Fig.~\ref{fig:general-DT} where it is
evident that by avoiding DT-iterations at each of the grid points (except at $x=0$) and 
using the fast LP algorithm, a lower complexity order algorithm can be obtained.

For any rational function, if the poles and residues are known then resolution
into partial fractions offers a straightforward means of computing the inverse
Fourier-Laplace transform. Let us apply this idea to the problem of generating
multi-solitons as discussed in the last paragraph: Poles of the Jost solutions
are known to be $\zeta_k^*$ (where $\zeta_k$ are the discrete eigenvalues), 
therefore, the resolution of the Darboux matrix into partial fractions reads as
\begin{multline}\label{eq:darboux-pf}
\mu_K(\zeta)D_K(0,t;\zeta)\\
= \sigma_0 + 
\sum_{k=1}^K\frac{\Res[\mu_K,\zeta_k^*]}{\zeta-\zeta_k^*}D_K(0,t;\zeta^*_k).
\end{multline}
The inversion of $(\zeta-\zeta_k^*)^{-1}$ leads to terms of the form
$-ie^{-i\zeta_k^*\tau}$, therefore, the quantities $e^{-2ih\zeta_k^*}$ must be
computed beforehand. Excluding the cost of computing the $K$ exponentials, the
complexity of this algorithm is $\bigO{KN}$ where $N$ is the number of samples
in the $\tau$-domain. In practice, replacing Lubich coefficients with that
obtained by resolution into partial-fractions leads to increase in error and
even failure to converge; however, for larger values of the index, the agreement 
between the two improves allowing us to reduce the overall complexity of
computing the discrete coefficients $\vv{P}^{(N_1)}_{k}$ by switching to the 
faster algorithm for $k>N_{\text{th}}$ where $N_{\text{th}}>0$ is a suitably 
chosen threshold\footnote{A recipe to choose $N_{\text{th}}$ based on the number
of samples $N$, the size of the computational domain $(L_2-L_1)$ and the eigenvalue with the
smallest imaginary part is provided in the 
Appendix~\ref{app:lubich-rational}.}. 

Before we conclude this section, it is worth mentioning that the case treated by
Rourke~\et~\cite{RM1992,RS1994} of rational reflection coefficient 
$\rho(\zeta)$ proceeds by reducing 
the problem to an equivalent problem of generating multi-solitons on a given 
half-space. Therefore, such cases are amenable to the method discussed in this 
article. 
\begin{figure*}[!t]
\centering
\subfloat[]
{\includegraphics[scale=1]{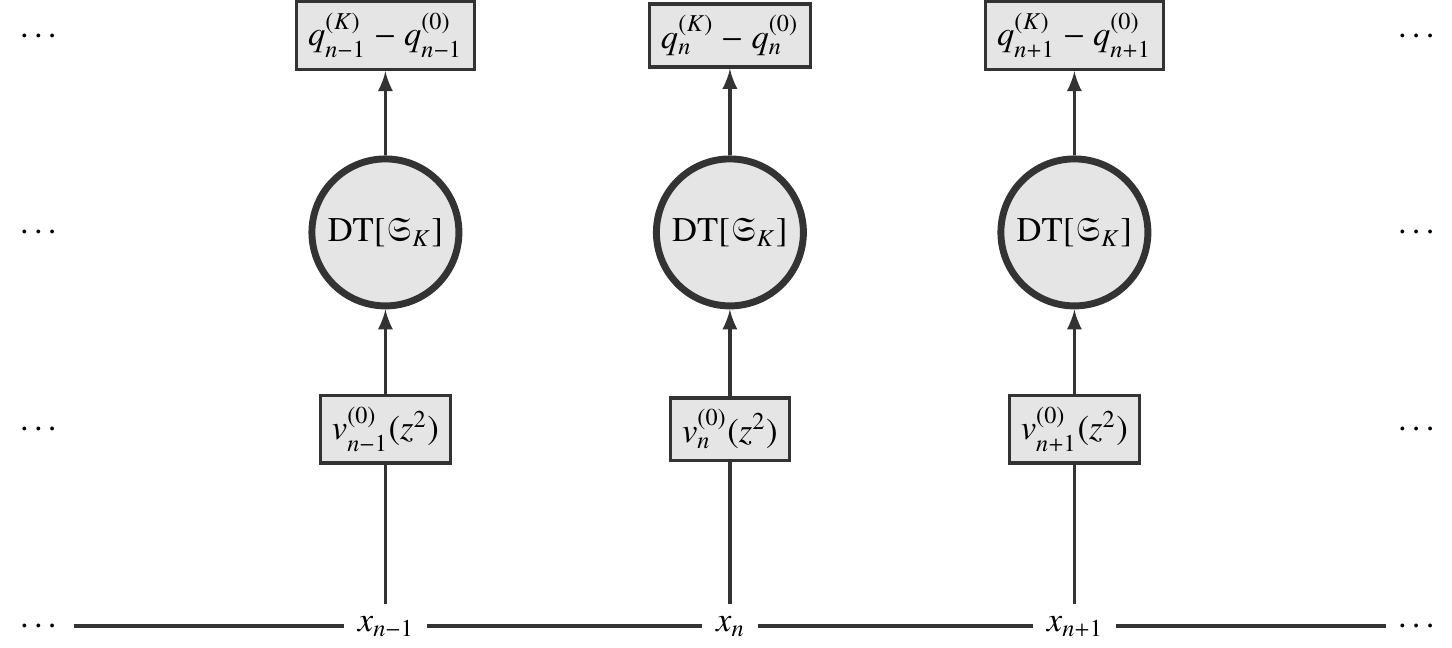}%
\label{fig:CDT-schema}}\\
\subfloat[]
{\includegraphics[scale=1]{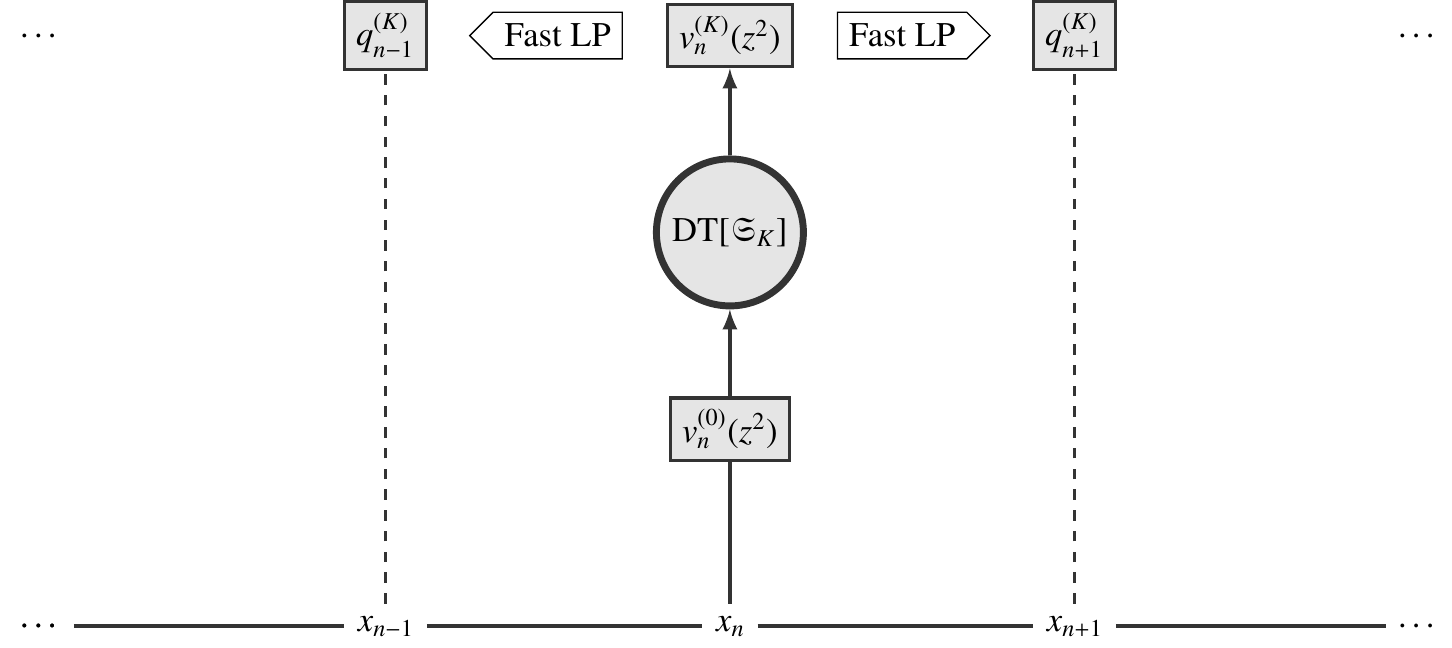}%
\label{fig:FDT-schema}}
\caption{The figure depicts the CDT and the FDT algorithm in (a) and (b),
respectively, for computing the
augmented potential, $q^{(K)}_n$ (which is the discrete approximation to
$q_K(x_n)$), by adding a given discrete
spectrum, $\mathfrak{S}_K$, to the seed potential, $q^{(0)}_n=q_0(x_n)$. The
DT-block, labeled as DT$[\mathfrak{S}_K]$, is
    described in Fig.~\ref{fig:DT-block}. The quantities $v^{(j)}_n(z^2)$ refer to 
the discrete approximation of $v_j(x_n,0;\zeta)$ in Fig.~\ref{fig:DT-block}.}
\label{fig:general-DT}
\end{figure*}
\subsection{General Darboux transformation: Addition of bound sates}
\label{sec:general-DT}
In this section, we address Problem~\ref{prob2} introduced in the beginning of
this article. To this end, let us note that 
the general Darboux transformation consists in adding a given discrete spectrum 
$\mathfrak{S}_K$ (as defined in Sec.~\ref{sec:scattering-data}) to a given seed
potential, $q_{\text{seed}} = q_0(x)$, which is assumed to be admissible as a 
scattering potential in the ZS-problem. The two algorithms developed for this 
purpose, namely, the classical Darboux transformation (CDT) and the fast Darboux 
transformation (FDT) meant to carry out the general Darboux transformation are 
described in the following subsections. For the sake of brevity of presentation, 
we restrict ourselves to the case $t=0$. 

\subsubsection{The CDT algorithm} 
\label{sec:CDT-general}
The basic idea behind the CDT algorithm is described in Sec.~\ref{sec:Darboux-degree-one} 
and also depicted in Fig.~\ref{fig:CDT-schema}. In the discrete framework developed in 
Sec.~\ref{sec:one-step-method}, the seed Jost solutions (which need to be
evaluated at the eigenvalues $\zeta_j$ to be added) can be computed via 
the sequential algorithm discussed in Sec.~\ref{sec:seq-scatter}. Using the notations
introduced in Sec.~\ref{sec:Darboux-degree-one} and Sec.~\ref{sec:eig-and-nconst}, and, 
introducing $\beta^{(j-1)}_n(z_j)$ as the discrete approximation to 
$\beta_{j-1}(x_n,0;\zeta_j)$, we have
\begin{equation}
\beta^{(j-1)}_n(z_j) 
    = \frac{P^{(n,j-1)}_1(z_j^2)-(z_j^2)^{-\ell_-+n}b_{j}S^{(m,j-1)}_1(z_j^2)}
{P^{(n,j-1)}_2(z_j^2) - (z_j^2)^{-\ell_-+n}b_{j}S^{(m,j-1)}_2(z_j^2)},
\end{equation}
where $(\zeta_j,b_j)\in\mathfrak{S}_K$, $m+n=N$ and $z_j=e^{2i\zeta_jh}$. Noting that
$v^{(j)}_n=(\vv{P}^{(j)}_{n}, \vv{S}^{(j)}_{m})$, the rest of the steps involved
are similar to that discussed in Sec.~\ref{sec:Darboux-degree-one}. 

The operational complexity of computing the
seed Jost solutions at $K$ eigenvalues using the sequential algorithm is $\bigO{KN}$ so 
that the overall complexity of the CDT algorithm is $\bigO{K^2N}$. A final remark that we would
like to make with regard to the CDT algorithm is that numerical computation of
the Jost solutions for complex values of the spectral parameter $\zeta$ tends to become
inaccurate on account of the $\zeta$-dependence of the truncation error
coefficient as discussed in Sec.~\ref{sec:one-step-method}. It is therefore
recommended that $\Im{\zeta_k}$ is kept below a certain threshold.

\subsubsection{The FDT algorithm}
\label{sec:FDT-general}
The fundamental idea of the FDT algorithm is the same as that described in
Sec.~\ref{sec:DT-pure-soliton} which is considers the problem of adding
bound states, described by $\mathfrak{S}_K$, to a null seed potential. The
difference merely lies in how we compute the seed Jost solutions required in the
DT-iterations at $x=0$ for a general seed potential. Following
Sec.~\ref{sec:Darboux-degree-one} and ~\ref{sec:eig-and-nconst},
note that evaluation of the Jost solutions at $\zeta=\zeta_j$ amounts to evaluating 
the approximating polynomial at $z_j=e^{2i\zeta_jh}$ (setting $x=0$ and $t=0$),
so that the recursive step for computing the $\beta$-coefficients reads as
\begin{equation}
\beta^{(j-1)}_{\ell_-}(z_j) 
= \frac{P^{(\ell_-,j-1)}_1(z_j^2)-b_{j}S^{(\ell_+,j-1)}_1(z_j^2)}
{P^{(\ell_-,j-1)}_2(z_j^2) - b_{j}S^{(\ell_+,j-1)}_2(z_j^2)}.
\end{equation}
where we have assumed $\ell_-,\,\ell_+\in\field{Z}$ for simplicity
and $(\zeta_j,b_j)\in\mathfrak{S}_K$. Noting that
$v^{(j)}_{\ell_-}=(\vv{P}^{(j)}_{\ell_-}, \vv{S}^{(j)}_{\ell_+})$, other steps of the
iteration are identical to that described in Sec.~\ref{sec:Darboux-degree-one}. Here, our
objective is not to follow the conventional Darboux transformation but merely obtain
the truncated scattering coefficients (for the left-sided and the right-sided
potential) at the origin so that a fast layer-peeling algorithm can 
be used to compute the samples of the augmented potential.  

The operational complexity of this algorithm can be worked out as follows: The cost of 
computing the Jost solutions (as a polynomial vector) is $\bigO{N\log^2N}$ and the
cost of evaluation of the Jost solutions using Horner's scheme is $\bigO{N}$ 
for each of the eigenvalues so that the overall complexity of computing the discrete truncated 
scattering coefficient at $x=0$ is $\bigO{K^2}+\bigO{KM}+\bigO{M\log M} + \bigO{N\log^2N}$ 
where $M$ is the number of nodes used in evaluating the Cauchy integral, $N$ is the 
number of samples of the potential and $K$ is the number of eigenvalues to be
added. Observing that $K\ll M$, $N$ and $M=\bigO{N}$, the overall complexity is
effectively $\bigO{N(K+\log^2N)}$ including the layer-peeling step. 

The convergence behavior of the FDT algorithm is studied in the 
Sec.~\ref{sec:accuracy-DT} where it is shown that 
the Darboux matrices can be computed with the same order of accuracy as that of 
the underlying one-step method used in the computation of the Jost solutions of 
the seed potential. Further, the global order of convergence matches that of the 
underlying one-step method for the computation of Lubich coefficients or the 
layer-peeling algorithm depending on which of the two is lower.

Finally, let us conclude this section by pointing out that 
if a fast and sufficiently accurate means of inversion of continuous
spectrum (i.e., no bound states present) is available 
then a fast inverse scattering algorithm can be easily obtained 
for the general cases using the 
FDT algorithm outlined in this section. The first results in this direction are
reported in~\cite{VW2017OFC} where the trapezoidal rule is used to develop two
algorithms of complexity $\bigO{N(K+\log^2N)}$ that exhibit a convergence
behavior of $\bigO{N^{-2}}$.

\section{Benchmarking methods}\label{sec:benchmark}
In this section, we discuss two of the conventional methods which are widely used for 
solving scattering problems. We would like to benchmark our method against these
known methods. Unlike the linear one-step methods, here we employ a staggered grid 
configuration given by $(x_{n+1/2})_{0\leq n<N}$ such that 
$x_{n+1/2}=x_n+h/2$. 

\subsection{Magnus integrator}
By applying the Magnus method with one-point
Gaussian quadrature (see~\cite{Magnus1954,IN1999,HL2003})
to the original ZS-problem in~\eqref{eq:zs-prob}, we
obtain
\begin{equation}
\vv{v}_{n+1}=e^{-i\zeta\sigma_3h+U_{n+1/2}h}\vv{v}_n.
\end{equation}
The exponential operator can be computed exactly as
\begin{multline}
    e^{-i\zeta\sigma_3h+U_{n+1/2}h} = \\
\begin{pmatrix}
\cosh(\Gamma) - \frac{i \zeta h }{\Gamma}\sinh(\Gamma) 
& \frac{Q_{n+1/2}}{\Gamma}\sinh(\Gamma)\\
\frac{R_{n+1/2}}{\Gamma}\sinh(\Gamma) 
& \cosh(\Gamma)+\frac{i \zeta h }{\Gamma}\sinh(\Gamma)\\
\end{pmatrix},
\end{multline}
where $\Gamma = \sqrt{Q_{n+1/2}R_{n+1/2} - \zeta^2h^2}$ where 
$Q_{n+1/2}=hq(x_n+h/2,t)$ and $R_{n+1/2}=hr(x_n+h/2,t)$. We refer to this
integrator as ``MG1'' signifying Magnus integrator with one-point Gauss
quadrature. This method is also referred to as the 
exponential mid-point rule in the literature and it can be shown to be
consistent and stable with an order $p=2$. Additionally, it also forms the part 
of the Lie-group methods~\cite{IN1999,HLW2006} as it retains the SU$(2)$ 
structure of the Jost solution $v = (\vs{\phi},\ovl{\vs{\phi}})$ for
$\zeta\in\field{R}$. It must be noted 
that this method is specially suited for highly oscillatory problems and has been employed 
by several authors to solve forward scattering problems~\cite{BO1992,BCT1998}.
Finally, let us mention that the method of computing the norming constants as
described in Sec.~\ref{sec:eig-and-nconst} can also be adapted to MG1.

\subsection{Split-Magnus method}
\label{sec:split-Magnus}
A further simplification obtained by applying Strang-type 
splitting~\cite{Strang1968} to the exponential operator provides the right discrete
framework for the layer-peeling algorithm. This simplification is achieved as
follows: 
\begin{multline}
e^{(-i\zeta\sigma_3+U_{n+1/2})h} = 
e^{-i\zeta\sigma_3h/2}
e^{U_{n+1/2}h}e^{-i\zeta\sigma_3h/2}\\
+ \bigO{h^3}.
\end{multline}
The order of approximation is determined by applying the
Baker-Campbell-Hausdorff (BCH) formula to the exponential 
operators~\cite[Chapter 4]{HLW2006}. Setting $\Gamma =
\sqrt{Q_{n+1/2}R_{n+1/2}}$, we have
\begin{align*}
e^{U_{n+1/2}h} &= 
\begin{pmatrix}
\cosh\Gamma & Q_{n+1/2}\frac{\sinh\Gamma}{\Gamma}\\
R_{n+1/2}\frac{\sinh\Gamma}{\Gamma} & \cosh\Gamma\\
\end{pmatrix}\\
    &=\frac{1}{\sqrt{1-\tanh^2\Gamma}}\begin{pmatrix}
1 & Q_{n+1/2}\frac{\tanh\Gamma}{\Gamma}\\
R_{n+1/2}\frac{\tanh\Gamma}{\Gamma} & 1\\
\end{pmatrix}\\
&=\frac{1}{\sqrt{1-\Gamma^2}}
\begin{pmatrix}
1& {Q_{n+1/2}}\\
{R_{n+1/2}} & 1\\
\end{pmatrix}+\bigO{h^3}.
\end{align*}
Therefore, the discretization scheme works out to be
\begin{equation}
\vv{v}_{n+1}=\frac{1}{\Theta_{n+1/2}^{1/2}}
\begin{pmatrix}
z^{-1}& {Q_{n+1/2}}\\
{R_{n+1/2}} & z\\
\end{pmatrix}\vv{v}_{n},
\end{equation}
where $\Theta_{n+1/2}=(1-Q_{n+1/2}R_{n+1/2})>0$. This form 
has been used by a number of authors in connection with the
conventional layer-peeling algorithm~\cite{BLK1985,FZ2000,BS2003} as well as for
the fast version of the layer-peeling algorithm~\cite{WP2013d,WP2015}. By 
employing the transformation $\vv{w}_{n} = e^{i\zeta\sigma_3h/2}\vv{v}_n$,
we obtain
\begin{equation}
\begin{split}
\vv{w}_{n+1}&=\frac{z^{-1}}{\Theta^{1/2}_{n+1/2}}
\begin{pmatrix}
1& z^2Q_{n+1/2}\\
R_{n+1/2}& z^2
\end{pmatrix}\vv{w}_n\\
&=z^{-1}M_{n+1}(z^2)\vv{w}_n,
\end{split}
\end{equation}
which maybe viewed as a modification of the implicit Euler scheme. The
integration scheme thus obtained is referred to as the 
split-Magnus (SM) method. The inverse relationship is given by
\begin{equation}
    \frac{z^{-1}}{\Theta^{1/2}_{n+1/2}}
\begin{pmatrix}
z^2& -z^2Q_{n+1/2}\\
-R_{n+1/2}&1
\end{pmatrix}\vv{w}_{n+1} = \vv{w}_n.
\end{equation}
The Jost solution can be put in to the form
\begin{equation}
\begin{split}
&\vs{\psi}_n=z^{\ell_{+}} z^{-m}
\begin{pmatrix}
z^{-1}&0\\
0&1
\end{pmatrix}
\vv{S}_m(z^2),\\
&\vs{\phi}_n=z^{\ell_{-}} z^{-n}
\begin{pmatrix}
    1&0\\
    0&z
\end{pmatrix}
\vv{P}_n(z^2),
\end{split}
\end{equation}
where $\vv{S}_m(z^2)$ and $\vv{P}_n(z^2)$ obey the same kind of transfer matrix 
relation as in~\eqref{eq:poly-scatter} with initial condition 
$\vv{S}_0=(0,1)^{\tp}$ and $\vv{P}_0=(1,0)^{\tp}$. The scattering coefficients 
work out to be
\begin{equation}\label{eq:scatter-coeffs-magnus}
\begin{split}
&a_{N}(z^2)={P}^{(N)}_1(z^2),\quad b_{N}(z)=z^{-2\ell_{+}+1}{P}^{(N)}_2(z^2),\\
&a_{N}(z^2)={S}^{(N)}_2(z^2),\quad 
    \ovl{b}_{N}(z)=z^{-2\ell_{-}-1}{S}^{(N)}_1(z^2).
\end{split}
\end{equation}
The layer-peeling property can be stated as 
\begin{equation}
R_{n+1/2} =\frac{P^{(n+1)}_{2,0}}{P^{(n+1)}_{1,0}},
\end{equation} 
with the following additional constraints:
\begin{equation}\label{eq:cond-Magnus}
    P^{(n+1)}_{1,0}   = \prod_{k=1}^{n+1}\Theta^{-1/2}_{k}>0,\,\,
    \vv{P}^{(n+1)}_{n+1} =0.
\end{equation}
The norming constants can be computed using any of the following formulas
\begin{equation}
\begin{split}
&b_k= (z_k^2)^{\ell_- - n}\frac{z_kP^{(n)}_2(z^2_k)}{S^{(m)}_2(z^2_k)},\\
&\frac{1}{b_k}= (z_k^2)^{\ell_+ - m}\frac{S^{(m)}_1(z^2_k)}{z_kP^{(n)}_1(z^2_k)}.
\end{split}
\end{equation}
Lastly, we note that a staggered grid configuration may prove
superior for potentials with jump discontinuity at any grid point because the
sampling of the potential at the points of discontinuity is avoided.

\begin{rem}
It must be noted that the CDT and the FDT algorithms are incompatible with 
the staggered grid configuration, therefore, the SM integrator is ruled out 
for all DT-related algorithms.
\end{rem}

\section{Discrete inverse scattering: Necessary and sufficient condition}
\label{sec:finite-support-seq}
In this section, we study the necessary and sufficient condition for the
inversion of the discrete scattering coefficients within the framework of the
numerical discretization introduced in Sec.~\ref{sec:differential-formulation}. 
Let $\{(\cdot)_{k}\}_{k=1}^N$ denote a sequence of quantities such as scalars, 
vectors or matrices.

\begin{defn} 
Let $d$ be a non-negative integer. A polynomial $\vv{P}_{n}(z)$ defined as 
in~\eqref{eq:poly-vec} (with coefficients 
$\vv{P}^{(n)}_k\in\field{C}^2,\,k=0,1,\ldots,n$) is said to belong to the 
class $\fs{P}(d;\field{C}^2)$ if $\deg[\vv{P}_{n}(z)]\leq d$ and, for all 
$z\in\field{T}$, we have
\begin{equation}\label{eq:poly-class-cond}
\vv{P}_{n}(z)\cdot\vv{P}_{n}^*(z)=1, 
\end{equation}
with ${P}_{1,0}^{(n)}\in\field{R}_+$.
\end{defn}
For any $\vv{P}_{n}(z)\in\fs{P}(d;\field{C}^2)$, on equating the coefficient of
the zeroth degree term on LHS and RHS of~\eqref{eq:poly-class-cond}, we obtain
\[
\sum_{k=0}^n \vv{P}^{(n)}_k\cdot\vv{P}^{(n)*}_k
=\sum_{k=0}^n \left[\left|{P}^{(n)}_{1,k}\right|^2+\left|{P}^{(n)}_{2,k}\right|^2\right]=1.    
\]
Therefore, $|{P}^{(n)}_{1,k}|\leq1$ and $|{P}^{(n)}_{2,k}|\leq1$ for
$k=0,1,\ldots,n$. Note that the condition ${P}_{1,0}^{(n)}\in\field{R}_+$
ensures that there are no constant phase factors in $\vv{P}_n(z)$ because 
the relation~\eqref{eq:poly-class-cond} is insensitive to constant phase factors.
\begin{defn}[Para-conjugate] For any scalar 
valued complex function, ${f}(z)$, we define $\ovl{f}(z)= f^*(1/z^*)$. For any vector 
valued complex function, $\vv{f}(z)=(f_1(z),f_2(z))^{\tp}$, we define 
\[
\ovl{\vv{f}}(z)= i\sigma_2\vv{f}^*(1/z^*)=
\begin{pmatrix}
\ovl{f}_2(z),\\
-\ovl{f}_1(z)
\end{pmatrix}.
\] 
For a matrix valued function, $M(z)$, we define 
\[
\ovl{M}(z)=i\sigma_2M^*(1/z^*)(i\sigma_2)^{\dagger}
=\sigma_2M^*(1/z^*)\sigma_2,
\] 
so that the operation $\ovl{(\cdot)}$ is distributive over matrix-vector and matrix-matrix
products.
\end{defn}

Based on the discrete formulation of the ZS-problem in 
Sec.~\ref{sec:differential-formulation}, we identified a discrete 
representation of the Jost solution which can also be stated in the form
(leaving out the factors independent of $n$)
\begin{equation}\label{eq:jost-sol}
w_n=\left(z^{-n}\vv{P}_n(z^2),z^{n}\ovl{\vv{P}}_n(z^2)\right),
\end{equation}
such that the column vectors are linearly independent for all $z\in\field{C}$.
This implies, $\det[w_n]\neq0$. In fact, the determinant must turn out
to be independent of $z^2$ so that we may put $\det[w_n]=W_n$ which translates
into the constraint\footnote{Given that $\det[w_n]$ is a polynomial, the
only way $\det[w_n]\neq0$ is when it is a polynomial of degree zero.}
\[
\det\left[\vv{P}_n(z^2),\,\ovl{\vv{P}}_n(z^2)\right]
=-\vv{P}_n(z^2)\cdot\vv{P}^*_n(1/z^{*2}) = W_n.
\]
For $z\in\field{T}$, 
\begin{equation}\label{eq:jost-cond1}
\vv{P}_n(z^2)\cdot\vv{P}^*_n(z^2) = -W_n>0.
\end{equation}
This condition is necessary for $w_n$, defined
by~\eqref{eq:jost-sol}, to be a Jost solution. Further, it is easy to verify
that $w_n$ satisfies the relation 
\begin{equation}
\ovl{w}_n=\sigma_2\left(z^{n}\ovl{\vv{P}}_n(z),z^{-n}{\vv{P}}_n(z^2)\right)\sigma_2
=-w_n.
\end{equation}
Finally, let us note
that $\tilde{w}_n = w_n/{\sqrt{-W_n}}$ forms a $\text{SU}(2)$-valued sequence for
$z\in\field{T}$. 

The discrete scattering problem will be assumed to be stated in the form
of a recurrence relation which reads as
\begin{equation}\label{eq:recurrence-jost-matrix}
w_{n+1} = z^{-1}M_{n+1}(z^2){w}_n,
\end{equation}
where $M_n(z^2)$ is a polynomial matrix of degree one. Note that $w_n$ as
defined by~\eqref{eq:jost-sol} satisfies the relation $\ovl{w}_n=-w_n$; therefore, 
in order that $w_{n+1}$ be a Jost solution, we must have 
$z^{-1}M_{n+1}(z^2)=z\ovl{M}_{n+1}(z^2)$. This relationship expands to
\begin{equation}\label{eq:tmat-sym}
\begin{pmatrix}
    m^{(n)}_{11}(z^2)& {m}^{(n)}_{12}(z^2)\\
    m^{(n)}_{21}(z^2)& {m}^{(n)}_{22}(z^2)
\end{pmatrix}
=z^2\begin{pmatrix}
    \ovl{m}^{(n)}_{22}(z^2) & -\ovl{m}^{(n)}_{21}(z^2)\\
   -\ovl{m}^{(n)}_{12}(z^2) &  \ovl{m}^{(n)}_{11}(z^2)
\end{pmatrix},
\end{equation}
and $\det[M_n(z^2)] = z^2C_n$ where $C_n$ is independent of $z$. 
Introducing the functions 
\begin{equation}
m^{(n)}_j(z^2)=m^{(n)}_{j,0}+m^{(n)}_{j,1}z^2,\quad j=1,2,
\end{equation}
it follows that the general form of the transfer matrix 
(of degree one in $z^2$) can be written as
\begin{equation}
    M_n(z^2)=
\begin{pmatrix}
    m^{(n)}_{1}(z^2) & -z^2\ovl{m}^{(n)}_{2}(z^2)\\
    m^{(n)}_{2}(z^2) & z^2\ovl{m}^{(n)}_{1}(z^2)
\end{pmatrix},
\end{equation}
with 
\begin{equation}
\begin{split}
&C_n=|m^{(n)}_{1,0}|^2 + |m^{(n)}_{2,0}|^2 + |m^{(n)}_{1,1}|^2+|m^{(n)}_{2,1}|^2,\\
&m^{(n)}_{1,0}m^{(n)*}_{1,1}+m^{(n)}_{2,0}m^{(n)*}_{2,1}=0.
\end{split}
\end{equation}
Let the inverse of $M_n(z^2)$ be denoted by $z^{-2}\wtilde{M}_n(z^2)$ which 
also satisfies a similar symmetry relation as in~\eqref{eq:tmat-sym} and 
\begin{equation*}
z^{-2}\wtilde{M}_n(z^2)=\frac{z^{-2}}{C_n}
\begin{pmatrix}
z^2\ovl{m}^{(n)}_{1}(z^2) & z^2\ovl{m}^{(n)}_{2}(z^2)\\
-m^{(n)}_{2}(z^2) & m^{(n)}_{1}(z^2) 
\end{pmatrix},
\end{equation*}
so that $\ovl{\wtilde{M}}_n=z^{2}\wtilde{M}_n$. Further, it is straightforward to 
verify that, for $z\in\field{T}$, the matrices $z^{-1}M_n/\sqrt{C_n}$ and 
$z^{-1}\wtilde{M}_n/\sqrt{C_n}$ are elements of 
$\text{SU}(2)$. The discrete scattering problem in its unitary form reads as 
\begin{equation}\label{eq:recurrence-jost-matrix-unitary}
\frac{w_{n+1}}{\sqrt{-W_{n+1}}} = \frac{z^{-1}}{\sqrt{C_{n+1}}}M_{n+1}(z^2)
\frac{{w}_n}{\sqrt{-W_n}}.
\end{equation}
Introducing $\mu_n, A_n,
B_n\in\field{C}$, the independent elements of transfer matrix can be put into
the form
\begin{equation}\label{eq:define-AB}
\begin{split}
    &m^{(n)}_1(z^2) = \mu_n(1 - z^2A^*_nB_n),\\
    &m^{(n)}_2(z^2) = \mu_n(A_n+B_nz^2),
\end{split}
\end{equation}
so that 
\begin{equation}
\label{eq:Cn-coeffs}
C_n = |\mu_n|^2(1+|A_n|^2)(1+|B_n|^2). 
\end{equation}
Setting $\mu_n=|\mu_n|e^{i\theta_n}$, the transfer matrix admits of the following
factorization
\begin{equation}\label{eq:transfer-matrix-general}
M_n=
|\mu_n|\begin{pmatrix}
    1& -A^*_n\\
    A_n&1\\
\end{pmatrix}
\begin{pmatrix}
    1& 0\\
    0&z^2\\
\end{pmatrix}
\begin{pmatrix}
    1& -B^*_n\\
    B_n&1\\
\end{pmatrix}e^{i\sigma_3\theta_n}.
\end{equation}
For the cases considered in this article, $\theta_n=0$, therefore, we assume
$\mu_n\in\field{R}_+$ so that it does not play a role in the unitary form of the
transfer matrix for $z\in\field{T}$. For a given initial condition and fixed
sequence of transfer matrices, the recurrence 
relation~\eqref{eq:recurrence-jost-matrix} leads to a unique
polynomial associated with the Jost solution $w_n$. In particular, the following result is
straightforward:
\begin{lemma}
\label{lemma:jost0}
Let $N$ be a finite positive integer. Let the vectors
$\vv{A}=(A_1,A_2,\ldots,A_N),\,\vv{B}=(B_1,B_2,\ldots,B_N)\in\field{C}^N$ 
and $\vs{\mu}=(\mu_1,\mu_2,\ldots,\mu_N)\in\field{R}^N_+$ define $\{M_n(z^2)\}_{n=1}^{N}$ 
through~\eqref{eq:transfer-matrix-general}. Let $w_0=\sigma_1$, 
then the recurrence relation~\eqref{eq:recurrence-jost-matrix} 
determines a sequence of Jost solutions $\{w_n\}_{n=1}^N$ such that for every 
$n$ ($1\leq n\leq N$) there exists a unique polynomial $\vv{P}_n(z^2)$ associated
with $w_n$.
\end{lemma}

Now let us consider an arbitrary polynomial $\vv{P}_n(z^2)$
satisfying~\eqref{eq:jost-cond1} for $n\geq0$. Assume
$\vv{P}_n(z^2)/\sqrt{-W_n}\in\fs{P}(n;\field{C}^2)$ and let $\vv{P}_{n+1}(z^2)$ be 
associated with $w_{n+1}$. To understand the properties of the polynomial 
$\vv{P}_{n+1}(z^2)$, we consider the recurrence
relation~\eqref{eq:recurrence-jost-matrix}. Equating the
coefficients of the zeroth degree term on the RHS and the LHS 
of~\eqref{eq:recurrence-jost-matrix}, we have
\begin{equation}\label{eq:general-recurr-LP1}
\begin{split}
&P^{(n+1)}_{1,0}=\mu_{n+1}\left(P^{(n)}_{1,0}-B^*_{n+1}P^{(n)}_{2,0} \right),\\
&P^{(n+1)}_{2,0}=\mu_{n+1}A_{n+1}\left(P^{(n)}_{1,0}-B^*_{n+1}P^{(n)}_{2,0}\right).
\end{split}
\end{equation}
It is straightforward to see that 
\begin{equation}\label{eq:general-LP1}
A_{n+1}=\frac{P^{(n+1)}_{2,0}}{P^{(n+1)}_{1,0}},
\end{equation}
and 
\begin{equation*}
P^{(n+1)}_{1,0} 
=\mu_{n+1}\left(1-B_{n+1}^*\frac{P^{(n)}_{2,0}}{P^{(n)}_{1,0}}\right)P^{(n)}_{1,0}.
\end{equation*}
Therefore, in order that
$\vv{P}_{n+1}(z^2)/\sqrt{-W_{n+1}}\in\fs{P}(n+1;\field{C}^2)$ where
$W_{n+1}=\det[w_{n+1}]=C_{n+1}W_n$, we must have
\[
1-B_{n+1}^*\frac{P^{(n)}_{2,0}}{P^{(n)}_{1,0}}\in\field{R}_+.
\]
\begin{lemma}
\label{lemma:jost}
Under the assumption of the previous lemma, setting $W_n=\det[w_n]$, the
polynomial $\vv{P}_n(z^2)/\sqrt{-W_n}\in\fs{P}(n;\field{C}^2)$ if and only if the 
sequence $\{(A_n,B_n)\}_{n=1}^N$ satisfies the constraint $(1-A_{n}B_{n+1}^*)\in\field{R}_+$ for
$1\leq n<N$. If $B_1=0$, then $\vv{P}_n(z^2)/\sqrt{-W_n}\in\fs{P}(n-1;\field{C}^2)$.
\end{lemma}
\begin{proof}
Using the recurrence relation~\eqref{eq:general-recurr-LP1} and 
the property~\eqref{eq:general-LP1} for all $n\geq0$, it is straightforward to see that 
\begin{equation}
\begin{split}
P^{(n+1)}_{1,0} 
&=\mu_{n+1}(1-A_{n}B_{n+1}^*)P^{(n)}_{1,0}\\
&=\mu_{1}\prod^{n}_{k=1}\mu_{k+1}(1-A_{k}B_{k+1}^*),
\end{split}
\end{equation}
for $n>0$ while $P^{(1)}_{1,0}=\mu_{1}$.
The proof of the first part of the lemma follows from this relation.

For the second part, equating the coefficients of $(z^2)^{n+1}$ on the RHS and the LHS 
of~\eqref{eq:recurrence-jost-matrix}, we have
\begin{align*}
&P^{(n+1)}_{1,n+1}=-\mu_{n+1}A^*_{n+1}\left(B_{n+1}P^{(n)}_{1,n}+P^{(n)}_{2,n}\right),\\
&P^{(n+1)}_{2,n+1}=\mu_{n+1}\left(B_{n+1}P^{(n)}_{1,n}+P^{(n)}_{2,n}\right),
\end{align*}
for $n\geq0$. These relations yield
\begin{equation}\label{eq:highest-coeff}
\begin{split}
&P^{(n+1)}_{1,n+1}=-\mu_{n+1}A^*_{n+1}P^{(n+1)}_{2,n+1},\\
&P^{(n+1)}_{2,n+1}=\mu_1B_1\prod^n_{k=1}\mu_{k+1}\left(1-B_{k+1}A^*_{k}\right).
\end{split}
\end{equation}
for $n>0$ while $P^{(1)}_{1,1}=-\mu_1A_1^*B_1$ and $P^{(1)}_{2,1}=\mu_1B_1$.
Therefore, if $B_1=0$ then $\vv{P}^{(n)}_{n}=0$ for $1\leq n\leq N$.
\end{proof}
Next we would like to analyze the inverse problem described as follows: Given an arbitrary 
polynomial $\vv{P}_{n+1}(z^2)$ associated with $w_{n+1}$ satisfying
\begin{equation}
\vv{P}_{n+1}(z^2)\cdot\vv{P}^*_{n+1}(z^2) = -W_{n+1}>0,
\end{equation}
for $n\geq0$ such that 
$\vv{P}_{n+1}(z^2)/\sqrt{-W_{n+1}}\in\fs{P}(n+1;\field{C}^2)$. Find a polynomial 
$\vv{P}_n(z^2)$ associated with $w_n$ and a transfer matrix $M_{n+1}(z^2)$ of
the form~\eqref{eq:transfer-matrix-general} such
that $w_n$, defined by
\begin{equation}\label{eq:recurrence-inv-jost-matrix}
w_{n} = z^{-1}\wtilde{M}_{n+1}(z^2){w}_{n+1},
\end{equation}
is a Jost solution. If such a polynomial $\vv{P}_n(z^2)$ exists then it must be
consistent with the recurrence relation
\begin{equation}\label{eq:recurrence-inv-jost-vec}
\vv{P}_n(z^2)=z^{-2}\wtilde{M}_{n+1}(z^2)\vv{P}_{n+1}(z^2),
\end{equation}
or, equivalently, 
\begin{equation}\label{eq:recurrence-jost-vec}
{M}_{n+1}(z^2)\vv{P}_n(z^2)=\vv{P}_{n+1}(z^2).
\end{equation}
Equating the coefficient of $z^{-2}$ on the RHS of~\eqref{eq:recurrence-inv-jost-vec}, we have
\begin{equation}
\begin{split}
&\frac{\mu_{n+1}}{C_{n+1}}\left(-A_{n+1}P_{1,0}^{(n+1)}+P_{2,0}^{(n+1)}\right)B^*_{n+1}=0,\\
&\frac{\mu_{n+1}}{C_{n+1}}\left(-A_{n+1}P_{1,0}^{(n+1)}+P_{2,0}^{(n+1)}\right)=0,
\end{split}
\end{equation}
which yields the recurrence relation~\eqref{eq:general-LP1}. Equating the
coefficients of $z^0$ on the RHS and the LHS of~\eqref{eq:recurrence-inv-jost-vec}, we obtain
\begin{align*}
C_{n+1}P^{(n)}_{1,0}&=\mu_{n+1}\left(P^{(n+1)}_{1,0}-A_{n+1}B^*_{n+1}P^{(n+1)}_{1,1} \right)\\
&\quad+\mu_{n+1}\left(A^*_{n+1}P^{(n+1)}_{2,0}+B^*_{n+1}P^{(n+1)}_{2,1} \right),\\
C_{n+1}P^{(n)}_{2,0}&=-\mu_{n+1}\left(B_{n+1}P^{(n+1)}_{1,0}+A_{n+1}P^{(n+1)}_{1,1}\right)\\
&\quad+\mu_{n+1}\left(-A^*_{n+1}B_{n+1}P^{(n+1)}_{2,0}+P^{(n+1)}_{2,1}\right).
\end{align*}
This yields
\begin{multline}\label{eq:Pn0-1}
P^{(n)}_{2,0}+B_{n+1}P^{(n)}_{1,0}\\=
\frac{\mu_{n+1}}{C_{n+1}}(1+|B_{n+1}|^2)\left(P^{(n+1)}_{2,1}-A_{n+1}P^{(n+1)}_{1,1}\right),
\end{multline}
and
\begin{multline}\label{eq:Pn0-2}
P^{(n)}_{1,0}-B^*_{n+1}P^{(n)}_{2,0}\\=
\frac{\mu_{n+1}}{C_{n+1}}(1+|A_{n+1}|^2)(1+|B_{n+1}|^2)P^{(n+1)}_{1,0}.
\end{multline}
which thanks to~\eqref{eq:Cn-coeffs} ($\mu_{n+1}\in\field{R}_+$) becomes
identical to~\eqref{eq:general-recurr-LP1}. 
Note that the relationship~\eqref{eq:Pn0-1} can also be verified by equating the coefficients of 
$z^2$ on the RHS and the LHS of~\eqref{eq:recurrence-jost-vec}: 
\begin{align*}
P^{(n+1)}_{1,1}&=\mu_{n+1}\left(P^{(n)}_{1,1}-A^*_{n+1}B_{n+1}P^{(n)}_{1,0} \right)\\
&\quad-\mu_{n+1}\left(B^*_{n+1}P^{(n)}_{2,1}+A^*_{n+1}P^{(n)}_{2,0} \right),\\
P^{(n+1)}_{2,1}&=\mu_{n+1}\left(A_{n+1}P^{(n)}_{1,1}+B_{n+1}P^{(n)}_{1,0}\right)\\
&\quad+\mu_{n+1}\left(P^{(n)}_{2,0}-A_{n+1}B^*_{n+1}P^{(n)}_{2,1}\right).
\end{align*}
This yields
\begin{multline*}
P^{(n+1)}_{2,1}-A_{n+1}P^{(n+1)}_{1,1}\\=\mu_{n+1}(1+|A_{n+1}|^2)
(P^{(n)}_{2,0}+B_{n+1}P^{(n)}_{1,0}),
\end{multline*}
which is identical to~\eqref{eq:Pn0-1} thanks to~\eqref{eq:Cn-coeffs}. Now, 
using~\eqref{eq:Pn0-1} and~\eqref{eq:Pn0-2}, we have
\begin{equation}\label{eq:general-LP2-B}
\frac{P^{(n)}_{2,0}+B_{n+1}P^{(n)}_{1,0}}{P^{(n)}_{1,0}-B^*_{n+1}P^{(n)}_{2,0}}=\chi_{n+1},
\end{equation}
where
\begin{equation}\label{eq:zs-prob-ratio}
\chi_{n+1}=\frac{P^{(n+1)}_{2,1}-A_{n+1}P^{(n+1)}_{1,1}}{(1+|A_{n+1}|^2)P^{(n+1)}_{1,0}}.
\end{equation}
So far we have found that the parameter $A_{n+1}$ of $M_{n+1}(z^2)$ must 
be set according to~\eqref{eq:general-LP1} so that we may write
\begin{equation}\label{eq:zeroth-coeff}
\vv{P}^{(n)}_{0}=\frac{(1+|A_{n+1}|^2)}{C_{n+1}/\mu_{n+1}}
\begin{pmatrix}
1+B^*_{n+1}\chi_{n+1}\\
\chi_{n+1}-B_{n+1}\\
\end{pmatrix}P^{(n+1)}_{1,0},
\end{equation}
where $\chi_{n+1}$ is known but $B_{n+1}$ is still an unknown. In order to compute $B_{n+1}$, we
introduce a free parameter, $\lambda_n=P^{(n)}_{2,0}/P^{(n)}_{1,0}$, so that
from~\eqref{eq:general-LP2-B}, we have
\begin{equation}\label{eq:param-B}
B_{n+1}=\frac{(1+|\lambda_n|^2)\chi_{n+1}}{1-|\chi_{n+1}|^2|\lambda_n|^2}
-\frac{(1+|\chi_{n+1}|^2)\lambda_{n}}{1-|\chi_{n+1}|^2|\lambda_n|^2}.
\end{equation}
Now, let us observe that 
\begin{align*}
& 1+B^*_{n+1}\chi_{n+1}
=\frac{1+|\chi_{n+1}|^2}{1-|\chi_{n+1}|^2|\lambda_n|^2}(1-\chi_{n+1}\lambda^*_n),\\
&\chi_{n+1}-B_{n+1}
=\frac{1+|\chi_{n+1}|^2}{1-|\chi_{n+1}|^2|\lambda_n|^2}(1-\chi_{n+1}\lambda^*_n)\lambda_n,
\end{align*}
and
\begin{equation*}
1-B^*_{n+1}\lambda_n=\frac{1+|\lambda_n|^2}{1-|\chi_{n+1}|^2|\lambda_n|^2}(1-\chi^*_{n+1}\lambda_n),
\end{equation*}
so that
\[
1+|B_{n+1}|^2=\frac{(1+|\lambda_n|^2)(1+|\chi_{n+1}|^2)}{(1-|\chi_{n+1}|^2|\lambda_n|^2)^2}
|1-\chi^*_{n+1}\lambda_n|^2.
\]
Now, the zeroth degree coefficient given by~\eqref{eq:zeroth-coeff} simplifies to
\begin{equation}
\vv{P}^{(n)}_{0}=\frac{(1-|\chi_{n+1}|^2|\lambda_n|^2)}{\mu_{n+1}(1+|\lambda_{n}|^2)}
\frac{P^{(n+1)}_{1,0}}{(1-\chi^*_{n+1}\lambda_n)}
\begin{pmatrix}
1\\
\lambda_n
\end{pmatrix}.
\end{equation}
Therefore, in order that
$\vv{P}_{n}(z^2)/\sqrt{-W_{n}}\in\fs{P}(n;\field{C}^2)$ where
$W_{n}=W_{n+1}/C_{n+1}$, we must have 
\begin{equation}\label{eq:positivity-zero-coeff}
\frac{1-\chi^*_{n+1}\lambda_n}{1-|\chi_{n+1}||\lambda_n|}\in\field{R}_+.
\end{equation}
The above condition can be enforced by setting $\lambda_n=\chi_{n+1}\omega_n$ where
we restrict ourselves to the case $\omega_n\in\field{R},\,\omega_n\geq0$. Under this
condition, the expressions for
$B_{n+1}$ and $\vv{P}_0^{(n)}$ simplifies to 
\begin{equation}\label{eq:param-B-simple}
B_{n+1}=\frac{(1-\omega_n)\chi_{n+1}}{1+|\chi_{n+1}|^2\omega_n},
\end{equation}
and 
\begin{equation}\label{eq:zero-coeff-simple}
\vv{P}^{(n)}_{0}=\frac{1+|\chi_{n+1}|^2\omega_n}{1+|\chi_{n+1}|^2\omega_n^2}
\begin{pmatrix}
1\\
\omega_n\chi_{n+1}
\end{pmatrix}\mu^{-1}_{n+1}P^{(n+1)}_{1,0},
\end{equation}
respectively. Clearly, the transfer matrix $M_{n+1}(z^2)$ as well as the
polynomial $\vv{P}_n(z^2)$ is not unique as it depends on a free parameter
$\omega_n\geq0$. Note that the parameter $\mu_{n+1}$ turns out to be merely a 
scale factor which does not play a role in the unitary form of the discrete
scattering problem.
Finally, let us observe that in order to predict the highest degree term that 
is non-zero in $\vv{P}_n(z^2)$, the
recurrence relation for $(z^2)^n\ovl{\vv{P}}(z^2)$ can be considered where the
zeroth degree term is $i\sigma_2\vv{P}^{(n)*}_n$ so that
\begin{equation}\label{eq:param-B-simple-highest-degree}
\vv{P}^{(n)}_{n}=\frac{1+|\chi_{n+1}|^2\omega_n}{1+|\chi_{n+1}|^2\omega_n^2}
\begin{pmatrix}
-\omega_n\chi^*_{n+1}\\
1
\end{pmatrix}\mu^{-1}_{n+1}P^{(n+1)}_{2, n+1}.
\end{equation}
\begin{rem}
In the discrete inverse scattering case, the two formulas~\eqref{eq:general-LP1}
and~\eqref{eq:zs-prob-ratio} remain invariant under any scaling of the polynomial
$\vv{P}_{n+1}(z^2)$. Therefore, knowledge of either $\vv{P}_{n+1}(z^2)$ or
$\vv{P}_{n+1}(z^2)/\sqrt{-W_{n+1}}$ is sufficient to
determine the transfer matrix $M_{n+1}(z^2)$.
\end{rem}
The discussion above regarding the discrete inverse scattering step can be
summarized in the following lemma:
\begin{lemma}
\label{lemma:discrete-inv-scatter-cond}
Given $\vv{P}_{n+1}(z^2)/\sqrt{-W_{n+1}}\in\fs{P}(d;\field{C}^2)$ where
$d\in\{n+1, n\}$ and
$\omega_n\in\field{R}_+$, there exists a unique unitary matrix 
$M_{n+1}(z^2)/\sqrt{C_{n+1}}$ for $z\in\field{T}$ and a polynomial
$\vv{P}_{n}(z^2)/\sqrt{-W_{n}}\in\fs{P}(d-1;\field{C}^2)$ such that
\[
\frac{\vv{P}_{n+1}(z^2)}{\sqrt{-W_{n+1}}}=
\frac{M_{n+1}(z^2)}{\sqrt{C_{n+1}}}\frac{\vv{P}_{n}(z^2)}{\sqrt{-W_{n}}}.
\]
Further, if $\omega_n\leq1$, then
\begin{multline}\label{eq:bound-AB-inv-scatter}
\left[(1+|A_{n+1}|^2)(1+|B_{n+1}|^2)\right]^{1/2}\\\leq
\frac{P^{(n)}_{1,0}/\sqrt{-W_n}}{P^{(n+1)}_{1,0}/\sqrt{-W_{n+1}}}.
\end{multline}
\end{lemma}
\begin{proof}
The first part of the lemma is evident from the discussion above. The second part 
follows from the inequality
\begin{equation*}
\frac{1+|\chi_{n+1}|^2\omega_n}{1+|\chi_{n+1}|^2\omega_n^2}\geq1,
\end{equation*}
for $\omega_n\leq1$.
\end{proof}
Next, we consider some of special cases where it is possible to obtain unique
solution of the discrete inverse scattering problem. It is worth noting
that these special cases belong to a certain choice of the values
$\{\omega_n\}_{n\in\field{Z}}$.

\subsection{Case I: $A_n=B_{n+1}$} Let $A_n=B_{n+1}$ and assume $A_n\in\field{D}$.
Then the forward scattering problem described in Lemma~\ref{lemma:jost0}
always yields a polynomial $\vv{P}_n(z^2)/\sqrt{-W_n}\in\fs{P}(n;\field{C}^2)$
on account of Lemma~\ref{lemma:jost}.

For discrete inverse scattering, the condition $A_n=B_{n+1}$ amounts to
$B_{n+1}=\chi_{n+1}\omega_n$. From~\eqref{eq:param-B-simple}, we have
\begin{equation*}
|\chi_{n+1}|^2\omega^2_n+2\omega_n-1=0,
\end{equation*}
which yields
\begin{equation}
\omega_n=\frac{1}{1+\sqrt{1+|\chi_{n+1}|^2}}
\end{equation}
as the admissible solution (the other root violates the positivity
constraint in~\eqref{eq:positivity-zero-coeff}). For this case, the
expression~\eqref{eq:zero-coeff-simple} for the zeroth degree coefficient simplifies to
\begin{equation}
\vv{P}^{(n)}_{0}=\frac{1}{2\mu_{n+1}\omega_n}
\begin{pmatrix}
1\\
\omega_n\chi_{n+1}
\end{pmatrix}P^{(n+1)}_{1,0},
\end{equation}
In the Lemma~\ref{lemma:discrete-inv-scatter-cond}, we favor the case of $d=n$
so that number of (vector) coefficients associated with $\vv{P}_n(z^2)$ be $n$.
If the steps described in the aforementioned lemma are carried out recursively to
the point $n=0$, we obtain
\begin{equation*}
\vv{P}^{(0)}_{0}=\frac{1}{2\mu_{1}\omega_0}
\begin{pmatrix}
1\\
\omega_0\chi_{1}
\end{pmatrix}P^{(1)}_{1,0}.
\end{equation*}
Note that $\chi_1=0$ on account of $\vv{P}^{(1)}_{1}=0$; therefore,
\begin{equation}
\frac{\vv{P}_{0}(z^2)}{\sqrt{-W_{0}}}=
\begin{pmatrix}
1\\
0
\end{pmatrix}\frac{P^{(1)}_{1,0}}{\sqrt{-W_1}}\sqrt{1+|A_1|^2}
=\begin{pmatrix}
1\\
0
\end{pmatrix}.
\end{equation}
Finally, we state the main result of this section which is a now merely a consequence
of the preceding lemmas applied to the case at hand:
\begin{prop}
\label{prop:caseI}
Let $\vv{A}=(A_1,A_2,\ldots,A_{N})\in\field{D}^N$ be an arbitrary vector. Let 
the transfer matrices $\{M_n(z^2)\}_{n=1}^N$ be determined 
by~\eqref{eq:transfer-matrix-general} using $\vv{A}$ together 
with $\vv{B}\in\field{D}^N$ given by $B_1=0$ and $B_{n}=A_{n-1}$ for $1<n\leq N$.
Then, corresponding to the initial condition $\vv{P}_0(z^2)=(1,0)^{\tp}$, the
recurrence relation
\[
\vv{P}_{n}(z^2)=M_{n}(z^2)\vv{P}_{n-1}(z^2),\quad 1\leq n\leq N,
\]
yields a unique polynomial
$\vv{P}_{N}(z^2)/\sqrt{-W_N}\in\fs{P}(N-1;\field{C}^{2})$ with 
$(-W_N) =\prod_{n=1}^{N}C_n>0$ such that 
\[
\left|{P^{(N)}_{2,0}}/{P^{(N)}_{1,0}}\right|<1.
\] 

Conversely, for any given polynomial $\breve{\vv{P}}_{N}(z^2)\in\fs{P}(N-1;\field{C}^{2})$
such that 
\[
\left|\breve{P}^{(N)}_{2,0}/\breve{P}^{(N)}_{1,0}\right|<1,
\] 
there exists a unique vector $\vv{A}=(A_1,A_2,\ldots,A_{N})\in\field{D}^N$ which 
determines the the transfer matrices $\{\wtilde{M}_n(z^2)/\sqrt{C_{n}}\}_{n=1}^N$ 
as stated above such that the recurrence relation
\[
\breve{\vv{P}}_{n-1}(z^2)=\frac{z^{-2}}{\sqrt{C_{n}}}\wtilde{M}_{n}(z^2)
\breve{\vv{P}}_{n}(z^2),
\]
starting from $n=N$ yields $\breve{\vv{P}}_0(z^2)=(1,0)^{\tp}$.
\end{prop}
Putting $\breve{\vv{P}}_{N}(z^2)=\vv{P}_{N}(z^2)/\sqrt{-W_N}$, note that the condition 
\[
\left|{P^{(N)}_{2,0}}/{P^{(N)}_{1,0}}\right|=\left|\breve{P}^{(N)}_{2,0}/\breve{P}^{(N)}_{1,0}\right|<1,
\] 
corresponds to the fact that $A_N\in\field{D}$ in the direct part of the
last proposition. The condition above is imposed explicitly in the converse 
part in order to ensure $A_N\in\field{D}$.
\begin{corr}
\label{corr:l2-norm-A}
Let $\vv{A}=(A_1,A_2,\ldots,A_{N})\in\field{D}^N$ correspond to
$\breve{\vv{P}}_N(z^2)\in\fs{P}(N-1;\field{C}^{2})$ as in the converse part of
the last proposition. Then the following estimate holds:
\[
\|\vv{A}\|_2=\left(\sum_{n=1}^N|A_n|^2\right)^{1/2}\leq
\left(\frac{1}{\breve{P}^{(N)}_{1,0}}-1\right)^{1/2}.
\]
\end{corr}
\begin{proof} 
The proof follows from the relation~\eqref{eq:bound-AB-inv-scatter} of
Lemma~\ref{lemma:discrete-inv-scatter-cond}.
\end{proof}
We conclude this section with the discussion of the
trapezoidal rule which corresponds to the case at hand. 
Let $\vv{Q}=(Q_1,Q_2,\ldots,Q_{N})\in\field{D}^N$. In the case 
of trapezoidal rule, it follows from the description in Sec.~\ref{sec:discrete-TR} 
that the coefficients $A_n$ and $B_n$ introduced in~\eqref{eq:define-AB} 
satisfy
\begin{equation*}
A_{n}=B_{n+1}=R_{n}=-Q_n^*,\quad 0<n<N,
\end{equation*}
with $A_N=Q_N$ and we choose $Q_0=B_1=0$. It also follows that the quantities 
$\mu_n\in\field{R}_+$ introduced in~\eqref{eq:define-AB} are given by
\begin{equation*}
\mu_n=\Theta^{-1}_{n}=(1+|Q_{n}|^2)^{-1},\quad 0<n\leq N.
\end{equation*}
Further, we have
\begin{equation*}
C_{n} = \frac{1+|Q_{n-1}|^2}{1+|Q_{n}|^2}
= \frac{\Theta_{n-1}}{\Theta_{n}},\quad 1<n\leq N,
\end{equation*}
while $C_1 = \Theta^{-1}_{1}$.
\subsection{Case II: $A_n\neq B_{n+1}$}
First, let us assume that $B_{n}=0$. The discussion of the forward scattering
problem is identical to that of the previous case. For discrete inverse
scattering, this case corresponds to $\omega_n=1$. The expression for the 
zeroth degree coefficient~\eqref{eq:zero-coeff-simple} simplifies to
\begin{equation}
\vv{P}^{(n)}_{0}=\frac{1}{\mu_{n+1}}
\begin{pmatrix}
1\\
\chi_{n+1}
\end{pmatrix}P^{(n+1)}_{1,0}. 
\end{equation}
As in the last section, we favor the case of $d=n$ in 
the Lemma~\ref{lemma:discrete-inv-scatter-cond}. Again, if the steps described in 
the aforementioned lemma are carried out recursively to
the point $n=0$, it is easy to conclude that
\begin{equation}
\frac{\vv{P}_{0}(z^2)}{\sqrt{-W_{0}}}
=\begin{pmatrix}
1\\
0
\end{pmatrix}.
\end{equation}
The necessary and sufficient condition for discrete inverse scattering in this
case can be stated as:
\begin{prop}
\label{prop:caseIIa}
Let $\vv{A}=(A_1,A_2,\ldots,A_{N})\in\field{C}^N$ be an arbitrary vector. Let 
the transfer matrices $\{M_n(z^2)\}_{n=1}^N$ be determined 
by~\eqref{eq:transfer-matrix-general} using $\vv{A}$ together 
with $B_{n}=0$ for $1\leq n\leq N$. Then, corresponding to the initial condition 
$\vv{P}_0(z^2)=(1,0)^{\tp}$, the
recurrence relation
\[
\vv{P}_{n}(z^2)=M_{n}(z^2)\vv{P}_{n-1}(z^2),\quad 1\leq n\leq N,
\]
yields a unique polynomial
$\vv{P}_{N}(z^2)/\sqrt{-W_N}\in\fs{P}(N-1;\field{C}^{2})$ with 
$(-W_N) =\prod_{n=1}^{N}C_n>0$. 

Conversely, for any given polynomial $\breve{\vv{P}}_{N}(z^2)\in\fs{P}(N-1;\field{C}^{2})$
there exists a unique vector $\vv{A}=(A_1,A_2,\ldots,A_{N})\in\field{C}^N$ which 
determines the the transfer matrices 
$\{\wtilde{M}_n(z^2)/\sqrt{C_{n}}\}_{n=1}^N$ as stated above 
such that the recurrence relation
\[
\breve{\vv{P}}_{n-1}(z^2)=\frac{z^{-2}}{\sqrt{C_{n}}}\wtilde{M}_{n}(z^2)
\breve{\vv{P}}_{n}(z^2),
\]
starting from $n=N$ yields $\breve{\vv{P}}_0(z^2)=(1,0)^{\tp}$.
\end{prop}

\begin{corr}
\label{eq:l2-norm-A-a}
Let $\vv{A}=(A_1,A_2,\ldots,A_{N})\in\field{C}^N$ correspond to
$\breve{\vv{P}}_N(z^2)\in\fs{P}(N-1;\field{C}^{2})$ as in the converse part of
the last proposition. Then the following estimate holds:
\[
\|\vv{A}\|_2=\left(\sum_{n=1}^N|A_n|^2\right)^{1/2}\leq
\left(\frac{1}{\left[\breve{P}^{(N)}_{1,0}\right]^2}-1\right)^{1/2}.
\]
\end{corr}

Secondly, let us assume that $A_{n}=0$. The discussion of the forward scattering
problem is identical to that of the previous case. For discrete inverse
scattering, this case corresponds to $\omega_n=0$. The expression for the 
zeroth degree coefficient~\eqref{eq:zero-coeff-simple} simplifies to
\begin{equation}
\vv{P}^{(n)}_{0}=\frac{1}{\mu_{n+1}}
\begin{pmatrix}
1\\
0
\end{pmatrix}P^{(n+1)}_{1,0}. 
\end{equation}
Here, we favor the case of $d=n+1$ in the Lemma~\ref{lemma:discrete-inv-scatter-cond}. Again, if 
the steps described in the aforementioned lemma are carried out recursively to
the point $n=0$, it is easy to conclude that
\begin{equation}
\frac{\vv{P}_{0}(z^2)}{\sqrt{-W_{0}}}
=\begin{pmatrix}
1\\
0
\end{pmatrix}.
\end{equation}
The expression for the highest degree
coefficient~\eqref{eq:param-B-simple-highest-degree}, simplifies to
\begin{equation}
\vv{P}^{(n)}_{n}=\frac{1}{\mu_{n+1}}
\begin{pmatrix}
0\\
1
\end{pmatrix}P^{(n+1)}_{2,n+1}. 
\end{equation}
The necessary and sufficient condition for discrete inverse scattering in this
case can be stated as:
\begin{prop}
\label{prop:caseIIb}
Let $\vv{B}=(B_1,B_2,\ldots,B_{N})\in\field{C}^N$ be an arbitrary vector. Let 
the transfer matrices $\{M_n(z^2)\}_{n=1}^N$
be determined by~\eqref{eq:transfer-matrix-general} using $\vv{B}$ together 
with $A_{n}=0$ for $1\leq n\leq N$. Then, corresponding to the initial condition 
$\vv{P}_0(z^2)=(1,0)^{\tp}$, the
recurrence relation
\[
\vv{P}_{n}(z^2)=M_{n}(z^2)\vv{P}_{n-1}(z^2),\quad 1\leq n\leq N,
\]
yields a unique polynomial
$\vv{P}_{N}(z^2)/\sqrt{-W_N}\in\fs{P}(N;\field{C}^{2})$ with 
$(-W_N) =\prod_{n=1}^{N}C_n>0$. 

Conversely, for any given polynomial $\breve{\vv{P}}_{N}(z^2)\in\fs{P}(N;\field{C}^{2})$
there exists a unique vector $\vv{B}=(B_1,B_2,\ldots,B_{N})\in\field{C}^N$ which 
determines the the transfer matrices 
$\{\wtilde{M}_n(z^2)/\sqrt{C_{n}}\}_{n=1}^N$ as stated above 
such that the recurrence relation
\[
\breve{\vv{P}}_{n-1}(z^2)=\frac{z^{-2}}{\sqrt{C_{n}}}\wtilde{M}_{n}(z^2)
\breve{\vv{P}}_{n}(z^2),
\]
starting from $n=N$ yields $\breve{\vv{P}}_0(z^2)=(1,0)^{\tp}$.
\end{prop}

\begin{corr}
Let $\vv{B}=(B_1,B_2,\ldots,B_{N})\in\field{C}^N$ correspond to
$\breve{\vv{P}}_N(z^2)\in\fs{P}(N;\field{C}^{2})$ as in the converse part of
the last proposition. Then the following estimate holds:
\[
\|\vv{B}\|_2=\left(\sum_{n=1}^N|B_n|^2\right)^{1/2}\leq
\left(\frac{1}{\left[\breve{P}^{(N)}_{1,0}\right]^2}-1\right)^{1/2}.
\]
\end{corr}

\subsubsection{Implicit Euler method}
Let $\vv{Q}=(Q_1,Q_2,\ldots,Q_{N})\in\field{C}^N$. For the implicit Euler
method, it is evident from the discussion in Sec.~\ref{sec:discrete-BDF1} that
\begin{equation*}
A_n=R_{n}=-Q^*_{n},\quad B_n=0,\quad 1\leq n\leq N,
\end{equation*}
and 
\begin{equation*}
\mu_n=\Theta^{-1}_{n}=(1+|Q_{n}|^2)^{-1}.
\end{equation*}
Further, 
\begin{equation*}
C_n = (1+|Q_{n}|^2)^{-1}=\Theta_n^{-1}.
\end{equation*}
\subsubsection{Split-Magnus method} 
For the split-Magnus
method, we consider the samples on a staggered grid so that 
$\vv{Q}=(Q_{1/2},Q_{3/2},\ldots,Q_{N-1/2})\in\field{C}^N$. It is evident from 
the discussion in Sec.~\ref{sec:split-Magnus} that 
\begin{equation*}
A_n=R_{n-1/2}=-Q^*_{n-1/2},\quad B_n=0,\quad 1\leq n\leq N,
\end{equation*}
and 
\begin{equation*}
\mu_n=\Theta^{-1/2}_{n-1/2}=(1+|Q_{n-1/2}|^2)^{-1/2}.
\end{equation*}
Further, $C_n = 1$. 

\subsubsection{Forward Euler method}
Let $\vv{Q}=(Q_0,Q_1,\ldots,Q_{N-1})\in\field{C}^N$. For the forward Euler
method, it is evident from the discussion in Sec.~\ref{sec:discrete-FE} that
\begin{equation*}
A_n=0,\quad B_{n}=R_{n-1}=-Q^*_{n-1},\quad 1\leq n\leq N,
\end{equation*}
and $\mu_n=1$. Further, 
\begin{equation*}
C_n = (1+|Q_{n-1}|^2)=\Theta_{n-1}.
\end{equation*}
\section{Stability and convergence analysis}
\label{sec:stability-convg-analysis}
The main objective of this section is to carry out an error-analysis
for various steps involved in the algorithms proposed in 
Sec.~\ref{sec:differential-formulation}. We first study the analyticity
properties of the scattering coefficients in order to understand the
difficulties involved in transitioning from the continuous to the discrete regime. In 
Sec.~\ref{sec:one-step-method}, we study the stability and convergence of the
numerical scheme for forward scattering. Note that the convergence of the layer-peeling 
algorithm where the input is synthesized using Lubich's method is not discussed
in this work, instead, we study it empirically. The error propagation in
layer-peeling procedure has been addressed in the work of
Bruckstien~\et~\cite{BKK1986}; however, on account of the underlying assumption
of piece-wise constant potential, the question of convergence beyond the first order 
cannot be addressed in their work. We leave these aspects for future research.

\subsection*{Notations}
The class of $m$-times differentiable complex-valued functions is denoted 
by $\fs{C}^m$. A function of class $\fs{C}^m$ is said to belong to 
$\fs{C}_0^m(\Omega)$, if the function and 
its derivatives up to order $m$ have a compact
support in $\Omega$ and if they vanish on the boundary ($\partial\Omega$). 
Complex-valued functions of bounded variation over
over $\field{R}$ is denoted by $\fs{BV}$ and the variation of any
function $f\in\fs{BV}$ over $\Omega\in\field{R}$ is denoted by
$\OP{V}[f;\Omega]$. If $q\in\fs{BV}$, then $\partial_xq\in\fs{L}^1$ 
exists almost everywhere such that 
$\|\partial_xq\|_{\fs{L}^1}\leq\OP{V}[q;\Omega]$~\cite[Chap.~16]{Jones2001}. Let 
$q^{(1)}$ to be equivalent to $\partial_xq$ so that
$\|q^{(1)}\|_{\fs{L}^1}=\|\partial_xq\|_{\fs{L}^1}$.

Let $J=(-\infty,L]$ and $d>0$. A complex-valued function $f(x)$ is said to
belong to the class $\fs{E}_d(J)$ if $\supp f\subset J$ and there exists a constant 
$\kappa_{\infty}>0$ such that the estimate $|f(x)|{\leq}\kappa_{\infty}e^{-2d|x|}$ 
holds almost everywhere in $J$. Clearly, $\fs{E}_d(J)\subset\fs{L}^{p}(J)$ for
$1\leq p\leq\infty$. Define $\field{S}_-(\mu)=\{\zeta\in\field{C}_-|\Im{\zeta}\geq-\mu\}$.

\subsection{Compactly supported and one-sided potentials}
\label{sec:compact-one-sided}
The Jost solution for compactly supported and one-sided potential are known to 
have analytic continuation into the upper-half of the complex
plane~\cite{AKNS1974,AS1981}. We detail 
some of these analyticity and decay properties of the Jost solutions required 
for our purpose. This discussion is motivated by the fact that our fast 
Darboux transformation (FDT) algorithm discussed in Sec.~\ref{sec:FDT-general} 
proceeds by computing the Jost solutions of a truncated potential which can be 
interpreted as one-sided (if it does not have a compact support). Further, the 
analyticity properties of the Jost solutions also determine the behavior of the 
Lubich coefficients as discussed in Sec.~\ref{sec:invscattering-Lubich}.

We begin with a study of the modified Jost solutions defined by
\begin{equation}
\wtilde{\vv{P}}(x;\zeta) 
=\vs{\phi}(x;\zeta)e^{i\zeta x}-
\begin{pmatrix}
1\\
0
\end{pmatrix}.
\end{equation}
Let $\Omega=[L_1,L_2]$ in the following unless stated otherwise. The system of 
equations~\eqref{eq:exp-int} can be transformed into a set of Volterra 
integral equations of the second kind for the modified Jost solution 
$\wtilde{\vv{P}}(x;\zeta)$:
\begin{equation}\label{eq:volterra}
\wtilde{\vv{P}}(x;\zeta)=\vs{\Phi}(x;\zeta)
+\int_{\Omega}\mathcal{K}(x,y;\zeta)\wtilde{\vv{P}}(y;\zeta)dy,
\end{equation}
where $\vs{\Phi}(x;\zeta)=(\Phi_1,\Phi_2)^{\tp}\in\field{C}^2$ with 
\begin{equation}
\begin{split}
\Phi_1(x;\zeta)&=\int_{L_1}^{x}dz\int_{L_1}^{z}dy\,q(z)r(y)e^{2i\zeta(z-y)}dy,\\
\Phi_2(x;\zeta)&=\int_{L_1}^{x}r(y)e^{2i\zeta(x-y)}dy,
\end{split}
\end{equation}
and the Volterra kernel 
$\mathcal{K}(x,y;\zeta)=\diag(\mathcal{K}_1,\mathcal{K}_2)\in\field{C}^{2\times2}$
is such that
\begin{equation}\label{eq:volterra-kernel}
\begin{split}
&\mathcal{K}_1(x,y;\zeta) = r(y)\int_{y}^{x}q(z)e^{2i\zeta(z-y)}dz,\\
&\mathcal{K}_2(x,y;\zeta) = q(y)\int_{y}^{x}r(z)e^{2i\zeta(x-z)}dz,
\end{split}
\end{equation}
with $\mathcal{K}(x,y;\zeta)=0$ for $y>x$.

\begin{theorem}\label{thm:jost-estimate}
Let $q\in\fs{L}^{1}$ be supported in $\Omega=[L_1,L_2]$ with $\kappa=\|q\|_{\fs{L}^1}$.
Then the estimate
\begin{equation}\label{eq:jost-estimate-result}
\norm{\wtilde{\vv{P}}(x;\zeta)}_{\fs{L}^{\infty}(\Omega)}\leq
\begin{cases}
    C,&\zeta\in\ovl{\field{C}}_+,\\
    Ce^{-2\Im(\zeta)(L_2-L_1)},&\zeta\in{\field{C}}_-.
\end{cases}
\end{equation}
holds with $C=\|\vv{D}\|\cosh\kappa$ where $\vv{D}=(\kappa^2/2,\kappa)^{\tp}$.
\end{theorem}

\begin{proof}
The proof can be obtained using the same method as in~\cite{AKNS1974,NMPZ1984}. For fixed
$\zeta\in\ovl{\field{C}}_+$, let $\OP{K}$ denote the Volterra integral operator
in~\eqref{eq:volterra} corresponding to the kernel $\mathcal{K}(x,y;\zeta)$ such
that
\begin{multline}
\OP{K}[\wtilde{\vv{P}}](x;\zeta)=\int_{\Omega}\mathcal{K}(x,y;\zeta)\wtilde{\vv{P}}(y;\zeta)dy\\
=\int_{L_1}^xdz\int_{L_1}^zdy
\begin{pmatrix}
q(z)r(y)e^{2i\zeta(z-y)}\wtilde{{P}}_1(y;\zeta)\\
q(y)r(z)e^{2i\zeta(x-z)}\wtilde{{P}}_2(y;\zeta)
\end{pmatrix}.
\end{multline}
Consider the $\fs{L}^{\infty}(\Omega)$-norm~\cite[Chap.~9]{GLS1990} of $\OP{K}$ given by
\begin{equation}
\|\OP{K}\|_{\fs{L}^{\infty}(\Omega)}=\esssup_{x\in\Omega}\int_{\Omega}\|\mathcal{K}(x,y;\zeta)\|dy,
\end{equation}
so that $\|\OP{K}\|_{\fs{L}^{\infty}(\Omega)}\leq\kappa^2/2$. 
The resolvent $\OP{R}$ of this operator exists and is given by the Neumann series
$\OP{R}=\sum_{n=1}^{\infty}\OP{K}_n$ where
$\OP{K}_n=\OP{K}\circ\OP{K}_{n-1}$ with $\OP{K}_{1}=\OP{K}$. It can also be
shown using the methods in~\cite{AKNS1974,NMPZ1984} that
$\|\OP{K}_n\|_{\fs{L}^{\infty}(\Omega)}\leq{\kappa^{2n}}/{(2n)!}$, 
yielding the estimate 
$\|\OP{R}\|_{\fs{L}^{\infty}(\Omega)}\leq [\cosh(\kappa)-1]$. Therefore, for any
$\vs{\Phi}(x;\zeta)\in\fs{L}^{\infty}(\Omega)$, the relationship
$\wtilde{\vv{P}}(x;\zeta)=\vs{\Phi}(x;\zeta)+\OP{R}[\vs{\Phi}](x;\zeta)$
implies, for $\zeta\in\ovl{\field{C}}_+$, 
\begin{equation}\label{eq:estimate-p-uh}
\|\wtilde{\vv{P}}(x;\zeta)\|_{\fs{L}^{\infty}(\Omega)}\leq 
\cosh(\kappa)\|\vs{\Phi}(x;\zeta)\|_{\fs{L}^{\infty}(\Omega)}.
\end{equation}
The result for $\ovl{\field{C}}_+$ in \eqref{eq:jost-estimate-result} follows 
from the observation that, for $\zeta\in\ovl{\field{C}}_+$, 
$\|\vs{\Phi}(x;\zeta)\|_{\fs{L}^{\infty}(\Omega)}\leq\|\vv{D}\|$ 
where $\vv{D}=({\kappa^2}/{2},\kappa)^{\tp}$.
Therefore, $C$ can be chosen to be $\|\vv{D}\|\cosh\kappa$. For 
the case $\field{C}_-$ of \eqref{eq:jost-estimate-result}, we consider
$\wtilde{\vv{P}}_-(x;\zeta)=\wtilde{\vv{P}}(x;\zeta)e^{-2i\zeta x}$. The 
Volterra integral equations then reads as
$\wtilde{\vv{P}}_-(x;\zeta)$:
\begin{equation}\label{eq:volterra-lh}
\wtilde{\vv{P}}_-(x;\zeta)=\vs{\Phi}_-(x;\zeta)
    +\int_{\Omega}\mathcal{K}_-(x,y;\zeta)\wtilde{\vv{P}}_-(y;\zeta)dy,
\end{equation}
where $\vs{\Phi}_-(x;\zeta)=\vs{\Phi}(x;\zeta)e^{-2i\zeta x}\in\field{C}^2$ 
and the Volterra kernel 
$\mathcal{K}_-(x,y;\zeta)=\diag(\mathcal{K}^{(-)}_1,\mathcal{K}^{(-)}_2)\in\field{C}^{2\times2}$
is such that
\begin{equation}\label{eq:volterra-kernel-lh}
\begin{split}
&\mathcal{K}^{(-)}_1(x,y;\zeta) = r(y)\int_{y}^{x}q(z)e^{-2i\zeta(x-z)}dz,\\
&\mathcal{K}^{(-)}_2(x,y;\zeta) = q(y)\int_{y}^{x}r(z)e^{-2i\zeta(z-y)}dz,
\end{split}
\end{equation}
with $\mathcal{K}_-(x,y;\zeta)=0$ for $y>x$. Using the approach outlined above,
it is possible to show that, for $\zeta\in\field{C}_-$, 
$\|\wtilde{\vv{P}}_-(x;\zeta)\|_{\fs{L}^{\infty}(\Omega)}\leq 
\cosh(\kappa)\|\vs{\Phi}_-(x;\zeta)\|_{\fs{L}^{\infty}(\Omega)}$.
The result for the case $\zeta\in\field{C}_-$ in \eqref{eq:jost-estimate-result}
then follows from the observation that 
$\|\vs{\Phi}_-(x;\zeta)\|_{\fs{L}^{\infty}(\Omega)}\leq
\|\vv{D}\|e^{2\Im(\zeta)L_1}$ for
$\zeta\in\field{C}_-$.
\end{proof}

\begin{theorem}\label{thm:jost-estimate2}
Let $q\in\fs{BV}$ with support in $\Omega=[L_1,L_2]$ such that
$q(x)=0$ for $x\in\partial\Omega$. Then, there exists a constant $C>0$ independent 
of $\zeta\in\field{C}$ such that 
the estimate
\begin{equation}
\label{eq:jost-estimate2-result}
\norm{\wtilde{\vv{P}}(x;\zeta)}_{\fs{L}^{\infty}(\Omega)}\leq
\frac{C}{1+|\zeta|}\times\begin{cases}
1&\zeta\in\ovl{\field{C}}_+,\\
e^{-2\Im(\zeta)(L_2-L_1)}&\zeta\in\field{C}_-,
\end{cases}    
\end{equation}
holds.
\end{theorem}
\begin{proof}
Consider the first term on the RHS of~\eqref{eq:volterra}: 
Integrating by parts, we obtain
\begin{align*}
    \Phi_1(x;\zeta)
&=\int_{L_1}^xq(z)e^{2i\zeta z}dz\int_{L_1}^zr(y)e^{-2i\zeta y}dy\\
&=\frac{-1}{2i\zeta}\int_{L_1}^xr(z)q(z)dz\\
&\quad+\frac{1}{2i\zeta}\int_{L_1}^xq(z)e^{2i\zeta z}
dz\int_{L_1}^z[\partial_yr(y)]e^{-2i\zeta y}dy,
\end{align*}
so that
\begin{multline*}
2(1+|\zeta|)|\Phi_1|\leq\int_{L_1}^x|q(z)|^2dz\\
+\int_{L_1}^x|q(z)|e^{-2\Im(\zeta)z}dz\int_{L_1}^z
[2|r(y)|+|\partial_yr(y)|]e^{2\Im(\zeta) y}dy.
\end{multline*}
Setting $2D_1 = \|q\|^2_2+\|q\|^2_1+\|q\|_1\|q^{(1)}\|_1$, we have
\begin{equation*}
|\Phi_1(x;\zeta)|\leq\frac{D_1}{1+|\zeta|}\times
\begin{cases}
1,&\zeta\in\ovl{\field{C}}_+,\\
e^{-2\Im(\zeta)(L_2-L_1)},&\zeta\in\field{C}_-.
\end{cases}
\end{equation*}
Again, integrating by parts, we have
\begin{align*}
\Phi_2(x;\zeta)
&=\int_{L_1}^{x}r(y)e^{2i\zeta(x-y)}dy\\
&=\frac{-1}{2i\zeta}r(x)
+\frac{1}{2i\zeta}\int_{L_1}^x[\partial_yr(y)] e^{2i\zeta (x-y)}dy,
\end{align*}
so that
\begin{multline*}
2(1+|\zeta|)|\Phi_2|\\\leq |r(x)|
    +\int_{L_1}^x[2|r(y)|+|\partial_yr(y)|]e^{-2\Im(\zeta)(x-y)}dy.
\end{multline*}
Putting $2D_2 = \|q\|_{\infty}+2\|q\|_1+\|q^{(1)}\|_1$, then
\begin{equation*}
|\Phi_2(x;\zeta)|\leq
\frac{D_2}{1+|\zeta|}\times
\begin{cases}
1,&\zeta\in\ovl{\field{C}}_+,\\
e^{-2\Im(\zeta)(L_2-L_1)},&\zeta\in\field{C}_-.
\end{cases}
\end{equation*}
Now, proceeding as in the proof of Theorem~\ref{thm:jost-estimate}, we conclude
that the estimate~\eqref{eq:jost-estimate2-result} holds with $C=\|\vv{D}\|\cosh(\|q\|_1)$ where
$\vv{D}=(D_1,D_2)^{\tp}$.

\end{proof}
Finally, let us extend the preceding two results to the one-sided potentials:
\begin{theorem}\label{thm:jost-estimate-exp}
Let $q\in\fs{E}_d(J)$ for some $d>0$ with $J=(-\infty,L]$. Let $\kappa_1=\|q\|_{\fs{L}^1(J)}$
and $\kappa_{\infty}>0$ be the constant such that $|q(x)|\leq\kappa_{\infty}e^{-2d|x|}$. Then, 
for every $\mu\in(0,d)$, the estimate
\begin{equation}\label{eq:jost-estimate-exp-result}
\norm{\wtilde{\vv{P}}(x;\zeta)}_{\fs{L}^{\infty}(J)}\leq
\begin{cases}
    C_1,&\zeta\in\ovl{\field{C}}_+,\\
    \frac{C_2e^{-2\Im(\zeta)L}}{[d^2-(\Im\zeta)^2]},&\zeta\in{\field{S}_-(\mu)}.
\end{cases}
\end{equation}
holds with constants $C_1$ and $C_2$ given by 
\[ 
C_1=\|\vv{D}\|\cosh\kappa_1,\quad C_2=\|\vv{E}\|\cosh\kappa_1,
\]
where $\vv{D}=(\kappa_1^2/2,\kappa_1)^{\tp}$ and
$\vv{E}=(\kappa^2_{\infty},d\kappa_{\infty})^{\tp}$.

In addition, if $\partial_xq\in\fs{E}_d(J)$, then there exists a constant 
$C>0$ independent of $\zeta\in\field{C}$ such that the estimate
\begin{equation}\label{eq:jost-estimate2-result-exp}
\norm{\wtilde{\vv{P}}(x;\zeta)}_{\fs{L}^{\infty}(J)}\leq
\frac{C}{1+|\zeta|},\quad\zeta\in\ovl{\field{C}}_+,
\end{equation}
holds.
\end{theorem}
\subsubsection{Error due to domain-truncation}
For the purpose of numerical solution of the ZS-problem posed on a unbounded
domain, it is mandatory to choose a computational
domain that is bounded. This requires truncation of the original unbounded
domain to a bounded domain, say, $[-L_-,L_+]$ where $L_-,L_+>0$. Let us observe here 
that the estimates obtained in Theorem~\ref{thm:jost-estimate-exp} can be 
improved slightly in order to give us a better control of the domain
truncation error. Let $\OP{K}_j$ denote the Volterra integral operator corresponding to the kernel
$\mathcal{K}_j$ for $j=1,2$ defined in~\eqref{eq:volterra-kernel}. 
Set the domain to be $J=(-\infty,-L_-]$ and assume the conditions stated in the
first part of Theorem~\ref{thm:jost-estimate-exp} to be true. Then it can be shown 
that, for $\zeta\in\ovl{\field{C}}_+$, we have 
\begin{equation*}
\begin{split}
&\|\wtilde{P}_1\|_{\fs{L}^{\infty}(J)}
=\|\Phi_1+\OP{K}_1[\Phi_1]\|_{\fs{L}^{\infty}(J)}\leq [\cosh(\kappa_1)-1],\\
&\|\wtilde{P}_2\|_{\fs{L}^{\infty}(J)}
=\|\Phi_2+\OP{K}_2[\Phi_2]\|_{\fs{L}^{\infty}(J)}\leq \sinh\kappa_1.
\end{split}
\end{equation*}
Now in any numerical scheme, one would take $(1,0)^{\tp}$ as the initial value for
the Jost solution $\vv{\phi}(x;\zeta)e^{i\zeta x}$ at $x=-L_-$. This step introduces an error which
is bounded by
$\max(\|\wtilde{P}_1\|_{\fs{L}^{\infty}(J)},\,\|\wtilde{P}_2\|_{\fs{L}^{\infty}(J)})$.
Let $L_->0$ be a free parameter and assume $q\in\fs{E}_d(\field{R})$. Now, if we require 
the maximum error to be equal to $\epsilon>0$, then it suffices to have 
$\sinh[\|q\chi_{(-\infty,-L_-]}\|_{\fs{L}^1}]=\epsilon$ which works out
to be
\begin{equation}\label{eq:truncation-criteria}
\|q\chi_{(-\infty,-L_-]}\|_{\fs{L}^1}\leq\log\left[\epsilon+\sqrt{1+\epsilon^2}\right].
\end{equation}
Similar result can be obtained for truncation from the right side by using the
property in Remark~\ref{rem:reflection}.
\subsection{Discretization in the spectral domain}
\label{sec:one-step-method}
Let the grid points be as defined in Sec.~\ref{eq:discrete-spec-domain}. In this
section we discuss of the stability and convergence properties of the numerical
methods developed in Sec.~\ref{eq:discrete-spec-domain}. To this end, we 
closely follow the terminology introduced in~\cite{Gautschi2012} adapted to
the problem at hand.

The general form of a one-step method for~\eqref{eq:exp-int} can be stated as 
\begin{equation}\label{eq:one-step-form}
\tilde{\vv{v}}_{n+1}=[\sigma_0+h\Lambda(x_n;h)]\tilde{\vv{v}}_n,    
\end{equation}
where dependence on the spectral parameter, $\zeta$, is suppressed. We keep the
spectral parameter fixed in the following discussion or allow it to vary over
any compact domain of $\field{C}$. The function $\Lambda(x_n;h)$ is referred to as
the \emph{update function} of the one-step method. The truncation error of this method 
is defined as
\begin{equation}
\vv{T}(x,\vv{y};h)=\frac{1}{h}[\tilde{\vv{v}}(x+h)-\tilde{\vv{v}}(x)]
-\Lambda(x;h)\tilde{\vv{v}}(x),
\end{equation}
with $\tilde{\vv{v}}(x)=\vv{y}$. A method is called 
\emph{consistent} if $\lim_{h\rightarrow0}\vv{T}(x,\vv{y};h)=0$. The
necessary and sufficient condition for consistency in this case 
is $\Lambda(x;0)= \wtilde{U}(x)$. A method is said
to have an order $p$ if, for some vector norm $\|\cdot\|$, 
$\|\vv{T}(x,\vv{y};h)\|\leq Ch^p$ holds uniformly over 
$\Omega\times\Gamma$ where $\Gamma\subset\field{C}^2$ is a compact set and
$C$ is independent of $x$, $\vv{y}$ and $h$. Let $\vv{x}_h=(x_n)_{0\leq n\leq
N}$ represent the grid. Let us introduce a vector-valued grid-function
as $\gf{u}=\{\vv{u}_n\}_{n=0}^N$ where $\vv{u}_n\in\field{C}^2$ such that value of
$\gf{u}$ at $x_n$ is $\vv{u}_n$. The class of such solutions
is denoted by $\fs{G}(\vv{x}_h)$. Define the infinity-norm of any grid-function as 
\begin{equation}
\|\gf{u}\|_{\infty}=\max_{0\leq n\leq N}\|\vv{u}_n\|.
\end{equation}    
In order to introduce the concept of stability of the one-step method, let us
define the residue operator as 
\begin{equation}
    (\OP{R}_h\gf{u})_n=\frac{1}{h}[\vv{u}_{n+1}-\vv{u}_n]-\Lambda(x_n;h)\vv{u}_n,
\end{equation}
for any grid-function $\gf{u}\in\fs{G}(\vv{x}_h)$ and $n<N$ 
(we set $(\OP{R}_h\gf{u})_{N}=(\OP{R}_h\gf{u})_{N-1}$). A method is said to be
\emph{stable} if there exists a constant $C_0$ for $h_0>0$ such that for any two
arbitrary grid-functions, 
$\gf{u},\,\gf{w}\in\fs{G}(\vv{x}_h)$, we have 
\begin{equation}
\|\gf{u}-\gf{w}\|_{\infty}\leq
C_0(\|\vv{u}_0-\vv{w}_0\|+\|\OP{R}_h\gf{u}-\OP{R}_h\gf{w}\|_{\infty}),
\end{equation}
for all $h\leq h_0$. 
\begin{rem} 
The intuition behind this definition, as explained in~\cite{Gautschi2012} is as
follows: if $\gf{u}\in\fs{G}(\vv{x}_h)$ denotes the
grid-function obtained by the one-step method using infinite-precision arithmetic
(so that $\OP{R}_h\gf{u}=0$) and if $\gf{w}\in\fs{G}(\vv{x}_h)$ denotes the
grid-function obtained using finite-precision arithmetic (initial conditions being the same, i.e.,
$\vv{u}_0=\vv{w}_0$), then any stable method must yield 
$\|\gf{u}-\gf{w}\|_{\infty}=\bigO{\epsilon}$ where $\epsilon$ is the machine
precision in the latter case. 
\end{rem}
Let $q\in\fs{BV}(\Omega)$, then there exist a constants $C$ and $h_0>0$
independent of $x$ and $h$ such that $\|\Lambda(x;h)\|<C$ for all $h\in[0,h_0]$, 
(where $\|\cdot\|$ is the induced matrix norm). This shows that for 
any two arbitrary vectors, $\vv{u},\vv{w}\in\field{C}^2$ and $h\in[0,h_0]$, the
Lipschitz condition,
\[
\|\Lambda(x;h)\vv{u}-\Lambda(x;h)\vv{w}\|\leq \|\Lambda(x;h)\|\,\|\vv{u}-\vv{w}\|\leq C\|\vv{u}-\vv{w}\|,
\]
is satisfied. Therefore, the stability of the one-step method~\eqref{eq:one-step-form} easily follows 
from~\cite[Theorem 5.3.1]{Gautschi2012}. Further, for any grid-function
$\gf{u}\in\fs{G}(\vv{x}_h)$, 
we have
\begin{equation*}
    \|\vv{u}_{n}\|\leq(1+Ch)\,\|\vv{u}_{n-1}\|\leq(1+Ch)^{n}\|\vv{u}_0\|.
\end{equation*}
Then using the inequality $(1+Ch)^{N}<e^{ChN}$, it follows that $\|\gf{u}\|_{\infty}\leq
e^{C(L_2-L_1)}\|\vv{u}_0\|$ which also guarantees the boundedness of the
numerical solution when computed using infinite precision.

Finally, consistency and stability for any given one-step method imply global
convergence. Moreover, if
$\tilde{\gf{v}}=\{\tilde{\vv{v}}(x_n)\}_{n=0}^N$, denotes the
grid-function determined by the exact solution and $\gf{u}\in\fs{G}(\vv{x}_h)$
be any grid-function obtained using the one-step method \eqref{eq:one-step-form} 
with initial condition $\vv{u}_0=\tilde{\vv{v}}(x_0)$, then
$\|\gf{u}-\tilde{\gf{v}}\|_{\infty}=\bigO{h^p}$ where $p$ is the order of the
one-step method~\cite[Theorem 5.3.2]{Gautschi2012}.

\subsubsection{Implicit Euler method}
Continuing from Sec.~\ref{sec:discrete-BDF1}, we have
\begin{equation*}
\tilde{\vv{v}}_{n+1} =
\left(\sigma_0-h\wtilde{U}_{n+1}\right)^{-1}\tilde{\vv{v}}_{n},
\end{equation*}
which determines the update function to be
\[
\Lambda(x_n;h)=\frac{-h(\det\wtilde{U}_{n+1})\sigma_0
+\wtilde{U}_{n+1}}{[1+(\det\wtilde{U}_{n+1})h^2]}.
\]
It is easy to verify that $\Lambda(x_n;0)=\wtilde{U}_n$, therefore, the method is
consistent. Using the Taylor's theorem
\[
\vv{T}(x,\vv{y};h)
=\left(\sigma_0-h\wtilde{U}(x+h)\right)^{-1}
\left[-\frac{h}{2}\partial^2_{x}\tilde{\vv{v}}(x')\right],
\]
where $x\leq x'\leq x+h$ and $\tilde{\vv{v}}(x)=\vv{y}$. Assuming 
that $q(x)\in\fs{C}^1_0(\Omega)$, we have 
$\partial^2_x\tilde{\vv{v}}=e^{i\sigma_3\zeta x}
(\partial_xU+U^2+2i\zeta[\sigma_3,U])\vv{v}$, therefore, the order of the method is
$p=1$. If the Jost solution under consideration is ${\vv{v}}=\vs{\phi}$, then
$\|e^{i\zeta x}\vs{\phi}\|$ is bounded for $\zeta\in\ovl{\field{C}}_+$ (see
Theorem~\ref{thm:jost-estimate}), consequently, the 
truncation error coefficient to the leading order in $\zeta$ 
is $|\zeta| he^{2\Im(\zeta) x}\|[\sigma_3,U]\|$. Evidently, the method is stable 
which together with its consistency imply convergence (with order $p=1$).

\subsubsection{Trapezoidal rule}
Continuing from Sec.~\ref{sec:discrete-TR}, we have
\begin{equation}
\tilde{\vv{v}}_{n+1} 
= \left(\sigma_0-\frac{h}{2}\wtilde{U}_{n+1}\right)^{-1}
\left(\sigma_0+\frac{h}{2}\wtilde{U}_{n}\right)\tilde{\vv{v}}_{n},
\end{equation}
so that the update function is given by
\[
\Lambda(x_n;h)=\frac{h(\wtilde{U}_{n+1}\wtilde{U}_{n}-\sigma_0\det\wtilde{U}_{n+1}) 
+2(\wtilde{U}_{n}+\wtilde{U}_{n+1})}
{\left[4+(\det\wtilde{U}_{n+1})h^2\right]}.   
\]
Again, it is easy to verify that $\Lambda(x_n;0)=\wtilde{U}_n$, therefore, the
method is consistent. Using the Taylor's theorem
\[
\vv{T}(x,\vv{y};h)
=\left[-\frac{h^2}{12}\partial^3_{x}\tilde{\vv{v}}(x)\right]
+\bigO{h^3},
\]
with $\tilde{\vv{v}}(x)=\vv{y}$. Assuming that $q(x)\in\fs{C}^2_0(\Omega)$, we have 
\[
\partial^3_x\tilde{\vv{v}}
=-4\zeta^2e^{i\sigma_3\zeta x}
\left([\sigma_3,[\sigma_3,U]]+\bigO{1/\zeta}\right)\vv{v},
\]
therefore, the order of the method is
$p=2$. Again, if the Jost solution under consideration is ${\vv{v}}=\vs{\phi}$, then
$\|e^{i\zeta x}\vs{\phi}\|$ is bounded for $\zeta\in\ovl{\field{C}}_+$ 
(see Theorem~\ref{thm:jost-estimate}), consequently, the 
truncation coefficient to the leading order in $\zeta$ is 
$|\zeta|^2 h^2e^{2\Im(\zeta) x}\|[\sigma_3,[\sigma_3,U]]\|/3$. Evidently, the method 
is stable which together with its consistency imply convergence (with order $p=2$).

\subsubsection{Split-Magnus method}
For the convergence analysis of the split-Magnus method described 
in Sec.~\ref{sec:split-Magnus}, let us observe that an equivalent
form of the integrator is
\[
\sqrt{\left(\sigma_0-h\wtilde{U}_{n+1/2}\right)}\tilde{\vv{v}}_{n+1} =
\sqrt{\left(\sigma_0+h\wtilde{U}_{n+1/2}\right)}\tilde{\vv{v}}.
\]
Using Taylor's theorem for matrix functions, we have
\[
\vv{T}(x,\vv{y};h)=\frac{h^2}{24}
\left(\partial_x^3\tilde{\vv{v}}-3\wtilde{U}\partial^2_x\tilde{\vv{v}}
-3\wtilde{U}^3\tilde{\vv{v}}\right)+\bigO{h^3}.  
\]
with $\tilde{\vv{v}}(x)=\vv{y}$. Assuming $U$ to be twice differentiable on 
$[x,x+h]$, we conclude that the order of the method is $p=2$. Further, this 
one-step method is consistent and stable, therefore, also
convergent for fixed $\zeta$ (or $\zeta$ varying in a compact domain). The
truncation error coefficient to the leading order in $\zeta$ is 
$|\zeta|^2 h^2e^{2\Im(\zeta) x}\|[\sigma_3,[\sigma_3,U]]\|/6$. This value can be
seen to be twice as small as that of the trapezoidal rule. Let us note that it
does not seem straightforward to determine which of the two one-step methods has 
smaller total truncation error coefficient (for fixed $\zeta$); however, the 
trapezoidal rule appears to exhibit smaller total truncation error in the 
numerical tests.

\subsection{Computation of norming constants}
\label{sec:nconst-ill-cond}
In Sec.~\ref{sec:eig-and-nconst}, it was stated that the computation of
norming constants from the discrete $b$-coefficients $b_N(z^2)$ is
ill-conditioned. This can be attributed to the nature of the truncation error
coefficients in the underlying one-step method for complex values of $\zeta$.
It is evidenced by the presence of a factor of the form $\exp[2\Im(\zeta)x]$
in the truncation error coefficient which tends to grow for $x>0$ 
(see Sec.~\ref{sec:one-step-method}). Therefore, 
it is better to ``truncate'' the
scattering potential at the origin\footnote{If the growth behavior of the
potential is known before-hand, then it possible to choose an optimal point of
truncation.} and solve the corresponding one-sided ZS-problems
as discussed in Sec.~\ref{sec:eig-and-nconst}. Finally, let us note that there
are other discretization schemes such as the \emph{exponential time-differencing}
(ETD) scheme~\cite{CM2002} which may alleviate these problems; however, it may
come at a cost of increased operational complexity. These ideas will
be explored in a future publication.

\subsection{Lubich's method}
\label{sec:lubich-convergence}
Starting from the functions $a(\zeta)$ and $\breve{b}(\zeta)$ analytic in the
upper-half of the complex plane, Lubich's construction as described in 
Sec.~\ref{sec:invscattering-Lubich} allows us to compute the polynomials
associated with the discrete scattering coefficients
$\vv{P}_N(z^2)=\{\vv{P}(z^2)\}_N$. Note
that, in the preceding section, we discussed the necessary and sufficient
condition for discrete inverse scattering with polynomials (which can be seen as
a finite-support sequence of coefficients). However, Lubich's
method yields infinite series that needs to be truncated. Therefore, 
the compatibility of Lubich's construction with the layer-peeling
algorithm cannot be studied within the framework of finite-support sequences.
However, it is possible to determine if $\vv{P}(z^2)$ can be associated to a
Jost solution prior to truncation of the series. If the
coefficients of the series decay sufficiently fast, the truncation introduces a
negligible error so that the layer-peeling criteria can be satisfied to a sufficient
degree of accuracy.

Let us first consider the case of compactly supported potential. Define the vector 
$\vv{P}(z^2)=(P_1,P_2)^{\tp}$ as 
\begin{equation*}
P_1(z^2) = a\left(\frac{i\delta(z^2)}{2h}\right),\quad
P_2(z^2) = \breve{b}\left(\frac{i\delta(z^2)}{2h}\right),
\end{equation*}
which can be expanded into a convergent Taylor series as 
in Sec.~\ref{sec:invscattering-Lubich} on account of the analyticity of the 
scattering coefficients over whole of the complex plane. Further note that
\begin{align*}
&P^*_1(1/z^{*2}) = a^*\left(\frac{i\delta(1/z^{*2})}{2h}\right),\\
&P^*_2(1/z^{*2}) = \breve{b}^*\left(\frac{i\delta(1/z^{*2})}{2h}\right).
\end{align*}
Therefore, for $z\in\field{T}$, we have\footnote{For sufficiently small $h$, it
can be verified that $|P_{1,0}|\neq0$. Other conditions pertaining to
the specific discretization schemes can be explicitly verified using the results
in Sec.~\ref{sec:invscattering-Lubich}.}
\begin{align*}
\vv{P}(z^2)\cdot\vv{P}^*(z^2)=(aa^*+\breve{b}\breve{b}^*)\circ
\left(\frac{i\delta(z^2)}{2h}\right)=1.
\end{align*}
Note that here 
we have used the fact that $a(\zeta)\ovl{a}(\zeta)+{b}(\zeta)\ovl{b}(\zeta)=1$ for 
all $\zeta\in\field{C}$, however, such a relationship would not hold if we relax 
the requirement of compact support of the potential.

Let $f(\zeta)$ denote either $a(\zeta)-1$ or $\breve{b}(\zeta)$. When $f(\zeta)$
is analytic in the upper-half of the complex plane, then on any compact domain
$\Gamma\subset\ovl{\field{C}}_+$ the 
functions can be regarded as Lipschitz continuous. Observing that $\delta(e^{-h})/h= 1 + \bigO{h^p}$
where $p=1$ for BDF1 and $p=2$ for TR, we have
\begin{equation}\label{eq:estimate-delta-lubich}
\biggl|\zeta - \frac{i}{2h}\delta(e^{2i\zeta h})\biggl|\leq
Ch^p,
\end{equation}
on any compact domain of $\Gamma\subset\ovl{\field{C}}_+$ and $h\in(0,\bar{h}]$
($\bar{h}>0$) where $C>0$ depends only on $\Gamma$ and $\bar{h}$. Therefore, using the 
estimate~\eqref{eq:estimate-delta-lubich}
and the Lipschitz continuity of $f(\zeta)$ one can assert that
there exists a constant $C'>0$ for a given $\Gamma$ and $h_0$ such that~\cite{Lubich1994}
\[
\biggl|f(\zeta)-f\biggl(\frac{i\delta(e^{2i\zeta h})}{2h}\biggl)\biggl|\leq
C'h^p.    
\]
Therefore, the Wronskian relationship, $|a(\xi)|^2+|b(\xi)|^2=1$ for
$\xi\in\field{R}$ can only be satisfied upto $\bigO{h^p}$ on any bounded
interval in $\field{R}$. 

Finally, as far as the truncation of the infinite series is concerned, let us
note that for the kind problems considered in this
article, Lubich's method is applied to rational functions with known poles in
$\field{C}_-$ which makes it easy to determine the decay behavior of
these coefficients using the method of partial-fractions 
(see Sec.~\ref{sec:DT-pure-soliton}).


\subsection{Darboux transformation}\label{sec:accuracy-DT} 
In this section, we study convergence behavior of the Darboux transformation
with numerically computed Jost solutions. Continuing from Sec.~\ref{sec:DT}, 
let $(\zeta_k,b_k)\in\mathfrak{S}_K$
denote the discrete eigenvalue and the corresponding norming constant. Define 
the Vandermonde matrix 
\[
F=\{F_{jk}\}_{K\times K}=\{\zeta_j^{k}|\,j=1,\ldots,K,\,k=0,\ldots,K-1\},
\] 
the diagonal matrix $\Gamma=\text{diag}(\gamma_1,\gamma_2,\ldots,\gamma_K)$ 
and the vectors
\begin{equation}
\vv{f}=\begin{pmatrix}
\zeta_1^K\\
\zeta_2^K\\
\vdots\\
\zeta_{K}^K\\
\end{pmatrix},\quad\vv{g}=\Gamma\vv{f}=\begin{pmatrix}
\zeta_1^K\gamma_1\\
\zeta_2^K\gamma_2\\
\vdots\\
\zeta_{K}^K\gamma_K\\
\end{pmatrix},
\end{equation}
where 
\begin{equation}\label{eq:gama-DT}
\gamma_k = \frac{\phi_2^{(0)}(0,0;\zeta_k) - b_{k}\psi_2^{(0)}(0,0;\zeta_k)}
{\phi_1^{(0)}(0,0;\zeta_k)-b_{k}\psi_1^{(0)}(0,0;\zeta_k)}.
\end{equation}
The unknown Darboux coefficients can be put into the vector form
\begin{equation}
\vv{D}_0=\begin{pmatrix}
d_0^{(0,K)}\\
d_0^{(1,K)}\\
\vdots\\
d_0^{(K-1,K)}
\end{pmatrix},\quad
\vv{D}_1=\begin{pmatrix}
    d_1^{(0,K)}\\
    d_1^{(1,K)}\\
\vdots\\
    d_1^{(K-1,K)}
\end{pmatrix},
\end{equation}
then the linear system of equations~\eqref{eq:linear-eq-Darboux} which determines the
coefficients of the Darboux matrix can be written as 
\begin{equation}
    \underbrace{-\begin{pmatrix}
\vv{f}\\
\vv{g}^*
\end{pmatrix}}_{\vs{w}}=
\underbrace{\begin{pmatrix}
F & \Gamma F\\
\Gamma^* F^* & -F^*
    \end{pmatrix}}_{\mathcal{W}}\underbrace{\begin{pmatrix}
\vv{D}_0\\
\vv{D}_1
    \end{pmatrix}}_{\vv{D}}.
\end{equation}
Note that the quantities $\vv{f}$ and $F$ are known exactly while $\Gamma$ (and
in turn $\vv{g}$) is determined only up to $\bigO{h^{p}}$, where $p$ is the order 
of convergence of the one-step method. Let $\|\cdot\|$ denote the Euclidean norm 
for vectors and the induced spectral norm for matrices. Define 
$\kappa(\mathcal{W})=\|\mathcal{W}^{-1}\|\cdot\|\mathcal{W}\|$ to 
be the condition number of $\mathcal{W}$; then, under the assumption 
$\|\mathcal{W}^{-1}\|\cdot\|\Delta\mathcal{W}\|<1$ (which can be satisfied 
for sufficiently small $h$), 
the standard perturbation theory~\cite[Chap.~11]{LT1985} yields the estimate
\begin{equation}
\frac{\|\Delta\vv{D}\|}{\|\vv{D}\|}
\leq\frac{\kappa(\mathcal{W})}{1-\kappa(\mathcal{W})\frac{\|\Delta\mathcal{W}\|}{\|\mathcal{W}\|}}
\left(\frac{\|\Delta\vs{w}\|}{\|\vs{w}\|}
+\frac{\|\Delta\mathcal{W}\|}{\|\mathcal{W}\|}\right).
\end{equation}
Given that the perturbations are of 
$\bigO{h^p}$, from above equation it follows
that the coefficients of the Darboux matrix can be determined up to
$\bigO{h^p}$. 

In order to determine the convergence behavior of the fast Darboux transformation (FDT)
algorithm as described in Sec.~\ref{sec:CDT-general}, we need to study the
convergence of the corresponding Lubich coefficients. To this end, let us denote by 
$\wtilde{D}_K(\zeta)$ the approximation to the Darboux matrix
${D}_K(\zeta)$ (for the sake of brevity, we suppress the dependence on $x$, $t$
and $\mathfrak{S}_K$). Now, using the partial fraction expansion as 
in~\eqref{eq:darboux-pf}, we have
\begin{multline}\label{eq:error-pf}
\mu_K(\zeta)\left[\wtilde{D}_K(\zeta)-D_{K}(\zeta)\right]\\
=\sum_{k=0}^K\frac{\Res[\mu_K,\zeta_k^*]}{(\zeta-\zeta_k^*)}
\left[\wtilde{D}_K(\zeta^*_k)-D_K(\zeta^*_k)\right].
\end{multline}
In order to establish the relationship between the error in the Darboux matrix as
stated above and the error in the coefficients of the Darboux matrix, we need
the following lemma:
\begin{lemma}
For a given discrete spectrum $\mathfrak{S}_K$ where $K$ is finite, the inequality
\[
\left\|\wtilde{D}_K(\zeta)-D_{K}(\zeta)\right\|\leq
2\left\|\Delta\vv{D}\right\|G_K(|\zeta|^2),
\]
holds for any $\zeta\in\field{C}$ where
\[
G_K(\xi) = 
\begin{cases}
\sqrt{\frac{\xi^{K}-1}{\xi-1}},&\xi\neq1,\\
\sqrt{K} & \xi=1.
\end{cases}
\]
\end{lemma}
\begin{proof}
From the the definition of the Darboux matrix, we have
\[
\left\|\wtilde{D}_K(\zeta)-D_{K}(\zeta)\right\|\leq\sum_{k=0}^{K-1}
\left\|\wtilde{D}_k^{(K)}-{D}_k^{(K)}\right\|\,|\zeta|^k,
\]
for $\zeta\in\field{C}$. Now, using the Cauchy-Schwartz inequality, we
obtain
\begin{multline*}
\sum_{k=0}^{K-1}\left\|\wtilde{D}_k^{(K)}-{D}_k^{(K)}\right\||\zeta|^k\\
\leq\sqrt{\sum_{k=0}^{K-1}\left\|\wtilde{D}_k^{(K)}-{D}_k^{(K)}\right\|^2}
\sqrt{\sum_{k=0}^{K-1}|\zeta|^{2k}}.
\end{multline*}
Note that this inequality does not change on replacing the spectral norm
($\|\cdot\|$) with the
Euclidean norm ($\|\cdot\|_{E}$) and it is easy to see
\[
\sqrt{\sum_{k=0}^{K-1}\left\|\wtilde{D}_k^{(K)}-{D}_k^{(K)}\right\|_{E}^2}=2\|\Delta\vv{D}\|,
\]
which concludes the proof.
\end{proof}
Let $\mathcal{D}_K(\tau)$ and $\wtilde{\mathcal{D}}_K(\tau)$ 
denote the inverse Fourier-Laplace transform of 
$\mu_K(\zeta)D_K(\zeta)-\sigma_0$ and
$\mu_K(\zeta)\wtilde{D}_K(\zeta)-\sigma_0$, respectively; then,
we have the following proposition for the rate of convergence:
\begin{prop}
Consider the discrete spectrum $\mathfrak{S}_K$ with finite $K$. If 
$\|\Delta\vv{D}\|=\bigO{h^p}$ where $p$ is order of the underlying one-step
method, then
\begin{equation*}
\left\|\wtilde{\mathcal{D}}_K(\tau)-\mathcal{D}_K(\tau)\right\|
=\bigO{h^p}.
\end{equation*}
\end{prop}
\begin{proof} Let the set of eigenvalues be $\mathfrak{E}_K$ corresponding to
$\mathfrak{S}_K$ and define
\[
R=2\max_{\zeta\in\mathfrak{E}_K}|\Res[\mu_K,\zeta^*]|G_K(|\zeta|^2),
\]
where $G_K$ is as defined in the forgoing lemma. From~\eqref{eq:error-pf} and the forgoing lemma, we have
\begin{align*}
&\left\|\wtilde{\mathcal{D}}_K(\tau)-\mathcal{D}_K(\tau)\right\|
\leq R\|\Delta\vv{D}\|\sum_{k=1}^K\left|\frac{1}{2\pi}\int^{\infty+ic}_{-\infty+ic}
\frac{e^{-i\zeta\tau}d\zeta}{\zeta-\zeta_k^*}\right|\\
&\leq R\|\Delta\vv{D}\|\sum_{k=1}^Ke^{-\eta_k\tau}\leq RK\|\Delta\vv{D}\|,
\end{align*}
where $\eta_k=\Im{\zeta_k}>0$. The result follows by setting $\|\Delta\vv{D}\|=\bigO{h^p}$.
\end{proof} 
Let the matrix-valued Lubich coefficients for ${D}_K(\zeta)$
and $\wtilde{D}_K(\zeta)$ be defined as
\begin{equation}
\begin{split}
&\mu_{K}\left(\frac{i\delta(z^2)}{2h}\right){D}_K\left(\frac{i\delta(z^2)}{2h}\right)
=\sum_{k=0}^{\infty}\Lambda_k(h)z^{2k},\\
&\mu_{K}\left(\frac{i\delta(z^2)}{2h}\right)\wtilde{D}_K\left(\frac{i\delta(z^2)}{2h}\right)
=\sum_{k=0}^{\infty}\tilde{\Lambda}_k(h)z^{2k},
\end{split}
\end{equation}
respectively. The zeroth Lubich coefficient is obtained by evaluating the
Darboux matrix at $\zeta=i\delta(0)/2h$. Therefore,
\begin{equation}
\left\|\Lambda_0-\tilde{\Lambda}_0\right\|\leq
2hR\|\Delta\vv{D}\|\sum_{k=1}^K\frac{1}{|i\delta(0)-2h\zeta_k^*|},
\end{equation}
leading to $\|\Lambda_0-\tilde{\Lambda}_0\|=\bigO{h^{p+1}}$.
Using the properties of the Lubich coefficients and the forgoing proposition, it follows that 
\begin{equation}
\left\|\Lambda_k-\tilde{\Lambda}_k\right\|=\bigO{h^{p+1}},\quad k\in\field{Z}_+.
\end{equation}

\section{Numerical Tests}\label{sec:numerical-tests}
In this section, we present several numerical tests to demonstrate the
performance of the numerical algorithms developed in this paper. For better 
numerical conditioning, we scale the scattering potential $q(x)$ of the
ZS-problem by a suitable scaling parameter such that $\|q\|_{\fs{L}^2}$ is unity 
or close to unity. Let us briefly review the effect of this scaling 
to~\eqref{eq:zs-prob}:  For some $\kappa>0$, let 
$V(y)={U}(x)/\kappa$, $y=\kappa x$ and $\lambda=\zeta/\kappa$ then
\begin{align*}
&\vv{v}_y(y/\kappa;\zeta) =
i(\zeta/\kappa)\sigma_3\vv{v}(y/\kappa;\zeta)+U(y/\kappa)\vv{v}(y/\kappa;\zeta),\\
&\vv{w}_y = -i\lambda\sigma_3\vv{w}+V(y)\vv{w},
\end{align*}
where $\vv{w}(y;\lambda)=\vv{v}(y/\kappa;\lambda\kappa)$.

For the sake of clarity, let us specify the acronyms used to denote the one-step
methods considered in this article for testing: \emph{implicit Euler method (BDF1), trapezoidal rule
(TR), Magnus method with one-point Gauss quadrature (MG1) and split-Magnus
method (SM)}.
The main focus of this section is to study the dependence of the 
total numerical error on the free parameters of a given algorithm together with
its total run-time. In particular, we have considered the test cases that test the performance of
the new methods introduced in this article against the so called benchmarking
methods (MG1 and SM) wherever possible. In all of the test cases described below, $N$
represents the number of samples which is taken from the set
$\mathfrak{N}=\{2^j,\,j=10,11,\ldots,20\}$. 

\subsection{Examples}
Our test cases are derived from following examples for which the exact
value of the quantities to be analyzed are known in a closed form or can be
evaluated to the machine precision by a known method.
\begin{figure}[!hb]
\centering
\includegraphics[width=0.48\textwidth]{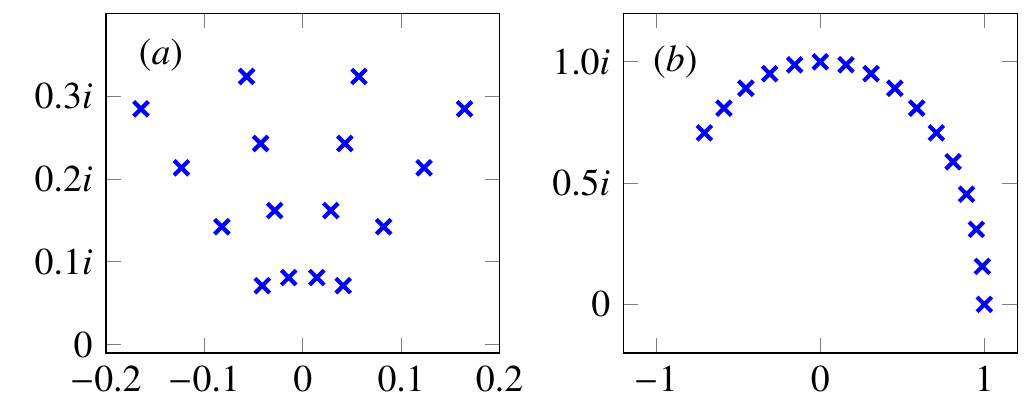}%
\caption{The figure shows the discrete spectrum, $\mathfrak{S}_{16}$ 
(defined by \eqref{eq:ds-test-ms}), where the eigenvalues and the norming constants
are shown in (a) and (b), respectively.
\label{fig:ds-test}}
\end{figure}

\begin{figure*}[!t]
\centering
\includegraphics[scale=1.0]{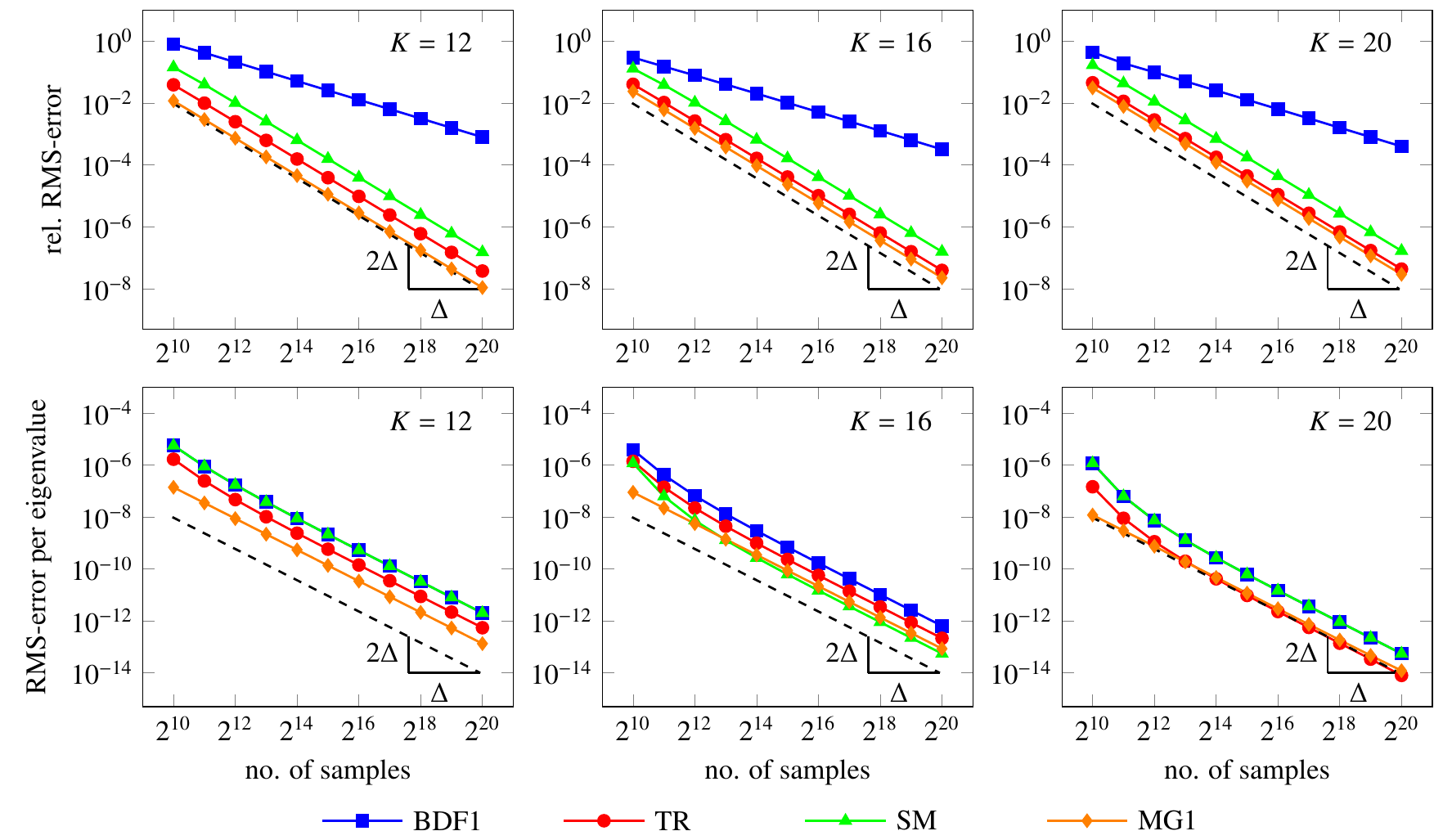}
\caption{The figure depicts the convergence analysis for the norming constants
(top row) and the discrete eigenvalues (bottom row). 
The numerical test (as described in Sec.~\ref{tc:discrete-spectrum}) is carried
out for fixed number of eigenvalues ($K\in\{12, 16, 20\}$).\label{fig:discrete-spectrum}}
\end{figure*}

\begin{figure*}[!t]
\centering
\includegraphics[scale=1]{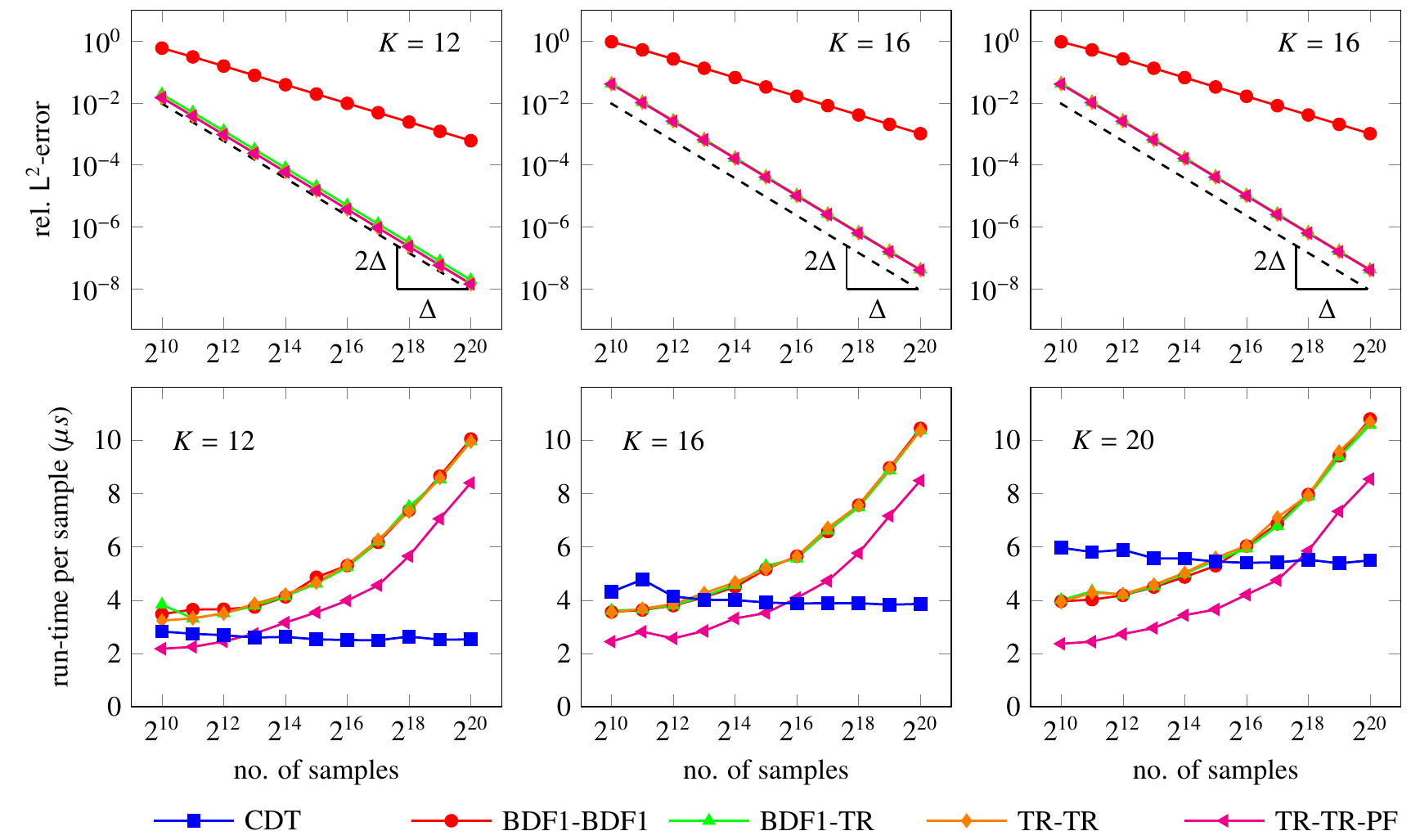}
\caption{The figure depicts the convergence analysis (top row) and the run-time
behavior (bottom row) for multi-solitons. The numerical test (as described in
Sec.~\ref{tc:ms}) is carried out with fixed number of eigenvalues 
($K\in\{12, 16, 20\}$).\label{fig:multi-soliton-err-convg}}
\end{figure*}

\begin{figure*}[!t]
\centering
\includegraphics[scale=1]{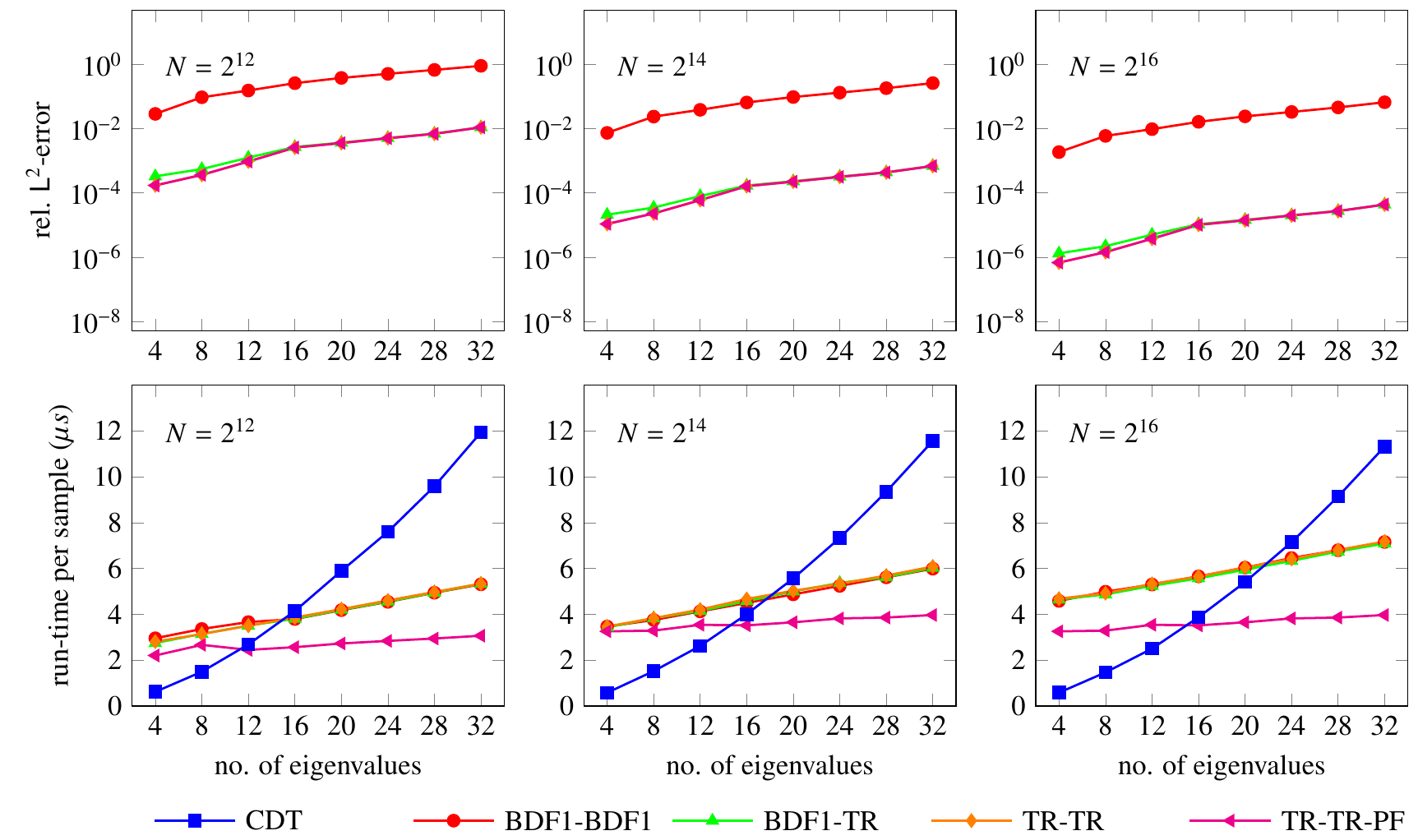}
\caption{The figure depicts the relative numerical error (top row) and the run-time (bottom row) 
for multi-solitons as a function of number of eigenvalues. The numerical test (as described in
Sec.~\ref{tc:ms}) is carried out with fixed number of samples 
($N\in\{2^{12},2^{14}, 2^{16}\}$).\label{fig:multi-soliton-err-convg-ev}}
\end{figure*}

\subsubsection{Multi-solitons}
\label{ex:multi-soliton}
Define a sequence of angles for $J\in\field{Z}_+$ by choosing
$\Delta\theta=(\pi-2\theta_0)/(J-1),\,\theta_0>0$, and 
\[
\theta_j=\theta_0 + (j-1)\Delta\theta,\,j=1,2,\ldots,J
\] 
so that $\theta_j\in[\theta_0,\pi-\theta_0]$. Then the eigenvalues for our numerical
experiment are chosen as
\begin{equation}
    \zeta_{j+J(l-1)}=le^{i\theta_j},\,l=1,2,\ldots,8,\,j=1,2,\ldots,J.
\end{equation}
The norming constants are chosen as
\begin{equation}
    b_j=e^{i\pi (j-1)/(8J-1)},\,j=1,2,\ldots,8J.    
\end{equation}
Here we choose, $\theta_0=\pi/3$ and $J=4$. Then we consider a sequence
of discrete spectra defined as
\begin{equation}
\label{eq:ds-test-ms}
\mathfrak{S}_{K} = \{(\zeta_k, b_k),\,k=1,2,\ldots,K\},
\end{equation}
where $K=4,8,\ldots,32$. Let $\mathfrak{E}_K$ be the set of all the eigenvalues. The potential can be computed 
with machine precision using the CDT algorithm which is taken as the reference for 
error analysis in this case. For fixed $K$,
the eigenvalues are scaled by the scaling parameter $\kappa =
2(\sum_{k=0}^K\Im\zeta_k)^{1/2}$. Let
$\eta_{\text{min}}=\min_{\zeta\in\mathfrak{E}_K}\Im\zeta$, then the computational domain for this
example is chosen as $[-L,\,L]$ where $L={11\kappa}/\eta_{\text{min}}$.

\subsubsection{Secant-hyperbolic potential}
\label{ex:sech}
The exact solution of the ZS-problem for the secant-hyperbolic potential was
first reported in~\cite{SY1974}. We summarize the results required for our
purpose as follows: Let the potential be written as
\begin{equation}\label{eq:sech-pot}
q(x) = A \sech{x},
\end{equation}
where $A$ is referred to as the amplitude. The scattering coefficients are then given by
\begin{align*}
a(\zeta) &=\frac{\Gamma\biggl(\frac{1}{2}-i\zeta\biggl)^2}
    {\Gamma\biggl(A+\frac{1}{2}-i\zeta\biggl)\Gamma\biggl(-A+\frac{1}{2}-i\zeta\biggl)},\,\,
    \zeta\in\ovl{\field{C}}_+,\\
    b(\zeta) &= -\sin\pi A\sech\pi\zeta,\,\,|\Im{\zeta}|<{1}/{2}. 
\end{align*}
The eigenvalues are given by 
\[
    \zeta_k = i\biggl(\tilde{A}-k\biggl),\,\,k=1,2,\ldots,K,
\]
where $K$ is the largest integer smaller than $\tilde{A}=(A+1/2)$. Putting 
$\tilde{A}_f=\tilde{A}-K$, the non-integer part of $\tilde{A}$, the 
$a$-coefficient can be written as a product of solitonic and radiative parts as follows
\[
    a(\zeta)=\underbrace{\biggl(\prod_{k=1}^K\frac{\zeta-\zeta_k}{\zeta-\zeta^*_k}\biggl)}_{a_S(\zeta)}
\underbrace{\frac{\Gamma\biggl(\frac{1}{2}-i\zeta\biggl)^2}
{\Gamma\biggl(\tilde{A}_f-i\zeta\biggl)\Gamma\biggl(1-\tilde{A}_f-i\zeta\biggl)}}_{a_R(\zeta)}.
\]
Note that $a_R(\zeta)$ belongs to a secant-hyperbolic potential with amplitude
$A_R=\tilde{A}_f-1/2\,\,(>0)$. The corresponding norming constants are given by
$b_k=(-1)^{k}$.

This example allows one to test the CDT and the FDT algorithms where the seed
potential can be taken as $q_0(x)=A_R\sech(x)$ and the sequence of discrete
spectra to be added,
\begin{equation}
\mathfrak{S}_{K} = \{(\zeta_k, b_k),\,k=1,2,\ldots,K\},\,\,K=4,\ldots,32,
\end{equation}
where we set $A_R = 0.4$. Corresponding to $\mathfrak{S}_{K}$, the
amplitude of the augmented secant-hyperbolic potential is given by $A = 0.4 + K$ 
and $\tilde{A} = 0.9 + K$. As in the last example, for fixed $K$,
the eigenvalues are scaled by the scaling parameter given by 
$\kappa = 2(\sum_{k=0}^K\Im\zeta_k)^{1/2}$. 

In order to choose the computational
domain $[-L,\,L]$ for the sech-potential~\eqref{eq:sech-pot} with the
aforementioned scaling, we can 
use the relation~\eqref{eq:truncation-criteria}. Choosing 
$\eta_{\text{min}}=\min_{\zeta\in\mathfrak{E}_K}\Im\zeta$ where $\mathfrak{E}_K$
is the set of all the eigenvalues and observing that 
\[
\|q\chi_{(-\infty,-L/\kappa]}\|_{\fs{L}^1}=A\tan^{-1}\left[\frac{1}{\sinh(L/\kappa)}\right],
\]
we have $L\approx[\eta_{\text{min}}\log(2A/\epsilon)](\kappa/\eta_{\text{min}})$
which rounds to $L\approx30(\kappa/\eta_{\text{min}})$ for $\epsilon=10^{-12}$.

\subsection{Test cases}

\subsubsection{Discrete spectrum}
\label{tc:discrete-spectrum}
For multi-soliton potentials described in Sec.~\ref{ex:multi-soliton}, we test
the convergence behavior with regard to the discrete spectrum for 
various discretization schemes, namely, BDF1, TR, SM and MG1. For the
convergence behavior of the numerically computed norming constants, we assume
that the eigenvalues are known exactly. The error in the norming constants is
quantified by
\begin{equation}
e_{\text{rel}} 
=\sqrt{\frac{\sum_{k=1}^K|b^{(\text{num.})}_k-b_k|^2}{\sum_{k=1}^K|b_k|^2}}.
\end{equation}
For the convergence behavior with regard to the eigenvalues, we compute
$a^{(\text{num.})}(\zeta_k)$ where the 
$a$-coefficient is computed numerically. The error is then quantified by
\begin{equation}
e_{\text{rel}} 
=\lim_{\eta\rightarrow\infty}\sqrt{\frac{1}{K}\sum_{k=1}^K
    \biggl|\frac{a^{(\text{num.})}(\zeta_k)}{a^{(\text{num.})}(i\eta)}\biggl|^2}.
\end{equation}
For MG1, the limit is evaluated by setting $\eta=100$. For others
$\lim_{\eta\rightarrow\infty}a(i\eta)=P^{(N)}_{1,0}$. Except for MG1, all
other schemes are implemented using the fast forward scattering algorithm 
(see Sec.~\ref{sec:FNFT}). The computation of the norming constant is discussed
in Sec.~\ref{sec:eig-and-nconst} and Sec.~\ref{sec:benchmark}.

\begin{figure}[!ht]
\centering
\includegraphics[width=0.48\textwidth]{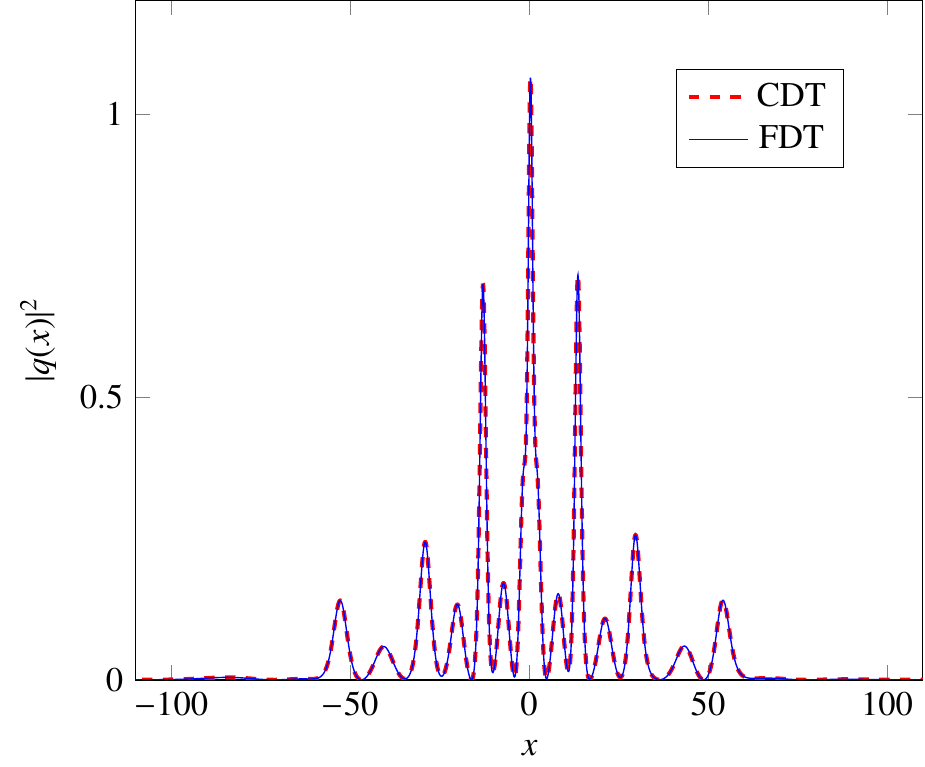}%
\caption{The figure compares the multi-soliton 
potential corresponding to the discrete spectrum, $\mathfrak{S}_{16}$ (defined
by~\eqref{eq:ds-test-ms}), generated using the CDT and the FDT algorithm after
scaling as discussed in Sec.~\ref{ex:sech}.}
\end{figure}

\begin{figure}[!ht]
\centering
\includegraphics[width=0.48\textwidth]{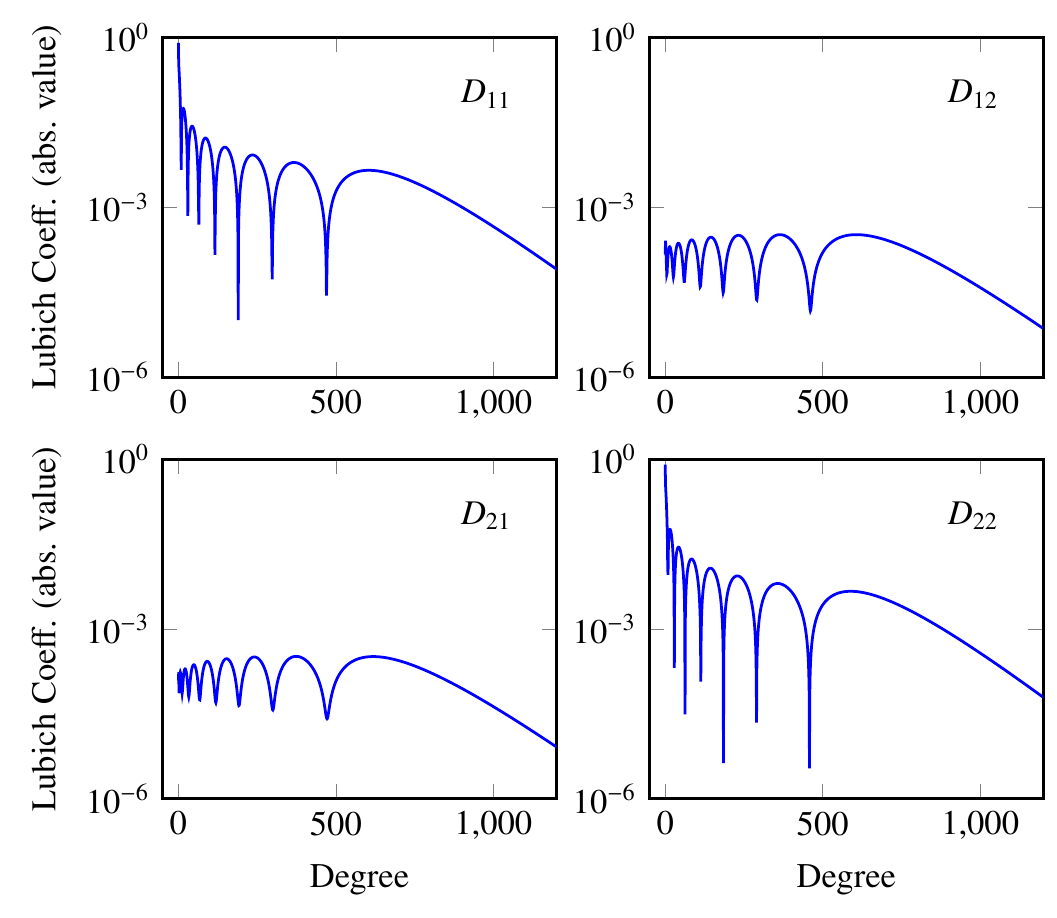}%
\caption{The figure shows the absolute value of the Lubich coefficients for the 
Darboux matrix elements corresponding to the discrete spectrum, $\mathfrak{S}_{16}$ (defined
by~\eqref{eq:ds-test-ms}). The underlying one-step method for the Lubich method here 
is the trapezoidal rule (TR).\label{fig:lubich-ms}}
\end{figure}

\begin{figure*}[!t]
\centering
\includegraphics[scale=1]{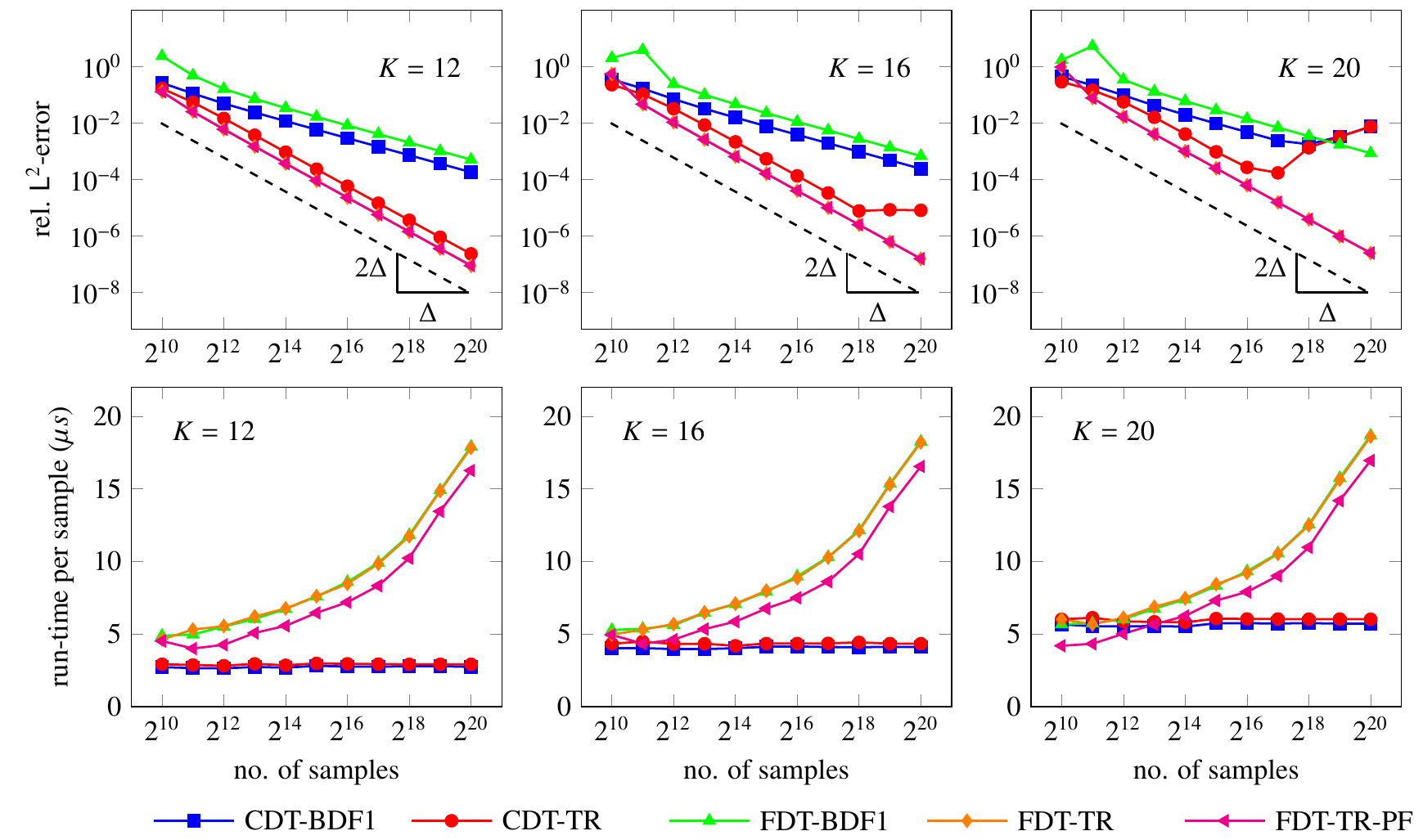}
\caption{The figure depicts convergence analysis (top row) and run-time behavior
for general Darboux transformation. The numerical test is carried out with a
soliton-free sech-potential as seed and fixed number of eigenvalues ($K\in\{12, 16, 20\}$) to be
added to the profile (see Sec.~\ref{tc:general-DT}).\label{fig:general-DT-err-convg}}
\end{figure*}
\subsubsection{Multi-soliton potential}
\label{tc:ms}
In this test case, we carry out the convergence analysis and a comparison of
run-time (per sample) of different variants of the FDT algorithm for 
multi-solitons as described in Sec.~\ref{ex:multi-soliton}. Note that
the CDT algorithm in this case gives the exact potential which allows us to
compute the total numerical error for the FDT algorithm for arbitrary discrete
spectra. The error is quantified
by
\begin{equation}
    \label{eq:rel-L2-error}
    e_{\text{rel}} = \frac{\|q^{(\text{num.})}-q\|_{\fs{L}^2}}
    {\|q\|_{\fs{L}^2}},
\end{equation}
where the integrals are evaluated numerically using the trapezoidal rule. 

The different variants of the FDT algorithm are described as follows: any 
one-step method for the ZS-problem can be combined with any one-step method for
the Lubich coefficients to obtain the FDT algorithm. In particular the relevant
combinations are: BDF1-BDF1, BDF1-TR and TR-TR. We also consider the 
partial-fraction variant of the TR-TR combination which is labeled as TR-TR-PF.
Note that the combination of a first order method for the ZS-problem with a second order
method for the Lubich coefficients or vice versa should lead to a first order
FDT algorithm. A second order method for the ZS-problem must be combined with a second order
method or higher for the Lubich coefficients in order to obtain a second order
FDT algorithm.

Parameters for the Lubich method are as follows: $M=8N$ and
$N_{\text{th}}=N/8$ (for the PF-variant). For the Cauchy integral, the radius of
the circular contour is $\varrho = \exp[-8/(N/2-1)]$.

\subsubsection{General Darboux transformation}
\label{tc:general-DT}
In this test case, we carry out the convergence analysis and a comparison of
run-time (per sample) of different variants of the CDT/FDT algorithm for 
the secant-hyperbolic potential as described in Sec.~\ref{ex:sech}. Note that, 
in the case of the secant-hyperbolic potential, the soliton-free seed potential 
as well as the augmented potential can be stated in a closed form.

The variants of 
the CDT/FDT algorithm are determined by the
underlying one-step method. Unlike in the last test case (Sec.~\ref{tc:ms}),
Lubich method uses the same one-step method as that of the ZS-problem. The 
total numerical error is quantified by \eqref{eq:rel-L2-error}. Parameters 
for the Lubich method are the same as in the last test case.

\begin{figure*}[!t]
\centering
\includegraphics[scale=1]{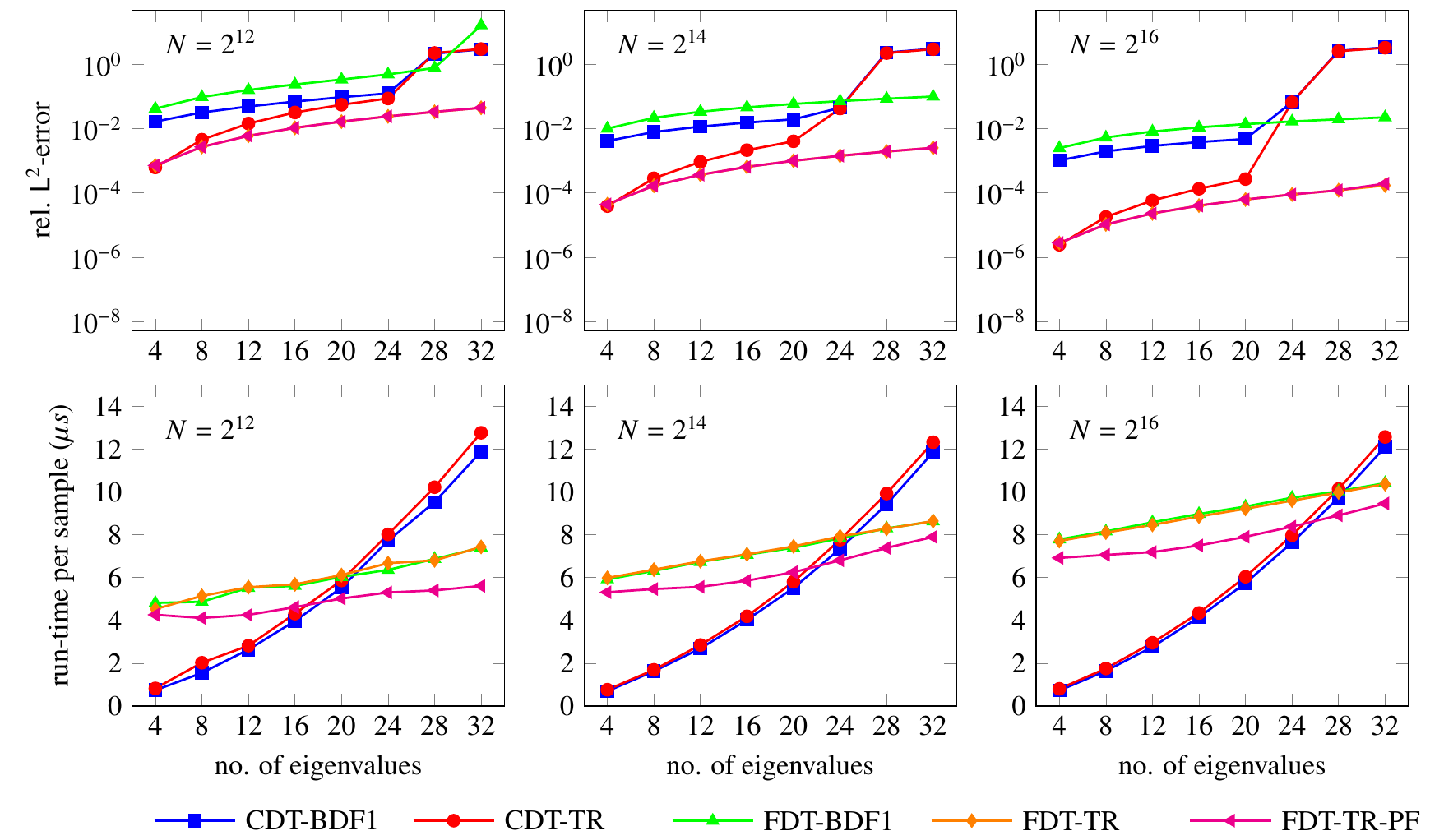}
\caption{
The figure depicts the relative numerical error (top row) and the run-time 
(bottom row) for the general Darboux transform as a function of number of 
eigenvalues. The numerical test (as described in
Sec.~\ref{tc:general-DT}) is carried out with fixed number of
samples ($N\in\{2^{12},2^{14}, 2^{16}\}$).\label{fig:general-DT-err-convg-ev}}
\end{figure*}
\subsection{Results and discussion}

\subsubsection{Discrete spectrum}
\label{re:discrete-spectrum}
For a given multi-soliton, this test case was designed to assess the performance
of the discretization schemes, namely, BDF1, TR, SM and MG1 in the determination
of the discrete spectrum. The results are plotted in
Fig.~\ref{fig:discrete-spectrum} where it can be easily seen that all the
methods considered show convergence at a rate that is 
determined by the underlying one-step method. However, the rate 
of convergence of BDF1 with regard to discrete eigenvalues seems to be better than
expected as evident from the plots in the bottom row of Fig~\ref{fig:discrete-spectrum}. The 
overall accuracy of MG1 is evidently superior to that of others while TR turns 
out to be a close second.

\subsubsection{Multi-soliton potential}
\label{re:ms}
This test case was designed to study the convergence and run-time behavior of
different variants of the FDT algorithm for multi-solitons. The results for 
fixed number of eigenvalues and varying number of samples is shown in 
Fig.~\ref{fig:multi-soliton-err-convg}. The second order 
of convergence of the schemes BDF1-TR, TR-TR and TR-TR-PF can be identified from
the plots in the top row of Fig.~\ref{fig:multi-soliton-err-convg}. The scheme
BDF1-TR shows better rate of convergence than expected. 

The run-time behavior of CDT for fixed number of eigenvalues is clearly
superior to all of the method as evident from
Fig.~\ref{fig:multi-soliton-err-convg} (bottom row). The scheme TR-PF however
becomes very competitive with the CDT algorithm.

The relative error and the run-time (per sample) as a function of the 
number of eigenvalues keeping the number of samples fixed is shown in
Fig.~\ref{fig:multi-soliton-err-convg-ev}. Here, the FDT algorithm outperforms
the CDT algorithm with TR-TR-PF variant being the fastest as evident from the
plots in the bottom row of Fig.~\ref{fig:multi-soliton-err-convg-ev}. It is
interesting to note that the relative error as a function of number of eigenvalues as
shown in plots at the top row of Fig.~\ref{fig:multi-soliton-err-convg-ev}
exhibits exponentially increasing behavior. This puts an upper limit to the
number of eigenvalues that can be handled with the FDT algorithm within a given
precision\footnote{Note that the CDT algorithm also suffers from this drawback. However, in
order to determine the upper limit for the CDT algorithm, one requires an
implementation which employs a variable precision arithmetic. This program is
not followed in this article.}.

\subsubsection{General Darboux transformation}
\label{re:general-DT}
This test case was designed to study the convergence and run-time behavior of
different variants of the CDT/FDT algorithm for a soliton-free seed potential. The 
results for fixed number of eigenvalues (that are meant to be added) 
and varying number of samples is shown in Fig.~\ref{fig:general-DT-err-convg}. The 
second order of convergence of the TR variant of the CDT/FDT algorithm can be 
identified from the plots in the top row of Fig.~\ref{fig:general-DT-err-convg}. However, the
TR variant of the CDT algorithm performs not only worse as compared to that of
the FDT algorithm but it also becomes unstable with increasing number of 
eigenvalues. Further, unlike the CDT algorithm, the BDF1 and TR variant of FDT
shows convergence (at an expected rate) with increasing number of samples.

The run-time behavior of CDT for fixed number of eigenvalues is clearly
superior to that of FDT as evident from
Fig.~\ref{fig:general-DT-err-convg-ev} (bottom row). The scheme TR-PF however
becomes very competitive to the CDT-TR variant. Note that CDT in this case is
reliable \emph{only} for small number of eigenvalues.

The relative error and the run-time (per sample) as a function of the number
of eigenvalues keeping the number of samples fixed is shown in
Fig.~\ref{fig:general-DT-err-convg-ev}. Here, the FDT algorithm outperforms
the CDT algorithm with TR-TR-PF variant being the fastest as evident from the
plots in the bottom row of Fig.~\ref{fig:general-DT-err-convg-ev}. As in the last
test case, the relative error as a function of number of eigenvalues as
shown in plots at the top row of Fig.~\ref{fig:general-DT-err-convg-ev}
exhibits exponentially increasing behavior. Note that FDT not only
outperforms CDT in terms of accuracy, it also exhibits superior numerical
conditioning with increasing number of eigenvalues as evident from
Fig.~\ref{fig:general-DT-err-convg-ev} (top row).

\section{Conclusion}
\label{sec:conclusion}
To conclude, we have presented a systematic approach to discretize the non-Hermitian 
Zakharov-Shabat (ZS) problem which is based on exponential one-step methods. The discrete 
framework thus obtained is amenable to FFT-based fast polynomial arithmetic and
also admits of a layer-peeling property. In this setting we have presented different
variants of a fast forward/inverse SU(2)-nonlinear Fourier transformation (NFT) 
algorithm. As a first step
towards developing a general fast inverse NFT, we have presented several ways to obtain 
a fast Darboux transformation (FDT) algorithm with an operational complexity of 
$\bigO{KN+N\log^2N}$ where $K$ is the number of eigenvalues to be added to 
a seed potential and $N$ is the number of samples of the potential. This
algorithm exhibits an order of convergence that matches the underlying
exponential one-step method. In particular, if one uses the \emph{trapezoidal
rule} of integration, the order of convergence is $\bigO{N^{-2}}$. The strength
of this algorithm was demonstrated by exhaustive numerical tests where we could
successfully add $32$ eigenvalues to a soliton-free seed potential. 
It must be noted that the FDT algorithm shows a promising route to a fast inverse NFT 
which is confirmed empirically in~\cite{VW2017OFC}--this forms the subject
matter of a sequel to this paper. 

Furthermore, we have also presented a second approach that
naively tries to mimic the classical Darboux transformation (CDT) scheme in the 
discrete framework developed for the ZS-problem with an arbitrary seed potential. This 
algorithm affords a complexity of $\bigO{K^2N}$; however, it turns out to be less 
accurate and numerically unstable beyond certain number of eigenvalues.

Finally, let us emphasize that, based on the ideas presented in this paper and
drawing on the pioneering work of Lubich on convolution quadrature, it seems 
plausible to anticipate the existence of higher-order convergent fast
forward/inverse NFT algorithms using (exponential) linear multistep 
methods--we hope to return to this theme in the future.


\begin{thebibliography}{68}%
\makeatletter
\providecommand \@ifxundefined [1]{%
 \@ifx{#1\undefined}
}%
\providecommand \@ifnum [1]{%
 \ifnum #1\expandafter \@firstoftwo
 \else \expandafter \@secondoftwo
 \fi
}%
\providecommand \@ifx [1]{%
 \ifx #1\expandafter \@firstoftwo
 \else \expandafter \@secondoftwo
 \fi
}%
\providecommand \natexlab [1]{#1}%
\providecommand \enquote  [1]{``#1''}%
\providecommand \bibnamefont  [1]{#1}%
\providecommand \bibfnamefont [1]{#1}%
\providecommand \citenamefont [1]{#1}%
\providecommand \href@noop [0]{\@secondoftwo}%
\providecommand \href [0]{\begingroup \@sanitize@url \@href}%
\providecommand \@href[1]{\@@startlink{#1}\@@href}%
\providecommand \@@href[1]{\endgroup#1\@@endlink}%
\providecommand \@sanitize@url [0]{\catcode `\\12\catcode `\$12\catcode
  `\&12\catcode `\#12\catcode `\^12\catcode `\_12\catcode `\%12\relax}%
\providecommand \@@startlink[1]{}%
\providecommand \@@endlink[0]{}%
\providecommand \url  [0]{\begingroup\@sanitize@url \@url }%
\providecommand \@url [1]{\endgroup\@href {#1}{\urlprefix }}%
\providecommand \urlprefix  [0]{URL }%
\providecommand \Eprint [0]{\href }%
\providecommand \doibase [0]{http://dx.doi.org/}%
\providecommand \selectlanguage [0]{\@gobble}%
\providecommand \bibinfo  [0]{\@secondoftwo}%
\providecommand \bibfield  [0]{\@secondoftwo}%
\providecommand \translation [1]{[#1]}%
\providecommand \BibitemOpen [0]{}%
\providecommand \bibitemStop [0]{}%
\providecommand \bibitemNoStop [0]{.\EOS\space}%
\providecommand \EOS [0]{\spacefactor3000\relax}%
\providecommand \BibitemShut  [1]{\csname bibitem#1\endcsname}%
\let\auto@bib@innerbib\@empty
\bibitem [{\citenamefont {Zakharov}\ and\ \citenamefont
  {Shabat}(1972)}]{ZS1972}%
  \BibitemOpen
  \bibfield  {author} {\bibinfo {author} {\bibfnamefont {V.~E.}\ \bibnamefont
  {Zakharov}}\ and\ \bibinfo {author} {\bibfnamefont {A.~B.}\ \bibnamefont
  {Shabat}},\ }\href@noop {} {\bibfield  {journal} {\bibinfo  {journal} {Sov.
  Phys. JETP}\ }\textbf {\bibinfo {volume} {34}},\ \bibinfo {pages} {62}
  (\bibinfo {year} {1972})}\BibitemShut {NoStop}%
\bibitem [{\citenamefont {Ablowitz}\ \emph {et~al.}(1974)\citenamefont
  {Ablowitz}, \citenamefont {Kaup}, \citenamefont {Newell},\ and\ \citenamefont
  {Segur}}]{AKNS1974}%
  \BibitemOpen
  \bibfield  {author} {\bibinfo {author} {\bibfnamefont {M.~J.}\ \bibnamefont
  {Ablowitz}}, \bibinfo {author} {\bibfnamefont {D.~J.}\ \bibnamefont {Kaup}},
  \bibinfo {author} {\bibfnamefont {A.~C.}\ \bibnamefont {Newell}}, \ and\
  \bibinfo {author} {\bibfnamefont {H.}~\bibnamefont {Segur}},\ }\href
  {\doibase 10.1002/sapm1974534249} {\bibfield  {journal} {\bibinfo  {journal}
  {Studies in Applied Mathematics}\ }\textbf {\bibinfo {volume} {53}},\
  \bibinfo {pages} {249} (\bibinfo {year} {1974})}\BibitemShut {NoStop}%
\bibitem [{\citenamefont {Ablowitz}\ and\ \citenamefont
  {Segur}(1981)}]{AS1981}%
  \BibitemOpen
  \bibfield  {author} {\bibinfo {author} {\bibfnamefont {M.}~\bibnamefont
  {Ablowitz}}\ and\ \bibinfo {author} {\bibfnamefont {H.}~\bibnamefont
  {Segur}},\ }\href {\doibase 10.1137/1.9781611970883} {\emph {\bibinfo {title}
  {Solitons and the Inverse Scattering Transform}}}\ (\bibinfo  {publisher}
  {Society for Industrial and Applied Mathematics},\ \bibinfo {year}
  {1981})\BibitemShut {NoStop}%
\bibitem [{\citenamefont {Kodama}\ and\ \citenamefont
  {Hasegawa}(1987)}]{HK1987}%
  \BibitemOpen
  \bibfield  {author} {\bibinfo {author} {\bibfnamefont {Y.}~\bibnamefont
  {Kodama}}\ and\ \bibinfo {author} {\bibfnamefont {A.}~\bibnamefont
  {Hasegawa}},\ }\href {\doibase 10.1109/JQE.1987.1073392} {\bibfield
  {journal} {\bibinfo  {journal} {Quantum Electronics, IEEE Journal of}\
  }\textbf {\bibinfo {volume} {23}},\ \bibinfo {pages} {510} (\bibinfo {year}
  {1987})}\BibitemShut {NoStop}%
\bibitem [{\citenamefont {Agrawal}(2013)}]{Agrawal2013}%
  \BibitemOpen
  \bibfield  {author} {\bibinfo {author} {\bibfnamefont {G.}~\bibnamefont
  {Agrawal}},\ }\href@noop {} {\emph {\bibinfo {title} {Nonlinear Fiber
  Optics}}},\ Academic Press\ (\bibinfo  {publisher} {Academic Press},\
  \bibinfo {year} {2013})\BibitemShut {NoStop}%
\bibitem [{\citenamefont {Hasegawa}\ and\ \citenamefont
  {Kodama}(1990)}]{HK1990GC}%
  \BibitemOpen
  \bibfield  {author} {\bibinfo {author} {\bibfnamefont {A.}~\bibnamefont
  {Hasegawa}}\ and\ \bibinfo {author} {\bibfnamefont {Y.}~\bibnamefont
  {Kodama}},\ }\href {\doibase 10.1364/OL.15.001443} {\bibfield  {journal}
  {\bibinfo  {journal} {Opt. Lett.}\ }\textbf {\bibinfo {volume} {15}},\
  \bibinfo {pages} {1443} (\bibinfo {year} {1990})}\BibitemShut {NoStop}%
\bibitem [{\citenamefont {Hasegawa}\ and\ \citenamefont
  {Kodama}(1991)}]{HK1991GC}%
  \BibitemOpen
  \bibfield  {author} {\bibinfo {author} {\bibfnamefont {A.}~\bibnamefont
  {Hasegawa}}\ and\ \bibinfo {author} {\bibfnamefont {Y.}~\bibnamefont
  {Kodama}},\ }\href {\doibase 10.1103/PhysRevLett.66.161} {\bibfield
  {journal} {\bibinfo  {journal} {Phys. Rev. Lett.}\ }\textbf {\bibinfo
  {volume} {66}},\ \bibinfo {pages} {161} (\bibinfo {year} {1991})}\BibitemShut
  {NoStop}%
\bibitem [{\citenamefont {Turitsyn}\ \emph {et~al.}(2012)\citenamefont
  {Turitsyn}, \citenamefont {Bale},\ and\ \citenamefont {Fedoruk}}]{TBF2012}%
  \BibitemOpen
  \bibfield  {author} {\bibinfo {author} {\bibfnamefont {S.~K.}\ \bibnamefont
  {Turitsyn}}, \bibinfo {author} {\bibfnamefont {B.~G.}\ \bibnamefont {Bale}},
  \ and\ \bibinfo {author} {\bibfnamefont {M.~P.}\ \bibnamefont {Fedoruk}},\
  }\href {\doibase 10.1016/j.physrep.2012.09.004} {\bibfield  {journal}
  {\bibinfo  {journal} {Physics Reports}\ }\textbf {\bibinfo {volume} {521}},\
  \bibinfo {pages} {135} (\bibinfo {year} {2012})},\ \bibinfo {note}
  {dispersion-Managed Solitons in Fibre Systems and Lasers}\BibitemShut
  {NoStop}%
\bibitem [{\citenamefont {Feced}\ and\ \citenamefont {Zervas}(2000)}]{FZ2000}%
  \BibitemOpen
  \bibfield  {author} {\bibinfo {author} {\bibfnamefont {R.}~\bibnamefont
  {Feced}}\ and\ \bibinfo {author} {\bibfnamefont {M.~N.}\ \bibnamefont
  {Zervas}},\ }\href {\doibase 10.1364/JOSAA.17.001FZ2000573} {\bibfield
  {journal} {\bibinfo  {journal} {J. Opt. Soc. Am. A}\ }\textbf {\bibinfo
  {volume} {17}},\ \bibinfo {pages} {1573} (\bibinfo {year}
  {2000})}\BibitemShut {NoStop}%
\bibitem [{\citenamefont {Brenne}\ and\ \citenamefont {Skaar}(2003)}]{BS2003}%
  \BibitemOpen
  \bibfield  {author} {\bibinfo {author} {\bibfnamefont {J.~K.}\ \bibnamefont
  {Brenne}}\ and\ \bibinfo {author} {\bibfnamefont {J.}~\bibnamefont {Skaar}},\
  }\href@noop {} {\bibfield  {journal} {\bibinfo  {journal} {J. Lightwave
  Technol.}\ }\textbf {\bibinfo {volume} {21}},\ \bibinfo {pages} {254}
  (\bibinfo {year} {2003})}\BibitemShut {NoStop}%
\bibitem [{\citenamefont {Rourke}\ and\ \citenamefont
  {Morris}(1992{\natexlab{a}})}]{RM1992-NMR}%
  \BibitemOpen
  \bibfield  {author} {\bibinfo {author} {\bibfnamefont {D.~E.}\ \bibnamefont
  {Rourke}}\ and\ \bibinfo {author} {\bibfnamefont {P.~G.}\ \bibnamefont
  {Morris}},\ }\href {\doibase 10.1016/0022-2364(92)90159-5} {\bibfield
  {journal} {\bibinfo  {journal} {Journal of Magnetic Resonance (1969)}\
  }\textbf {\bibinfo {volume} {99}},\ \bibinfo {pages} {118 } (\bibinfo {year}
  {1992}{\natexlab{a}})}\BibitemShut {NoStop}%
\bibitem [{\citenamefont {Rourke}\ and\ \citenamefont
  {Morris}(1992{\natexlab{b}})}]{RM1992}%
  \BibitemOpen
  \bibfield  {author} {\bibinfo {author} {\bibfnamefont {D.~E.}\ \bibnamefont
  {Rourke}}\ and\ \bibinfo {author} {\bibfnamefont {P.~G.}\ \bibnamefont
  {Morris}},\ }\href {\doibase 10.1103/PhysRevA.46.3631} {\bibfield  {journal}
  {\bibinfo  {journal} {Phys. Rev. A}\ }\textbf {\bibinfo {volume} {46}},\
  \bibinfo {pages} {3631} (\bibinfo {year} {1992}{\natexlab{b}})}\BibitemShut
  {NoStop}%
\bibitem [{\citenamefont {Rourke}\ and\ \citenamefont
  {Saunders}(1994)}]{RS1994}%
  \BibitemOpen
  \bibfield  {author} {\bibinfo {author} {\bibfnamefont {D.~E.}\ \bibnamefont
  {Rourke}}\ and\ \bibinfo {author} {\bibfnamefont {J.~K.}\ \bibnamefont
  {Saunders}},\ }\href {\doibase 10.1063/1.530617} {\bibfield  {journal}
  {\bibinfo  {journal} {Journal of Mathematical Physics}\ }\textbf {\bibinfo
  {volume} {35}},\ \bibinfo {pages} {848} (\bibinfo {year} {1994})}\BibitemShut
  {NoStop}%
\bibitem [{\citenamefont {Hasegawa}\ and\ \citenamefont {Nyu}(1993)}]{HN1993}%
  \BibitemOpen
  \bibfield  {author} {\bibinfo {author} {\bibfnamefont {A.}~\bibnamefont
  {Hasegawa}}\ and\ \bibinfo {author} {\bibfnamefont {T.}~\bibnamefont {Nyu}},\
  }\href {\doibase 10.1109/50.219570} {\bibfield  {journal} {\bibinfo
  {journal} {J. Lightwave Technol.}\ }\textbf {\bibinfo {volume} {11}},\
  \bibinfo {pages} {395} (\bibinfo {year} {1993})}\BibitemShut {NoStop}%
\bibitem [{\citenamefont {Yousefi}\ and\ \citenamefont
  {Kschischang}(2014)}]{Yousefi2014compact}%
  \BibitemOpen
  \bibfield  {author} {\bibinfo {author} {\bibfnamefont {M.~I.}\ \bibnamefont
  {Yousefi}}\ and\ \bibinfo {author} {\bibfnamefont {F.~R.}\ \bibnamefont
  {Kschischang}},\ }\href@noop {} {\bibfield  {journal} {\bibinfo  {journal}
  {IEEE Trans. Inf. Theory}\ }\textbf {\bibinfo {volume} {60}},\ \bibinfo
  {pages} {4312} (\bibinfo {year} {2014})}\BibitemShut {NoStop}%
\bibitem [{\citenamefont {Aref}\ \emph {et~al.}(2016)\citenamefont {Aref},
  \citenamefont {Le},\ and\ \citenamefont {B\"uelow}}]{ALB2016}%
  \BibitemOpen
  \bibfield  {author} {\bibinfo {author} {\bibfnamefont {V.}~\bibnamefont
  {Aref}}, \bibinfo {author} {\bibfnamefont {S.~T.}\ \bibnamefont {Le}}, \ and\
  \bibinfo {author} {\bibfnamefont {H.}~\bibnamefont {B\"uelow}},\ }in\
  \href@noop {} {\emph {\bibinfo {booktitle} {ECOC 2016 - Post Deadline Paper;
  42nd European Conference on Optical Communication}}}\ (\bibinfo {year}
  {2016})\ pp.\ \bibinfo {pages} {1--3}\BibitemShut {NoStop}%
\bibitem [{\citenamefont {Turitsyn}\ \emph {et~al.}(2017)\citenamefont
  {Turitsyn}, \citenamefont {Prilepsky}, \citenamefont {Le}, \citenamefont
  {Wahls}, \citenamefont {Frumin}, \citenamefont {Kamalian},\ and\
  \citenamefont {Derevyanko}}]{TPLWFK2017}%
  \BibitemOpen
  \bibfield  {author} {\bibinfo {author} {\bibfnamefont {S.~K.}\ \bibnamefont
  {Turitsyn}}, \bibinfo {author} {\bibfnamefont {J.~E.}\ \bibnamefont
  {Prilepsky}}, \bibinfo {author} {\bibfnamefont {S.~T.}\ \bibnamefont {Le}},
  \bibinfo {author} {\bibfnamefont {S.}~\bibnamefont {Wahls}}, \bibinfo
  {author} {\bibfnamefont {L.~L.}\ \bibnamefont {Frumin}}, \bibinfo {author}
  {\bibfnamefont {M.}~\bibnamefont {Kamalian}}, \ and\ \bibinfo {author}
  {\bibfnamefont {S.~A.}\ \bibnamefont {Derevyanko}},\ }\href {\doibase
  10.1364/OPTICA.4.000307} {\bibfield  {journal} {\bibinfo  {journal} {Optica}\
  }\textbf {\bibinfo {volume} {4}},\ \bibinfo {pages} {307} (\bibinfo {year}
  {2017})}\BibitemShut {NoStop}%
\bibitem [{\citenamefont {Cartledge}\ \emph {et~al.}(2017)\citenamefont
  {Cartledge}, \citenamefont {Guiomar}, \citenamefont {Kschischang},
  \citenamefont {Liga},\ and\ \citenamefont {Yankov}}]{CGKLY2017}%
  \BibitemOpen
  \bibfield  {author} {\bibinfo {author} {\bibfnamefont {J.~C.}\ \bibnamefont
  {Cartledge}}, \bibinfo {author} {\bibfnamefont {F.~P.}\ \bibnamefont
  {Guiomar}}, \bibinfo {author} {\bibfnamefont {F.~R.}\ \bibnamefont
  {Kschischang}}, \bibinfo {author} {\bibfnamefont {G.}~\bibnamefont {Liga}}, \
  and\ \bibinfo {author} {\bibfnamefont {M.~P.}\ \bibnamefont {Yankov}},\
  }\href {\doibase 10.1364/OE.25.001916} {\bibfield  {journal} {\bibinfo
  {journal} {Opt. Express}\ }\textbf {\bibinfo {volume} {25}},\ \bibinfo
  {pages} {1916} (\bibinfo {year} {2017})}\BibitemShut {NoStop}%
\bibitem [{\citenamefont {Dar}\ and\ \citenamefont {Winzer}(2017)}]{DW2017}%
  \BibitemOpen
  \bibfield  {author} {\bibinfo {author} {\bibfnamefont {R.}~\bibnamefont
  {Dar}}\ and\ \bibinfo {author} {\bibfnamefont {P.~J.}\ \bibnamefont
  {Winzer}},\ }\href@noop {} {\bibfield  {journal} {\bibinfo  {journal} {J.
  Lightwave Technol.}\ }\textbf {\bibinfo {volume} {35}},\ \bibinfo {pages}
  {903} (\bibinfo {year} {2017})}\BibitemShut {NoStop}%
\bibitem [{\citenamefont {Vaibhav}\ and\ \citenamefont
  {Wahls}(2017)}]{VW2017OFC}%
  \BibitemOpen
  \bibfield  {author} {\bibinfo {author} {\bibfnamefont {V.}~\bibnamefont
  {Vaibhav}}\ and\ \bibinfo {author} {\bibfnamefont {S.}~\bibnamefont
  {Wahls}},\ }in\ \href {\doibase 10.1364/OFC.2017.Tu3D.2} {\emph {\bibinfo
  {booktitle} {Optical Fiber Communication Conference}}}\ (\bibinfo
  {publisher} {Optical Society of America},\ \bibinfo {year} {2017})\ p.\
  \bibinfo {pages} {Tu3D.2}\BibitemShut {NoStop}%
\bibitem [{\citenamefont {Lin}(1990)}]{Lin1990}%
  \BibitemOpen
  \bibfield  {author} {\bibinfo {author} {\bibfnamefont {J.}~\bibnamefont
  {Lin}},\ }\href {\doibase 10.1007/BF02015338} {\bibfield  {journal} {\bibinfo
   {journal} {Acta Mathematicae Applicatae Sinica}\ }\textbf {\bibinfo {volume}
  {6}},\ \bibinfo {pages} {308} (\bibinfo {year} {1990})}\BibitemShut {NoStop}%
\bibitem [{\citenamefont {Gu}\ \emph {et~al.}(2005)\citenamefont {Gu},
  \citenamefont {Hu},\ and\ \citenamefont {Zhou}}]{GHZ2005}%
  \BibitemOpen
  \bibfield  {author} {\bibinfo {author} {\bibfnamefont {C.}~\bibnamefont
  {Gu}}, \bibinfo {author} {\bibfnamefont {A.}~\bibnamefont {Hu}}, \ and\
  \bibinfo {author} {\bibfnamefont {Z.}~\bibnamefont {Zhou}},\ }\href {\doibase
  10.1007/1-4020-3088-6} {\emph {\bibinfo {title} {Darboux Transformations in
  Integrable Systems: Theory and their Applications to Geometry}}},\
  Mathematical Physics Studies\ (\bibinfo  {publisher} {Springer Netherlands},\
  \bibinfo {year} {2005})\BibitemShut {NoStop}%
\bibitem [{\citenamefont {Hari}\ \emph {et~al.}(2014)\citenamefont {Hari},
  \citenamefont {Kschischang},\ and\ \citenamefont {Yousefi}}]{HKY2014}%
  \BibitemOpen
  \bibfield  {author} {\bibinfo {author} {\bibfnamefont {S.}~\bibnamefont
  {Hari}}, \bibinfo {author} {\bibfnamefont {F.}~\bibnamefont {Kschischang}}, \
  and\ \bibinfo {author} {\bibfnamefont {M.}~\bibnamefont {Yousefi}},\ }in\
  \href {\doibase 10.1109/QBSC.2014.6841191} {\emph {\bibinfo {booktitle}
  {Communications (QBSC), 2014 27th Biennial Symposium on}}}\ (\bibinfo {year}
  {2014})\ pp.\ \bibinfo {pages} {92--95}\BibitemShut {NoStop}%
\bibitem [{\citenamefont {Dong}\ \emph {et~al.}(2015)\citenamefont {Dong},
  \citenamefont {Hari}, \citenamefont {Gui}, \citenamefont {Zhong},
  \citenamefont {Yousefi}, \citenamefont {Lu}, \citenamefont {Wai},
  \citenamefont {Kschischang},\ and\ \citenamefont {Lau}}]{DHGZYLWKL2015}%
  \BibitemOpen
  \bibfield  {author} {\bibinfo {author} {\bibfnamefont {Z.}~\bibnamefont
  {Dong}}, \bibinfo {author} {\bibfnamefont {S.}~\bibnamefont {Hari}}, \bibinfo
  {author} {\bibfnamefont {T.}~\bibnamefont {Gui}}, \bibinfo {author}
  {\bibfnamefont {K.}~\bibnamefont {Zhong}}, \bibinfo {author} {\bibfnamefont
  {M.~I.}\ \bibnamefont {Yousefi}}, \bibinfo {author} {\bibfnamefont
  {C.}~\bibnamefont {Lu}}, \bibinfo {author} {\bibfnamefont {P.-K.~A.}\
  \bibnamefont {Wai}}, \bibinfo {author} {\bibfnamefont {F.~R.}\ \bibnamefont
  {Kschischang}}, \ and\ \bibinfo {author} {\bibfnamefont {A.~P.~T.}\
  \bibnamefont {Lau}},\ }\href {\doibase 10.1109/LPT.2015.2432793} {\bibfield
  {journal} {\bibinfo  {journal} {Photonics Technology Letters, IEEE}\ }\textbf
  {\bibinfo {volume} {27}},\ \bibinfo {pages} {1621} (\bibinfo {year}
  {2015})}\BibitemShut {NoStop}%
\bibitem [{\citenamefont {Hari}\ and\ \citenamefont
  {Kschischang}(2016)}]{HK2016}%
  \BibitemOpen
  \bibfield  {author} {\bibinfo {author} {\bibfnamefont {S.}~\bibnamefont
  {Hari}}\ and\ \bibinfo {author} {\bibfnamefont {F.~R.}\ \bibnamefont
  {Kschischang}},\ }\href {\doibase 10.1109/JLT.2016.2577702} {\bibfield
  {journal} {\bibinfo  {journal} {Journal of Lightwave Technology}\ }\textbf
  {\bibinfo {volume} {34}},\ \bibinfo {pages} {3529} (\bibinfo {year}
  {2016})}\BibitemShut {NoStop}%
\bibitem [{\citenamefont {Terauchi}\ and\ \citenamefont
  {Maruta}(2013)}]{TM2013}%
  \BibitemOpen
  \bibfield  {author} {\bibinfo {author} {\bibfnamefont {H.}~\bibnamefont
  {Terauchi}}\ and\ \bibinfo {author} {\bibfnamefont {A.}~\bibnamefont
  {Maruta}},\ }in\ \href@noop {} {\emph {\bibinfo {booktitle} {OECC/PS, 2013
  18th}}}\ (\bibinfo {year} {2013})\ pp.\ \bibinfo {pages} {1--2}\BibitemShut
  {NoStop}%
\bibitem [{\citenamefont {B\"ulow}(2014)}]{B2014}%
  \BibitemOpen
  \bibfield  {author} {\bibinfo {author} {\bibfnamefont {H.}~\bibnamefont
  {B\"ulow}},\ }in\ \href {\doibase 10.1109/ECOC.2014.6963840} {\emph {\bibinfo
  {booktitle} {Optical Communication (ECOC), 2014 European Conference on}}}\
  (\bibinfo {year} {2014})\ pp.\ \bibinfo {pages} {1--3}\BibitemShut {NoStop}%
\bibitem [{\citenamefont {B\"{u}low}(2015)}]{B2015}%
  \BibitemOpen
  \bibfield  {author} {\bibinfo {author} {\bibfnamefont {H.}~\bibnamefont
  {B\"{u}low}},\ }\href@noop {} {\bibfield  {journal} {\bibinfo  {journal} {J.
  Lightwave Technol.}\ }\textbf {\bibinfo {volume} {33}},\ \bibinfo {pages}
  {1433} (\bibinfo {year} {2015})}\BibitemShut {NoStop}%
\bibitem [{\citenamefont {Matsuda}\ \emph {et~al.}(2014)\citenamefont
  {Matsuda}, \citenamefont {Terauchi},\ and\ \citenamefont {Maruta}}]{MTM2014}%
  \BibitemOpen
  \bibfield  {author} {\bibinfo {author} {\bibfnamefont {Y.}~\bibnamefont
  {Matsuda}}, \bibinfo {author} {\bibfnamefont {H.}~\bibnamefont {Terauchi}}, \
  and\ \bibinfo {author} {\bibfnamefont {A.}~\bibnamefont {Maruta}},\ }in\
  \href@noop {} {\emph {\bibinfo {booktitle} {Optical Fibre Technology, 2014
  OptoElectronics and Communication Conference and Australian Conference on}}}\
  (\bibinfo {year} {2014})\ pp.\ \bibinfo {pages} {1016--1018}\BibitemShut
  {NoStop}%
\bibitem [{\citenamefont {Aref}\ \emph {et~al.}(2015)\citenamefont {Aref},
  \citenamefont {B{\"{u}}low}, \citenamefont {Schuh},\ and\ \citenamefont
  {Idler}}]{ABSI2015}%
  \BibitemOpen
  \bibfield  {author} {\bibinfo {author} {\bibfnamefont {V.}~\bibnamefont
  {Aref}}, \bibinfo {author} {\bibfnamefont {H.}~\bibnamefont {B{\"{u}}low}},
  \bibinfo {author} {\bibfnamefont {K.}~\bibnamefont {Schuh}}, \ and\ \bibinfo
  {author} {\bibfnamefont {W.}~\bibnamefont {Idler}},\ }\href@noop {} {\enquote
  {\bibinfo {title} {Experimental demonstration of nonlinear frequency division
  multiplexed transmission},}\ } (\bibinfo {year} {2015}),\ \bibinfo {note}
  {{arXiv}:1508.02577[cs.IT]}\BibitemShut {NoStop}%
\bibitem [{\citenamefont {B{\"{u}}low}\ \emph {et~al.}(2016)\citenamefont
  {B{\"{u}}low}, \citenamefont {Aref},\ and\ \citenamefont {Idler}}]{BAI2016}%
  \BibitemOpen
  \bibfield  {author} {\bibinfo {author} {\bibfnamefont {H.}~\bibnamefont
  {B{\"{u}}low}}, \bibinfo {author} {\bibfnamefont {V.}~\bibnamefont {Aref}}, \
  and\ \bibinfo {author} {\bibfnamefont {W.}~\bibnamefont {Idler}},\
  }\href@noop {} {\enquote {\bibinfo {title} {Transmission of waveforms
  determined by 7 eigenvalues with psk-modulated spectral amplitudes},}\ }
  (\bibinfo {year} {2016}),\ \bibinfo {note}
  {{arXiv}:1605.08069[physics.optics]}\BibitemShut {NoStop}%
\bibitem [{\citenamefont {Cox}\ and\ \citenamefont {Matthews}(2002)}]{CM2002}%
  \BibitemOpen
  \bibfield  {author} {\bibinfo {author} {\bibfnamefont {S.~M.}\ \bibnamefont
  {Cox}}\ and\ \bibinfo {author} {\bibfnamefont {P.~C.}\ \bibnamefont
  {Matthews}},\ }\href {\doibase 10.1006/jcph.2002.6995} {\bibfield  {journal}
  {\bibinfo  {journal} {Journal of Computational Physics}\ }\textbf {\bibinfo
  {volume} {176}},\ \bibinfo {pages} {430} (\bibinfo {year}
  {2002})}\BibitemShut {NoStop}%
\bibitem [{\citenamefont {Gautschi}(2012)}]{Gautschi2012}%
  \BibitemOpen
  \bibfield  {author} {\bibinfo {author} {\bibfnamefont {W.}~\bibnamefont
  {Gautschi}},\ }\href {\doibase 10.1007/978-0-8176-8259-0} {\emph {\bibinfo
  {title} {Numerical Analysis}}}\ (\bibinfo  {publisher} {Birkh{\"a}user
  Boston},\ \bibinfo {year} {2012})\BibitemShut {NoStop}%
\bibitem [{\citenamefont {Born}\ and\ \citenamefont {Wolf}(1999)}]{BW1999}%
  \BibitemOpen
  \bibfield  {author} {\bibinfo {author} {\bibfnamefont {M.}~\bibnamefont
  {Born}}\ and\ \bibinfo {author} {\bibfnamefont {E.}~\bibnamefont {Wolf}},\
  }\href@noop {} {\emph {\bibinfo {title} {Principles of Optics:
  Electromagnetic Theory of Propagation, Interference and Diffraction of
  Light}}},\ \bibinfo {edition} {seventh}\ ed.\ (\bibinfo  {publisher}
  {Cambridge University Press},\ \bibinfo {year} {1999})\BibitemShut {NoStop}%
\bibitem [{\citenamefont {Henrici}(1993)}]{Henrici1993}%
  \BibitemOpen
  \bibfield  {author} {\bibinfo {author} {\bibfnamefont {P.}~\bibnamefont
  {Henrici}},\ }\href@noop {} {\emph {\bibinfo {title} {Applied and
  Computational Complex Analysis, Volume 3: Discrete Fourier Analysis, Cauchy
  Integrals, Construction of Conformal Maps, Univalent Functions}}},\ Applied
  and Computational Complex Analysis\ (\bibinfo  {publisher} {John Wiley \&
  Sons, Inc.},\ \bibinfo {year} {1993})\BibitemShut {NoStop}%
\bibitem [{\citenamefont {Bruckstein}\ and\ \citenamefont
  {Kailath}(1987)}]{BK1987}%
  \BibitemOpen
  \bibfield  {author} {\bibinfo {author} {\bibfnamefont {A.~M.}\ \bibnamefont
  {Bruckstein}}\ and\ \bibinfo {author} {\bibfnamefont {T.}~\bibnamefont
  {Kailath}},\ }\href {\doibase 10.1137/1029075} {\bibfield  {journal}
  {\bibinfo  {journal} {SIAM Review}\ }\textbf {\bibinfo {volume} {29}},\
  \bibinfo {pages} {359} (\bibinfo {year} {1987})}\BibitemShut {NoStop}%
\bibitem [{\citenamefont {Ablowitz}\ \emph {et~al.}(2004)\citenamefont
  {Ablowitz}, \citenamefont {Prinari},\ and\ \citenamefont
  {Trubatch}}]{APT2004}%
  \BibitemOpen
  \bibfield  {author} {\bibinfo {author} {\bibfnamefont {M.~J.}\ \bibnamefont
  {Ablowitz}}, \bibinfo {author} {\bibfnamefont {B.}~\bibnamefont {Prinari}}, \
  and\ \bibinfo {author} {\bibfnamefont {A.~D.}\ \bibnamefont {Trubatch}},\
  }\href@noop {} {\emph {\bibinfo {title} {Discrete and Continuous Nonlinear
  Schr{\"o}dinger Systems}}},\ \bibinfo {series} {London Mathematical Society
  Lecture Note Series}, Vol.\ \bibinfo {volume} {302}\ (\bibinfo  {publisher}
  {Cambridge University Press},\ \bibinfo {year} {2004})\BibitemShut {NoStop}%
\bibitem [{\citenamefont {Wahls}\ and\ \citenamefont
  {Poor}(2015{\natexlab{a}})}]{WP2013d}%
  \BibitemOpen
  \bibfield  {author} {\bibinfo {author} {\bibfnamefont {S.}~\bibnamefont
  {Wahls}}\ and\ \bibinfo {author} {\bibfnamefont {H.~V.}\ \bibnamefont
  {Poor}},\ }\href {\doibase 10.1109/TIT.2015.2485944} {\bibfield  {journal}
  {\bibinfo  {journal} {IEEE Transactions on Information Theory}\ }\textbf
  {\bibinfo {volume} {61}},\ \bibinfo {pages} {6957} (\bibinfo {year}
  {2015}{\natexlab{a}})}\BibitemShut {NoStop}%
\bibitem [{\citenamefont {Wahls}\ and\ \citenamefont
  {Poor}(2015{\natexlab{b}})}]{WP2015}%
  \BibitemOpen
  \bibfield  {author} {\bibinfo {author} {\bibfnamefont {S.}~\bibnamefont
  {Wahls}}\ and\ \bibinfo {author} {\bibfnamefont {H.~V.}\ \bibnamefont
  {Poor}},\ }in\ \href {\doibase 10.1109/ISIT.2015.7282741} {\emph {\bibinfo
  {booktitle} {Information Theory (ISIT), 2015 IEEE International Symposium
  on}}}\ (\bibinfo {year} {2015})\ pp.\ \bibinfo {pages}
  {1676--1680}\BibitemShut {NoStop}%
\bibitem [{\citenamefont {Lubich}(1988{\natexlab{a}})}]{Lubich1988I}%
  \BibitemOpen
  \bibfield  {author} {\bibinfo {author} {\bibfnamefont {C.}~\bibnamefont
  {Lubich}},\ }\href {\doibase 10.1007/BF01398686} {\bibfield  {journal}
  {\bibinfo  {journal} {Numerische Mathematik}\ }\textbf {\bibinfo {volume}
  {52}},\ \bibinfo {pages} {129} (\bibinfo {year}
  {1988}{\natexlab{a}})}\BibitemShut {NoStop}%
\bibitem [{\citenamefont {Lubich}(1988{\natexlab{b}})}]{Lubich1988II}%
  \BibitemOpen
  \bibfield  {author} {\bibinfo {author} {\bibfnamefont {C.}~\bibnamefont
  {Lubich}},\ }\href {\doibase 10.1007/BF01462237} {\bibfield  {journal}
  {\bibinfo  {journal} {Numerische Mathematik}\ }\textbf {\bibinfo {volume}
  {52}},\ \bibinfo {pages} {413} (\bibinfo {year}
  {1988}{\natexlab{b}})}\BibitemShut {NoStop}%
\bibitem [{\citenamefont {Lubich}(1994)}]{Lubich1994}%
  \BibitemOpen
  \bibfield  {author} {\bibinfo {author} {\bibfnamefont {C.}~\bibnamefont
  {Lubich}},\ }\href {\doibase 10.1007/s002110050033} {\bibfield  {journal}
  {\bibinfo  {journal} {Numerische Mathematik}\ }\textbf {\bibinfo {volume}
  {67}},\ \bibinfo {pages} {365} (\bibinfo {year} {1994})}\BibitemShut
  {NoStop}%
\bibitem [{\citenamefont {Neugebauer}\ and\ \citenamefont
  {Meinel}(1984)}]{NM1984}%
  \BibitemOpen
  \bibfield  {author} {\bibinfo {author} {\bibfnamefont {G.}~\bibnamefont
  {Neugebauer}}\ and\ \bibinfo {author} {\bibfnamefont {R.}~\bibnamefont
  {Meinel}},\ }\href {\doibase 10.1016/0375-9601(84)90827-2} {\bibfield
  {journal} {\bibinfo  {journal} {Physics Letters A}\ }\textbf {\bibinfo
  {volume} {100}},\ \bibinfo {pages} {467} (\bibinfo {year}
  {1984})}\BibitemShut {NoStop}%
\bibitem [{\citenamefont {Vaibhav}\ and\ \citenamefont
  {Wahls}(2016)}]{VW2016OFC}%
  \BibitemOpen
  \bibfield  {author} {\bibinfo {author} {\bibfnamefont {V.}~\bibnamefont
  {Vaibhav}}\ and\ \bibinfo {author} {\bibfnamefont {W.}~\bibnamefont
  {Wahls}},\ }in\ \href {\doibase 10.1364/OFC.2016.W2A.34} {\emph {\bibinfo
  {booktitle} {Optical Fiber Communication Conference}}}\ (\bibinfo
  {publisher} {Optical Society of America},\ \bibinfo {year} {2016})\ p.\
  \bibinfo {pages} {W2A.34}\BibitemShut {NoStop}%
\bibitem [{\citenamefont {Lamb}(1980)}]{Lamb1980}%
  \BibitemOpen
  \bibfield  {author} {\bibinfo {author} {\bibfnamefont {G.~L.}\ \bibnamefont
  {Lamb}},\ }\href@noop {} {\emph {\bibinfo {title} {Elements of soliton
  theory}}},\ Pure and applied mathematics\ (\bibinfo  {publisher} {John Wiley
  \& Sons, Inc.},\ \bibinfo {year} {1980})\BibitemShut {NoStop}%
\bibitem [{\citenamefont {Steudel}(2002)}]{S2002}%
  \BibitemOpen
  \bibfield  {author} {\bibinfo {author} {\bibfnamefont {H.}~\bibnamefont
  {Steudel}},\ }\href {\doibase 10.1088/0266-5611/24/2/025015} {\bibfield
  {journal} {\bibinfo  {journal} {Contemporary Mathematics}\ }\textbf {\bibinfo
  {volume} {301}},\ \bibinfo {pages} {331} (\bibinfo {year}
  {2002})}\BibitemShut {NoStop}%
\bibitem [{\citenamefont {Steudel}\ and\ \citenamefont {Kaup}(2008)}]{SK2008}%
  \BibitemOpen
  \bibfield  {author} {\bibinfo {author} {\bibfnamefont {H.}~\bibnamefont
  {Steudel}}\ and\ \bibinfo {author} {\bibfnamefont {D.~J.}\ \bibnamefont
  {Kaup}},\ }\href {\doibase 10.1088/0266-5611/24/2/025015} {\bibfield
  {journal} {\bibinfo  {journal} {Inverse Problems}\ }\textbf {\bibinfo
  {volume} {24}},\ \bibinfo {pages} {025015} (\bibinfo {year}
  {2008})}\BibitemShut {NoStop}%
\bibitem [{\citenamefont {Wahls}\ and\ \citenamefont {Poor}(2013)}]{WP2013b}%
  \BibitemOpen
  \bibfield  {author} {\bibinfo {author} {\bibfnamefont {S.}~\bibnamefont
  {Wahls}}\ and\ \bibinfo {author} {\bibfnamefont {H.~V.}\ \bibnamefont
  {Poor}},\ }in\ \href@noop {} {\emph {\bibinfo {booktitle} {Proc. IEEE Int.
  Conf. Acoust. Speech Signal Process. (ICASSP)}}}\ (\bibinfo {address}
  {Vancouver, Canada},\ \bibinfo {year} {2013})\BibitemShut {NoStop}%
\bibitem [{\citenamefont {Aref}(2016)}]{V2016}%
  \BibitemOpen
  \bibfield  {author} {\bibinfo {author} {\bibfnamefont {V.}~\bibnamefont
  {Aref}},\ }\href@noop {} {\enquote {\bibinfo {title} {Control and detection
  of discrete spectral amplitudes in nonlinear {F}ourier spectrum},}\ }
  (\bibinfo {year} {2016}),\ \bibinfo {note}
  {{arXiv}:1605.06328[math.NA]}\BibitemShut {NoStop}%
\bibitem [{\citenamefont {Henrici}(1964)}]{Henrici1964}%
  \BibitemOpen
  \bibfield  {author} {\bibinfo {author} {\bibfnamefont {P.}~\bibnamefont
  {Henrici}},\ }\href@noop {} {\emph {\bibinfo {title} {Elements of numerical
  analysis}}},\ Wiley international edition\ (\bibinfo  {publisher} {John Wiley
  \& Sons, Inc.},\ \bibinfo {year} {1964})\BibitemShut {NoStop}%
\bibitem [{\citenamefont {McClary}(1983)}]{McClary1983}%
  \BibitemOpen
  \bibfield  {author} {\bibinfo {author} {\bibfnamefont {W.~K.}\ \bibnamefont
  {McClary}},\ }\href {\doibase 10.1190/1.1441417} {\bibfield  {journal}
  {\bibinfo  {journal} {Geophysics}\ }\textbf {\bibinfo {volume} {48}},\
  \bibinfo {pages} {1371} (\bibinfo {year} {1983})}\BibitemShut {NoStop}%
\bibitem [{\citenamefont {Tustin}(1947)}]{Tustin1947}%
  \BibitemOpen
  \bibfield  {author} {\bibinfo {author} {\bibfnamefont {A.}~\bibnamefont
  {Tustin}},\ }\href@noop {} {\bibfield  {journal} {\bibinfo  {journal}
  {Journal of the Institution of Electrical Engineers}\ }\textbf {\bibinfo
  {volume} {94}},\ \bibinfo {pages} {130} (\bibinfo {year} {1947})}\BibitemShut
  {NoStop}%
\bibitem [{\citenamefont {Magnus}(1954)}]{Magnus1954}%
  \BibitemOpen
  \bibfield  {author} {\bibinfo {author} {\bibfnamefont {W.}~\bibnamefont
  {Magnus}},\ }\href {\doibase 10.1002/cpa.3160070404} {\bibfield  {journal}
  {\bibinfo  {journal} {Communications on Pure and Applied Mathematics}\
  }\textbf {\bibinfo {volume} {7}},\ \bibinfo {pages} {649} (\bibinfo {year}
  {1954})}\BibitemShut {NoStop}%
\bibitem [{\citenamefont {Iserles}\ and\ \citenamefont
  {N{\o}rsett}(1999)}]{IN1999}%
  \BibitemOpen
  \bibfield  {author} {\bibinfo {author} {\bibfnamefont {A.}~\bibnamefont
  {Iserles}}\ and\ \bibinfo {author} {\bibfnamefont {S.}~\bibnamefont
  {N{\o}rsett}},\ }\href {\doibase 10.1098/rsta.1999.0362} {\bibfield
  {journal} {\bibinfo  {journal} {Philosophical Transactions of the Royal
  Society of London A: Mathematical, Physical and Engineering Sciences}\
  }\textbf {\bibinfo {volume} {357}},\ \bibinfo {pages} {983} (\bibinfo {year}
  {1999})}\BibitemShut {NoStop}%
\bibitem [{\citenamefont {Hochbruck}\ and\ \citenamefont
  {Lubich}(2003)}]{HL2003}%
  \BibitemOpen
  \bibfield  {author} {\bibinfo {author} {\bibfnamefont {M.}~\bibnamefont
  {Hochbruck}}\ and\ \bibinfo {author} {\bibfnamefont {C.}~\bibnamefont
  {Lubich}},\ }\href {\doibase 10.1137/S0036142902403875} {\bibfield  {journal}
  {\bibinfo  {journal} {SIAM Journal on Numerical Analysis}\ }\textbf {\bibinfo
  {volume} {41}},\ \bibinfo {pages} {945} (\bibinfo {year} {2003})}\BibitemShut
  {NoStop}%
\bibitem [{\citenamefont {Hairer}\ \emph {et~al.}(2006)\citenamefont {Hairer},
  \citenamefont {Lubich},\ and\ \citenamefont {Wanner}}]{HLW2006}%
  \BibitemOpen
  \bibfield  {author} {\bibinfo {author} {\bibfnamefont {E.}~\bibnamefont
  {Hairer}}, \bibinfo {author} {\bibfnamefont {C.}~\bibnamefont {Lubich}}, \
  and\ \bibinfo {author} {\bibfnamefont {G.}~\bibnamefont {Wanner}},\ }\href
  {\doibase 10.1007/3-540-30666-8} {\emph {\bibinfo {title} {Geometric
  Numerical Integration: Structure-Preserving Algorithms for Ordinary
  Differential Equations}}},\ \bibinfo {edition} {2nd}\ ed.,\ Springer Series
  in Computational Mathematics\ (\bibinfo  {publisher} {Springer-Verlag Berlin
  Heidelberg},\ \bibinfo {year} {2006})\BibitemShut {NoStop}%
\bibitem [{\citenamefont {Boffetta}\ and\ \citenamefont
  {Osborne}(1992)}]{BO1992}%
  \BibitemOpen
  \bibfield  {author} {\bibinfo {author} {\bibfnamefont {G.}~\bibnamefont
  {Boffetta}}\ and\ \bibinfo {author} {\bibfnamefont {A.~R.}\ \bibnamefont
  {Osborne}},\ }\href {\doibase 10.1016/0021-9991(92)90370-E} {\bibfield
  {journal} {\bibinfo  {journal} {Journal of Computational Physics}\ }\textbf
  {\bibinfo {volume} {102}},\ \bibinfo {pages} {252} (\bibinfo {year}
  {1992})}\BibitemShut {NoStop}%
\bibitem [{\citenamefont {Burtsev}\ \emph {et~al.}(1998)\citenamefont
  {Burtsev}, \citenamefont {Camassa},\ and\ \citenamefont
  {Timofeyev}}]{BCT1998}%
  \BibitemOpen
  \bibfield  {author} {\bibinfo {author} {\bibfnamefont {S.}~\bibnamefont
  {Burtsev}}, \bibinfo {author} {\bibfnamefont {R.}~\bibnamefont {Camassa}}, \
  and\ \bibinfo {author} {\bibfnamefont {I.}~\bibnamefont {Timofeyev}},\ }\href
  {\doibase 10.1006/jcph.1998.6087} {\bibfield  {journal} {\bibinfo  {journal}
  {Journal of Computational Physics}\ }\textbf {\bibinfo {volume} {147}},\
  \bibinfo {pages} {166} (\bibinfo {year} {1998})}\BibitemShut {NoStop}%
\bibitem [{\citenamefont {Strang}(1968)}]{Strang1968}%
  \BibitemOpen
  \bibfield  {author} {\bibinfo {author} {\bibfnamefont {G.}~\bibnamefont
  {Strang}},\ }\href {\doibase 10.1137/0705041} {\bibfield  {journal} {\bibinfo
   {journal} {SIAM Journal on Numerical Analysis}\ }\textbf {\bibinfo {volume}
  {5}},\ \bibinfo {pages} {506} (\bibinfo {year} {1968})}\BibitemShut {NoStop}%
\bibitem [{\citenamefont {Bruckstein}\ \emph {et~al.}(1985)\citenamefont
  {Bruckstein}, \citenamefont {Levy},\ and\ \citenamefont {Kailath}}]{BLK1985}%
  \BibitemOpen
  \bibfield  {author} {\bibinfo {author} {\bibfnamefont {A.~M.}\ \bibnamefont
  {Bruckstein}}, \bibinfo {author} {\bibfnamefont {B.~C.}\ \bibnamefont
  {Levy}}, \ and\ \bibinfo {author} {\bibfnamefont {T.}~\bibnamefont
  {Kailath}},\ }\href {\doibase 10.1137/0145017} {\bibfield  {journal}
  {\bibinfo  {journal} {SIAM Journal on Applied Mathematics}\ }\textbf
  {\bibinfo {volume} {45}},\ \bibinfo {pages} {312} (\bibinfo {year}
  {1985})}\BibitemShut {NoStop}%
\bibitem [{\citenamefont {Bruckstein}\ \emph {et~al.}(1986)\citenamefont
  {Bruckstein}, \citenamefont {Koltracht},\ and\ \citenamefont
  {Kailath}}]{BKK1986}%
  \BibitemOpen
  \bibfield  {author} {\bibinfo {author} {\bibfnamefont {A.~M.}\ \bibnamefont
  {Bruckstein}}, \bibinfo {author} {\bibfnamefont {I.}~\bibnamefont
  {Koltracht}}, \ and\ \bibinfo {author} {\bibfnamefont {T.}~\bibnamefont
  {Kailath}},\ }\href {\doibase 10.1137/0907088} {\bibfield  {journal}
  {\bibinfo  {journal} {SIAM Journal on Scientific and Statistical Computing}\
  }\textbf {\bibinfo {volume} {7}},\ \bibinfo {pages} {1331} (\bibinfo {year}
  {1986})}\BibitemShut {NoStop}%
\bibitem [{\citenamefont {Jones}(2001)}]{Jones2001}%
  \BibitemOpen
  \bibfield  {author} {\bibinfo {author} {\bibfnamefont {F.}~\bibnamefont
  {Jones}},\ }\href@noop {} {\emph {\bibinfo {title} {Lebesgue Integration on
  Euclidean Space}}},\ Jones and Bartlett books in mathematics\ (\bibinfo
  {publisher} {Jones and Bartlett},\ \bibinfo {year} {2001})\BibitemShut
  {NoStop}%
\bibitem [{\citenamefont {Novikov}\ \emph {et~al.}(1984)\citenamefont
  {Novikov}, \citenamefont {Manakov}, \citenamefont {Pitaevskii},\ and\
  \citenamefont {Zakharov}}]{NMPZ1984}%
  \BibitemOpen
  \bibfield  {author} {\bibinfo {author} {\bibfnamefont {S.}~\bibnamefont
  {Novikov}}, \bibinfo {author} {\bibfnamefont {S.~V.}\ \bibnamefont
  {Manakov}}, \bibinfo {author} {\bibfnamefont {L.~P.}\ \bibnamefont
  {Pitaevskii}}, \ and\ \bibinfo {author} {\bibfnamefont {V.~E.}\ \bibnamefont
  {Zakharov}},\ }\href@noop {} {\emph {\bibinfo {title} {Theory of Solitons:
  The Inverse Scattering Method}}},\ Contemporary Soviet Mathematics\ (\bibinfo
   {publisher} {Consultant Bureau, New York and London},\ \bibinfo {year}
  {1984})\BibitemShut {NoStop}%
\bibitem [{\citenamefont {Gripenberg}\ \emph {et~al.}(1990)\citenamefont
  {Gripenberg}, \citenamefont {Londen},\ and\ \citenamefont
  {Staffans}}]{GLS1990}%
  \BibitemOpen
  \bibfield  {author} {\bibinfo {author} {\bibfnamefont {G.}~\bibnamefont
  {Gripenberg}}, \bibinfo {author} {\bibfnamefont {S.~O.}\ \bibnamefont
  {Londen}}, \ and\ \bibinfo {author} {\bibfnamefont {O.}~\bibnamefont
  {Staffans}},\ }\href@noop {} {\emph {\bibinfo {title} {Volterra Integral and
  Functional Equations}}},\ \bibinfo {series} {Encyclopedia of Mathematics and
  Its Applications}, Vol.~\bibinfo {volume} {34}\ (\bibinfo  {publisher}
  {Cambridge University Press},\ \bibinfo {year} {1990})\BibitemShut {NoStop}%
\bibitem [{\citenamefont {Lancaster}\ and\ \citenamefont
  {Tismenetsky}(1985)}]{LT1985}%
  \BibitemOpen
  \bibfield  {author} {\bibinfo {author} {\bibfnamefont {P.}~\bibnamefont
  {Lancaster}}\ and\ \bibinfo {author} {\bibfnamefont {M.}~\bibnamefont
  {Tismenetsky}},\ }\href@noop {} {\emph {\bibinfo {title} {The Theory of
  Matrices: With Applications}}},\ \bibinfo {edition} {2nd}\ ed.,\ Computer
  Science and Scientific Computing Series\ (\bibinfo  {publisher} {Academic
  Press},\ \bibinfo {year} {1985})\BibitemShut {NoStop}%
\bibitem [{\citenamefont {Satsuma}\ and\ \citenamefont
  {Yajima}(1974)}]{SY1974}%
  \BibitemOpen
  \bibfield  {author} {\bibinfo {author} {\bibfnamefont {J.}~\bibnamefont
  {Satsuma}}\ and\ \bibinfo {author} {\bibfnamefont {N.}~\bibnamefont
  {Yajima}},\ }\href {\doibase 10.1143/PTPS.55.284} {\bibfield  {journal}
  {\bibinfo  {journal} {Progress of Theoretical Physics Supplement}\ }\textbf
  {\bibinfo {volume} {55}},\ \bibinfo {pages} {284} (\bibinfo {year}
  {1974})}\BibitemShut {NoStop}%
\bibitem [{\citenamefont {Baker}\ and\ \citenamefont
  {Graves-Morris}(1981)}]{BGM1981}%
  \BibitemOpen
  \bibfield  {author} {\bibinfo {author} {\bibfnamefont {G.~A.}\ \bibnamefont
  {Baker}}\ and\ \bibinfo {author} {\bibfnamefont {P.}~\bibnamefont
  {Graves-Morris}},\ }\enquote {\bibinfo {title} {{P}ad\'{e} {A}pproximants},}\
  \ (\bibinfo  {publisher} {Addison-Wesley publishing Company, Inc.},\ \bibinfo
  {address} {Massachusetts},\ \bibinfo {year} {1981})\ Chap.~\bibinfo {chapter}
  {2}, pp.\ \bibinfo {pages} {43--45}\BibitemShut {NoStop}%
\bibitem [{\citenamefont {Olver}\ \emph {et~al.}(2010)\citenamefont {Olver},
  \citenamefont {Lozier}, \citenamefont {Boisvert},\ and\ \citenamefont
  {Clark}}]{Olver:2010:NHMF}%
  \BibitemOpen
  \bibinfo {editor} {\bibfnamefont {F.~W.~J.}\ \bibnamefont {Olver}}, \bibinfo
  {editor} {\bibfnamefont {D.~W.}\ \bibnamefont {Lozier}}, \bibinfo {editor}
  {\bibfnamefont {R.~F.}\ \bibnamefont {Boisvert}}, \ and\ \bibinfo {editor}
  {\bibfnamefont {C.~W.}\ \bibnamefont {Clark}},\ eds.,\ \href@noop {} {\emph
  {\bibinfo {title} {{NIST Handbook of Mathematical Functions}}}}\ (\bibinfo
  {publisher} {Cambridge University Press},\ \bibinfo {address} {New York,
  NY},\ \bibinfo {year} {2010})\BibitemShut {NoStop}%
\end{thebibliography}
\providecommand{\noopsort}[1]{}\providecommand{\singleletter}[1]{#1}%

\appendix
\section{Lubich coefficients for rational functions with simple poles}\label{app:lubich-rational}
Consider the simplest case of a rational function with a simple pole 
$E(\zeta)=(\zeta-\zeta_0)^{-1}$ where $\Im\zeta_0<0$. It satisfies the kind of
growth estimated stated in~\eqref{eq:growth-a-b}. The inverse Fourier-Laplace transform 
is given by $e(\tau)=-ie^{-i\zeta_0\tau}$. The 
Lubich coefficients corresponding to the trapezoidal rule is defined through
\begin{equation}\label{eq:lubich-F}
\begin{split}
&E\left(\frac{i\delta(z^2)}{2h}\right)=\frac{1}{\left[\frac{i\delta(z^2)}{2h}-\zeta_0\right]}\\
&=-ih\frac{(1+z^2)}{(1+i\zeta_0h)}\sum_{k=0}^{\infty}\left(\frac{1-i\zeta_0h}{1+i\zeta_0h}\right)^kz^{2k},
\end{split}
\end{equation}
where we note that $|{1-i\zeta_0h}|/|{1+i\zeta_0h}|<1$ on account of $\Im\zeta_0<0$. The 
Lubich coefficient $e_k$ is defined as the coefficient
of $z^{2k}$ in the RHS of~\eqref{eq:lubich-F} which can be worked out
explicitly: $e_0 = {-ih}/(1+i\zeta_0h)$ and, for $k>0$,
\begin{equation}
e_k
=\frac{-2ih}{[1+(\zeta_0h)^2]}\left(\frac{1-i\zeta_0h}{1+i\zeta_0h}\right)^k.
\end{equation}
Note that when $\Re\zeta_0=0$, then we may restrict $h\in(0,\bar{h}]$ so that 
$1+\eta_0h\neq0$ where $\eta_0=\Im{\zeta}$. In the following, we wish to study the error involved in
replacing $e_k$ with $-2ihe^{-2i\zeta_0hk}$. For $k>0$, this difference is given by
\begin{multline}\label{eq:lubich-estimate}
|e_k+2ihe^{-2i\zeta_0hk}|=\frac{2h}{|1+(\zeta_0h)^2|}\times\\
\left|\left(\frac{1-i\zeta_0h}{1+i\zeta_0h}\right)^k-e^{-2i\zeta_0hk}[1+(\zeta_0h)^2]\right|.
\end{multline}
Using the $[1/1]$-Pad\'e approximant~\cite{BGM1981}, we have
\begin{equation}
e^{-2i\zeta_0h} =\left(\frac{1-i\zeta_0h}{1+i\zeta_0h}\right)+\bigO{h^3}.
\end{equation}
Next, let us show that, for $h\in(0,\bar{h}]$, there exists a positive integer
$n>1$ dependent only on $\bar{h}$ such that
\begin{equation}
\left|\frac{1-i\zeta_0h}{1+i\zeta_0h}\right|\leq e^{2\eta_0h/n}.
\end{equation}
Recalling $\eta_0=\Im\zeta_0$, let $n$ be chosen such that 
\begin{equation}
n>\sup_{h\in(0,\bar{h}]}-2\eta_0h/\log\left|\frac{1+i\zeta_0h}{1-i\zeta_0h}\right|.
\end{equation}
From the inequality~\cite[Chap.~4]{Olver:2010:NHMF}
\[
\frac{-2\eta_0h}{\log\left|\frac{1+i\zeta_0h}{1-i\zeta_0h}\right|}\geq
\frac{-4\eta_0h}{\left|\frac{1+i\zeta_0h}{1-i\zeta_0h}\right|^2-1}=|1-i\zeta_0h|^2,
\]
it follows that $n>1$. Also, observing
\[
\frac{-2\eta_0h}{\log\left|\frac{1+i\zeta_0h}{1-i\zeta_0h}\right|}\leq
\frac{-4\eta_0h}{1-\left|\frac{1-i\zeta_0h}{1+i\zeta_0h}\right|^2}\leq|1+i\zeta_0h|^2,
\]
it suffices to choose $n$ to be the smallest integer greater than
$(1+|\zeta_0|\bar{h})^2$. Now, using standard inequalities for 
exponential function~\cite[Chap.~4]{Olver:2010:NHMF}, we have
\begin{widetext}
\begin{align*}
\left|\left(\frac{1-i\zeta_0h}{1+i\zeta_0h}\right)^k-e^{-2ik\zeta_0h}\right|
&\leq\left|\left(\frac{1-i\zeta_0h}{1+i\zeta_0h}\right)-e^{-2i\zeta_0h}\right|
\sum_{j=0}^{\infty}\exp\left[2\eta_0h\left(\frac{k-1}{n}+j\frac{n-1}{n}\right)\right]
\leq {C' h^2}e^{2\eta_0kh/n}
\end{align*}
where 
\end{widetext}
\[
C'=e^{-2\eta_0\bar{h}/n}\frac{1-2\eta_0\bar{h}}{2\eta_0(1-1/n)}\times\text{const.}.
\]
Finally, using the last estimate and 
from~\eqref{eq:lubich-estimate}, we conclude
\begin{align*}
|e_k+2ihe^{-2i\zeta_0hk}|\leq
\frac{C'h^3}{|1-i\zeta_0h|}e^{2\eta_0kh/n},
\end{align*}
where $C'>0$ is independent of $h$ and $k$. If in addition $(-\eta_0\bar{h})<1$, then 
for $h\in(0,\bar{h}]$ and $k>0$, we may write 
\begin{equation}
|e_k+2ihe^{-2i\zeta_0hk}|\leq C e^{2\eta_0kh/n}h^3,
\end{equation}
where $C=C'/(1+\eta_0\bar{h})$ is independent of $h$ and $k$. Using this
estimate, let us now show that one can make an informed choice of the parameter
$N_{\text{th}}\in\field{Z}_+$ introduced in sections~\ref{sec:FL-transform-lubich} 
and~\ref{sec:DT-pure-soliton} in connection with the partial-fraction variant of
FDT. To this end, we start with 
\[
e^{2\eta_0N_{\text{th}}h/n}=\epsilon,
\]
where $\epsilon$ is a positive number less than unity. Putting $N_{\text{th}}=N/m$
and using $h\sim 2L/N$ where $2L=L_2-L_1$ and $N\in\field{Z}_+$, we have
\begin{equation}
m\sim\frac{4\eta_0L}{n\log\epsilon}.
\end{equation}
Setting $\epsilon=e^{-1}$ and for $|\zeta_0|\bar{h}<1$ one can set $n=4$ so
that $m\sim |\eta_0|L$. In case of multi-solitons, we would
like to tune the parameter $N_{\text{th}}$ with respect to the eigenvalue with
the smallest imaginary part. Here $\zeta_0$ must be replaced by the complex 
conjugate of the eigenvalue with smallest imaginary part. For example, if the smallest 
of all the imaginary parts of eigenvalues is unity and $L=10$, we should choose 
$m=10$ (or, $m=8$ so that $N_{\text{th}}$ is a power of $2$ when $N$ is a
power of $2$).
\end{document}